\documentclass[economy,bind]{dukedissertation}
\pdfoutput=1

\usepackage{amsmath}
  \allowdisplaybreaks[1]
\usepackage{amssymb}
\usepackage{amsthm}
\usepackage{color}
\usepackage{subfigure}
\usepackage{enumitem}
\usepackage{graphicx}
\usepackage{mathabx}
\usepackage{mathtools}
\usepackage{multirow}
\usepackage{pgfplots}
  \pgfplotsset{compat=newest}
\usepackage{setspace}
\usepackage{siunitx}

\usepackage{cleveref}

\usepackage[backend=bibtex8,hyperref=auto,style=numeric]{biblatex}
\newtheorem{theorem}{Theorem}
\newtheorem{lemma}[theorem]{Lemma}

\theoremstyle{definition}
\newtheorem{convention}{Convention}
\newtheorem{definition}{Definition}

\newtheorem{mathconjecture}{Math Conjecture}
\AfterEndEnvironment{mathconjecture}{\noindent\ignorespaces}
\newtheorem{physicsconjecture}{Physics Conjecture}
\AfterEndEnvironment{physicsconjecture}{\noindent\ignorespaces}
\newtheorem{question}{Question}
\AfterEndEnvironment{question}{\noindent\ignorespaces}

\DeclarePairedDelimiter\abs{\lvert}{\rvert} 
\DeclarePairedDelimiter\norm{\lVert}{\rVert} 
\newcommand\tensor{\otimes} 
\newcommand\cc{\widebar} 
\newcommand\R{\mathbb{R}} 
\newcommand\C{\mathbb{C}} 
\renewcommand\d{\partial} 
\DeclareMathOperator\Ric{Ric} 
\DeclareMathOperator\divergence{div} 
\DeclareMathOperator\trace{tr} 
\DeclareSIUnit\solarmass{\ensuremath{M_\Sun}} 
\DeclareSIUnit\year{yr} 
\DeclareSIUnit\lightyear{ly} 

\newcommand\Vtil{\tilde{V}}
\newcommand\RDM{R_{\text{DM}}}
\newcommand\laplacian\Delta
\DeclareMathOperator\Ai{Ai}
\DeclareMathOperator\Bi{Bi}
\newcommand\SF{\mathcal{F}}
\newcommand\hodge{\mathord{\star}}
\DeclarePairedDelimiterX\ip[2]{\langle}{\rangle}{#1,#2} 

\newcommand\omegatrue{\omega_{\text{true}}}

\graphicspath{{./Pictures/}}

\author{Andrew S. Goetz}
\title{The Einstein-Klein-Gordon Equations, Wave Dark Matter, and the Tully-Fisher Relation}
\supervisor{Hubert Bray}
\department{Mathematics}

\copyrighttext{ All rights reserved except the rights granted by the\\
   \href{http://creativecommons.org/licenses/by-nc/3.0/us/}
        {Creative Commons Attribution-Noncommercial License}
}

\member{Paul Aspinwall}
\member{Mark Stern}
\member{Lenny Ng}

\begin{document}

\maketitle

\abstract

We examine the Einstein equation coupled to the Klein-Gordon equation for a complex-valued scalar field. These two equations together are known as the Einstein-Klein-Gordon system. In the low-field, non-relativistic limit, the Einstein-Klein-Gordon system reduces to the Poisson-Schr\"odinger system. We describe the simplest solutions of these systems in spherical symmetry, the spherically symmetric static states, and some scaling properties they obey. We also describe some approximate analytic solutions for these states.

The EKG system underlies a theory of wave dark matter, also known as scalar field dark matter (SFDM), boson star dark matter, and Bose-Einstein condensate (BEC) dark matter. We discuss a possible connection between the theory of wave dark matter and the baryonic Tully-Fisher relation, which is a scaling relation observed to hold for disk galaxies in the universe across many decades in mass. We show how fixing boundary conditions at the edge of the spherically symmetric static states implies Tully-Fisher-like relations for the states. We also catalog other ``scaling conditions'' one can impose on the static states and show that they do not lead to Tully-Fisher-like relations---barring one exception which is already known and which has nothing to do with the specifics of wave dark matter.

\dedication{To my parents, Stewart and Carolyn.}

\tableofcontents 
\listoftables	
\listoffigures	
\abbreviations


\section*{Symbols}

\begin{symbollist}
\item[$\C$] Complex number field
\item[$G$] Einstein curvature tensor
\item[$\Lambda$] Cosmological constant
\item[$\R$] Real number field
\item[$R$] Scalar curvature
\item[$\Ric$] Ricci curvature tensor
\item[$T$] Energy-momentum tensor
\end{symbollist}

\section*{Abbreviations}

\begin{symbollist}
\item[BEC] Bose-Einstein condensate
\item[BTFR] Baryonic Tully-Fisher relation
\item[EKG] Einstein-Klein-Gordon
\item[PS] Poisson-Schr\"odinger
\item[SFDM] Scalar field dark matter
\item[TFR] Tully-Fisher relation
\end{symbollist}

\acknowledgements

I thank my advisor, Hubert Bray, for introducing me to general relativity, the mystery of large-scale unexplained curvature, and the Tully-Fisher relation. I also thank him for his guidance in matters mathematical, career-related, and personal, and for his patience with me as I struggled to do research.

I thank the professors in the Duke math department, particularly the other members of my dissertation committee, for their teaching and mentorship. I thank my graduate student friends, in particular Hangjun Xu, Henri Roesch, Tatsunari Watanabe, Ben Gaines, Kevin Kordek, Ioannis Sgouralis, Chris O'Neill, and Anil Venkatesh.

I thank Duke University for supporting me with a stipend for teaching and research, with a James B. Duke fellowship from 2009--2012, and with a graduate student summer research fellowship in 2012.

I thank the other programs and institutions which have supported me during my time at Duke: the University of Tennessee, which funded my trips to the Barrett Memorial Lectures in 2011 and 2013; the Mathematical Science Research Institute (MSRI), which funded my trip to the summer graduate workshop on mathematical relativity in 2012; the Park City Mathematics Institute (PCMI), which funded my trip to the summer program on geometric analysis; and the Erwin Schr\"odinger Institute (ESI), the European Mathematical Society (EMS), the International Association of Mathematical Physics (IAMP), and the National Science Foundation (NSF), which funded my trip to the summer school on mathematical relativity at ESI in 2014.

I thank my friends in Durham, too numerous to name, whose friendship and encouragement played a major role in the completion of my degree.

Finally I thank my parents, Stewart and Carolyn, and my sister, Kathryn, for their love, advice, and support.

%
%
%
\chapter{Introduction} \label{chap:intro}

\section{General Relativity} \label{sec:gr}

\subsection{Introduction to General Relativity} \label{ssec:gr_intro}
General relativity is a physical theory written in the mathematical language of semi-Riemannian geometry. In this dissertation we assume a basic knowledge of semi-Riemannian geometry---see \cite{oneill83}. General relativity models the universe as a $4$-dimensional manifold called a \emph{spacetime} which is equipped with a Lorentzian metric of signature $\mathord{-}\mathord{+}\mathord{+}\mathord{+}$ or $\mathord{+}\mathord{-}\mathord{-}\mathord{-}$. In this dissertation we will use the $\mathord{-}\mathord{+}\mathord{+}\mathord{+}$ convention. The fundamental equation of general relativity is the Einstein equation, which in geometrized units (see \cref{ssec:gr_units}) is
\begin{equation} \label{eq:einstein}
G + \Lambda g = 8\pi T.
\end{equation}
Here $G$ is the Einstein curvature tensor, a $(0,2)$-tensor defined in terms of the metric tensor $g$, the Ricci curvature tensor $\Ric$, and the scalar curvature $R$ as
\begin{equation} \label{eq:einsteincurvature}
G = \Ric - \frac12 Rg.
\end{equation}
The parameter $\Lambda$ is the cosmological constant, whose effect on the evolution of the universe is only seen at very large scales of distance and time. In this dissertation we will only be studying the universe at galactic scales and take $\Lambda=0$ to simplify matters. On the right side of \cref{eq:einstein} we have the energy-momentum tensor $T$, another $(0,2)$-tensor which encapsulates the information about the matter and energy content of the universe. Both $G$ and $T$ are symmetric tensors.

At the highest descriptive level, \cref{eq:einstein} says that the curvature of spacetime is correlated with the matter and energy content of the universe. The presence of matter and energy warps spacetime; conversely, the curvature of spacetime affects the motion of matter and energy. Light rays (i.e., photons) and so-called ``test particles'' (objects with masses small enough that their effect on the curvature tensor is negligible) follow geodesics, which are purely geometrically defined curves. Thus gravity is explained using geometry.

\subsection{Geometrized Units} \label{ssec:gr_units}
It is common in mathematical physics to choose units so that fundamental physical constants have magnitude $1$. For example, by measuring distance in light-years and time in years, the speed of light is
\begin{equation*}
c = \frac{1\text{ light-year}}{1\text{ year}} = \SI[per-mode=symbol]{1}{\lightyear\per\year}.
\end{equation*}
Often once this choice is made the units are omitted and it is said that ``$c=1$''.

In general relativity it is customary to measure mass, time, and distance in units such that $G=c=1$. Here $G$ is the universal gravitational constant, not the Einstein tensor from \cref{eq:einstein}. Units such that $G=c=1$ are called \emph{geometrized units}. Geometrized units are useful because they simplify formulas, but they also lead to strange locutions. For example one can say ``the Sun's radius is $\num{2.3}$ seconds'' or ``the Sun's radius is $\num{9.4e35}$ kilograms'' and not be talking nonsense. The first statement means that light takes $\num{2.3}$ seconds to cross a distance equal to the Sun's radius. An easy way to understand the second statement is to use the formula for the Schwarzschild radius of a black hole of mass $M$:
\begin{equation} \label{eq:schwarzschild_radius}
R_{\text{sch}} = \frac{2GM}{c^2}.
\end{equation}
In geometrized units, \cref{eq:schwarzschild_radius} becomes
\begin{equation}
R_{\text{sch}} = 2M.
\end{equation}
Thus, an easy way to understand ``the Sun's radius is $\num{9.4e35}$ kilograms'' is to restate it as ``the radius of a black hole with a mass of $\num{9.4e35}$ kilograms would be half the Sun's radius''.

\begin{table}[p]
\centering
\caption{Common masses in geometrical units of time or distance.} \label{table:masses}
\begin{tabular}{|l|c|c|c|c|}
\hline
\textbf{Unit} & \textbf{Seconds} & \textbf{Years} & \textbf{Meters} & \textbf{AUs} \\
\hline
\hline
kilogram & $\num{2.48e-36}$ & $\num{7.85e-44}$ & $\num{7.43e-28}$ & $\num{4.96e-39}$ \\
\hline
\hline
\textbf{Astronomical Body} & \textbf{Seconds} & \textbf{Years} & \textbf{Meters} & \textbf{AUs} \\
\hline
\hline
Sun & $\num{4.93e-6}$ & $\num{1.56e-13}$ & $\num{1480}$ & $\num{9.87e-9}$ \\
\hline
Earth & $\num{1.48e-11}$ & $\num{4.69e-19}$ & $\num{.00443}$ & $\num{2.96e-14}$ \\
\hline
Moon & $\num{1.82e-13}$ & $\num{5.77e-21}$ & $\num{5.45e-5}$ & $\num{3.65e-16}$ \\
\hline
Jupiter & $\num{4.70e-9}$ & $\num{1.49e-16}$ & $\num{1.41}$ & $\num{9.42e-12}$ \\
\hline
Cygnus X-1 black hole & $\num{7.4e-5}$ & $\num{2.3e-12}$ & $\num{22000}$ & $\num{1.5e-7}$ \\
\hline
Sag A* black hole & $\num{20}$ & $\num{6e-7}$ & $\num{6e9}$ & $\num{.04}$ \\
\hline
Milky Way & $\num{e7}$ & $\num{e-1}$ & $\num{e15}$ & $\num{e4}$ \\
\hline
\end{tabular}
\end{table}

\begin{table}[p]
\centering
\caption{Common times in geometrical units of distance or mass.} \label{table:times}
\begin{tabular}{|l|c|c|c|c|}
\hline
\textbf{Unit} & \textbf{Meters} & \textbf{AUs} & \textbf{Kilograms} & \textbf{Solar Masses} \\
\hline
\hline
second & $\num{3.00e8}$ & $\num{.00200}$ & $\num{4.04e35}$ & $\num{203000}$ \\
\hline
day & $\num{2.59e13}$ & $\num{173}$ & $\num{3.49e40}$ & $\num{1.75e10}$ \\
\hline
year & $\num{9.46e15}$ & $\num{63200}$ & $\num{1.27e43}$ & $\num{6.41e12}$ \\
\hline
\hline
\textbf{Astronomical Time} & \textbf{Meters} & \textbf{AUs} & \textbf{Kilograms} & \textbf{Solar Masses} \\
\hline
\hline
age of universe & $\num{1.31e26}$ & $\num{8.73e14}$ & $\num{1.76e53}$ & $\num{8.84e22}$ \\
\hline
age of solar system & $\num{4.3e25}$ & $\num{2.9e14}$ & $\num{5.8e52}$ & $\num{2.9e22}$ \\
\hline
\end{tabular}
\end{table}

\begin{table}[p]
\centering
\caption{Common distances in geometrical units of mass or time.} \label{table:distances}
\begin{tabular}{|l|c|c|c|c|}
\hline
\textbf{Unit} & \textbf{Kilograms} & \textbf{Solar Mass} & \textbf{Seconds} & \textbf{Years} \\
\hline
\hline
meter & $\num{1.35e27}$ & $\num{.000677}$ & $\num{3.34e-9}$ & $\num{1.06e-16}$ \\
\hline
AU & $\num{2.01e38}$ & $\num{1.01e8}$ & $\num{499}$ & $\num{1.58e-5}$ \\
\hline
light-year & $\num{1.27e43}$ & $\num{6.41e12}$ & $\num{3.16e7}$ & $\num{1}$ \\
\hline
parsec & $\num{4.16e43}$ & $\num{2.09e13}$ & $\num{1.03e8}$ & $\num{3.26}$ \\
\hline
\hline
\textbf{Astronomical Body} & \textbf{Kilograms} & \textbf{Solar Mass} & \textbf{Seconds} & \textbf{Years} \\
\hline
\hline
Sun (mean radius) & $\num{9.37e35}$ & $\num{471000}$ & $\num{2.32}$ & $\num{7.36e-8}$ \\
\hline
Earth (mean radius) & $\num{8.58e33}$ & $\num{4310}$ & $\num{.0213}$ & $\num{6.73e-10}$ \\
\hline
Moon (mean radius) & $\num{2.34e33}$ & $\num{1180}$ & $\num{.00579}$ & $\num{1.84e-10}$ \\
\hline
Jupiter (mean radius) & $\num{9.41e34}$ & $\num{47300}$ & $\num{.233}$ & $\num{7.39e-9}$ \\
\hline
\end{tabular}
\end{table}

In \cref{table:masses,table:times,table:distances} we have collected some common astronomical masses, times, and distances and list their values in geometrized units. In this dissertation we will adhere to the following convention:
\begin{convention} \label{conv:geometrized_units}
Unless otherwise noted, all formulas will be given in geometrized units and all physical quantities will be given in (light-)years.
\end{convention}

\subsection{General Relativity: A Success Story} \label{ssec:gr_success}
The study of \cref{eq:einstein} has yielded amazing physical and mathematical insights. In 1916 Einstein used general relativity to explain the anomalous precession of Mercury and predicted the correct angle at which light would be deflected by the sun. He also predicted the gravitational redshift of light, which was definitively observed in 1959, the existence of gravitational waves, for which indirect evidence was discovered in 1993, and the possibility of gravitational lensing, which is currently being used extensively to map the distribution of matter in the universe. The Schwarzschild solution, discovered in 1916, suggested the possibility of black holes, whose existence is now accepted.

In addition to these successes in solar system dynamics and astrophysics, general relativity is indispensible in cosmology. The Friedmann-Lema\^itre-Robertson-Walker metric provides the standard model of the universe and includes the Big Bang. Also, the most popular way to explain the accelerating expansion of the universe observed in 1998 is with the cosmological constant $\Lambda$. This often goes under the name of ``dark energy''.

Fundamental theoretical work in geometry contributing to our understanding of the universe includes the positive mass theorem first proved by Schoen and Yau \cite{schoen79,schoen81b} and later Witten \cite{witten81}, the singularity theorems of Hawking and Penrose \cite{penrose65,hawking75}, and the proofs of the Riemannian Penrose inequality by Huisken and Ilmanen \cite{huisken01} and Bray \cite{bray01}.

\section{Large-Scale Unexplained Curvature (Dark Matter)}

\subsection{Evidence for Dark Matter} \label{ssec:evidence}
In 1933, the astronomer Fritz Zwicky measured the velocites (via gravitational redshift) of galaxies belonging to the Coma cluster (see \cref{fig:coma}). The galaxies seemed to be moving too fast; that is, given the total mass of the cluster implied by the amount of visible matter, the gravitational pull on most galaxies should have been nowhere near strong enough to keep the galaxies from flying off into deep space, given their velocities. In fact, a quantitative analysis showed that for the galaxies to remain gravitationally bound, the total mass of the Coma cluster must be at least ten times the mass of the visible matter. Zwicky hypothesized that ``Dunkle Materie''---dark matter---was present in the Coma cluster, but invisible because it did not interact with light---see \cite{zwicky37}.

\begin{figure}[hbt]
\includegraphics[width=\linewidth]{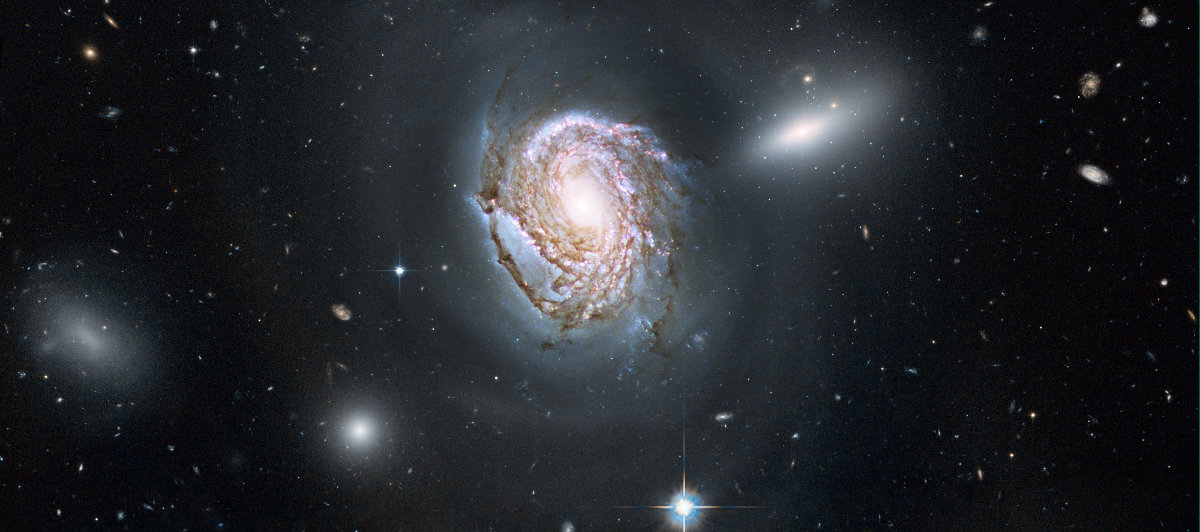}
\caption{Spiral galaxy NGC 4911 in the Coma Cluster. Credit: NASA, ESA, and the Hubble Heritage Team (STScI/AURA).}
\label{fig:coma}
\end{figure}

In the decades since Zwicky's proposal, other evidence has arisen that supports the hypothesis of dark matter. Constraints from Big-Bang nucleosynthesis point to a baryon density only 5\% of the critical density, whereas detailed studies of the anisotropies in the cosmic microwave background (CMB) indicate that the universe is flat (i.e. at critical density). Furthermore the anisotropies in the CMB are so small that without dark matter to ``seed'' large-scale structure growth, the universe could not be as far along in its evolution as it is today. These arguments suggest that most of the matter/energy content of the universe is non-baryonic.

Gravitational lensing studies have also produced evidence for dark matter. Clusters of galaxies can act as giant gravitational lenses and produce multiple images of background galaxies. This is an example of ``strong gravitational lensing''. Strong gravitional lensing can provide an estimate of the mass of the lens, and the evidence consistently points to the mass of clusters being much greater than the visible mass. Studying ``weak gravitation lensing'' is also very popular. One famous study \cite{markevitch04} of the ``bullet cluster'' describes observations of the collision of two clusters of galaxies (see \cref{fig:bullet_cluster}). In clusters, most of the baryons do not lie in stars but in the intracluster medium (ICM) in the form of hot gas. In a collision of clusters, it is only the gas that really collides and is slowed by friction. The stars in the galaxies are much too far apart to collide and pass through each other relatively unaffected. Thus, a collision between clusters separates the ICM from the stars. \Cref{fig:bullet_cluster} shows this separation. In pink is the ICM, and on the right is a very obvious bow shock from which the epithet ``bullet cluster'' derives. In blue are the galaxies with their stars. Remarkably, gravitational lensing reveals that the blue regions contain much more matter than the pink regions, exactly the opposite of what we would expect if all the matter were baryonic. The bullet cluster is usually cited as one of the best pieces of evidence for the existence of large quantites of non-baryonic dark matter which can be separated from baryonic matter.

\begin{figure}[hbt]
\includegraphics[width=\linewidth]{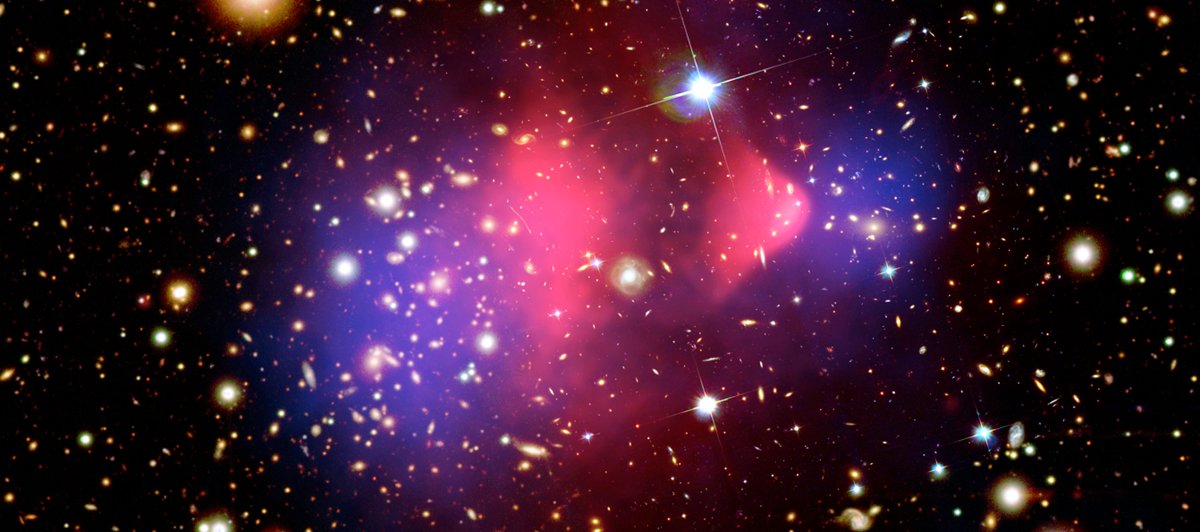}
\caption{The ``bullet cluster''. In the pink regions lies the gas, which contains most of the baryons. In the blue regions lie the stars. Gravitational lensing shows that most of the mass lies in the blue regions. This suggests that most of the mass of clusters is non-baryonic dark matter. Credit: X-ray: NASA/CXC/CfA/M.Markevitch et al.; Optical: NASA/STScI; Magellan/U.Arizona/D.Clowe et al.; Lensing Map: NASA/STScI; ESO WFI; Magellan/U.Arizona/D.Clowe et al.}
\label{fig:bullet_cluster}
\end{figure}

The evidence for dark matter that is most relevant for this dissertation comes from studying disk galaxies. In the 1970s, Vera Rubin and Kent Ford began publishing ``rotation curves'' for disk galaxies which plotted the rotational velocity of stars, gas, and dust in a galaxy as a function of radius---see \cite{rubin70,rubin80}. If all matter were baryonic, rotation curves would decline significantly with increasing radius, and stars at the outer edges of a galaxy would orbit much more slowly than stars close to the center. Instead, most rotation curves are roughly flat, indicating that most baryonic matter in a disk galaxy orbits at the same speed. The standard explanation for this phenomenon (but see \cref{ssec:mond} below) is that each disk galaxy is embedded in a large spherical ``halo'' of dark matter which dominates the baryonic matter by a factor of 5 or more and which flattens out the rotation curve.

\subsection{Interpreting the Evidence} \label{ssec:interpretation}
We have described many of the types of observations that lead most experts to believe in the existence of dark matter. It is important to understand that dark matter has never been directly observed. All astronomical observations ultimately originate with the act of gathering electromagnetic radiation with telescopes. Because dark matter is, almost by definition, matter which does not interact with light, it could never be observed astronomically, only inferred. For this reason, arguments about dark matter which rely mainly on astronomical observations are especially ``theory-laden''. For example, at the most basic level all the arguments from the previous section assume that the theory of general relativity and its low-field, non-relativistic limit (Newtonian gravity) are basically correct. If this assumption is not justified, then the reason our observations don't make sense is that we just don't understand gravity, and the arguments for dark matter fall flat. Some physicists have pursued this line of thought and proposed various modifications of general relativity and Newtonian gravity. The most popular is Modified Newtonian Dynamics, or MOND. We briefly discuss MOND in \cref{ssec:mond}.

Even if general relativity and Newtonian gravity are basically correct, there is still a problem of interpretation of the data. In general relativity, gravity is a manifestation of the curvature of spacetime. The mysteries which the dark matter hypothesis is supposed to solve are gravitational mysteries, or, we could just as easily say, \emph{curvature} mysteries: we observe curvature of spacetime on large scales for which we are unable to account. Whether this anomalous curvature is due to some kind of particle which we could detect with sufficiently advanced technology or is a purely geometrical phenomenon is an open question.

As an analogy, consider the cosmological constant term $\Lambda g$ which appears in Einstein's equation \eqref{eq:einstein}. We have the freedom to place it on the left or right side of the equals sign without changing the mathematics. On the one hand if it is on the right, then it contributes to the energy-momentum tensor $T$ which contains the information about the matter and energy content of the universe. Thus we are tempted talk of ``dark energy'' which permeates spacetime. On the other hand if it is on the left, then we can incorporate it in the Einstein curvature tensor and we are tempted to talk of a cosmological constant, a purely geometrical phenomenon which affects how matter and energy curve spacetime.

Similarly, the large-scale unexplained curvature we observe might be due to large quantities of dark matter particles which are curving spacetime. Or, the curvature might be more embedded in the laws of physics than physical stuff embedded in the universe.

\subsection{Particle Theories of Dark Matter}
The most popular theories of dark matter today postulate a dark matter particle in the $\Lambda$CDM cosmological paradigm. $\Lambda$ is the cosmological constant and the acronym CDM stands for ``cold dark matter''. Cold dark matter is dark matter which travels at non-relativistic speeds, as opposed to hot dark matter which travels at relativistic speeds. It is generally agreed that hot dark matter is not viable because it does not lead to the type of structure formation necessary to create our universe. Cold dark matter is by far the most popular type of dark matter under consideration. The most popular particle candidates which are cold are WIMPs.

Weakly Interacting Massive Particles (WIMPs) interact with themselves and other particles via the gravitational force and the weak force. The inspiration for WIMPs comes from particle physics. Some simple extensions of the Standard Model involving supersymmetry predict WIMP-like particles, and it is a natural idea to identify these hypothetical particles with the mysterious dark matter. However, the theory of WIMPs does have its problems. The odds against supersymmetry are increasing as the Large Hadron Collider continues to operate without discovering evidence for it, so the main theoretical structure supporting WIMPs is not looking strong. Also, numerical simulations of WIMPs in galaxies continue to support the notion that there should be ``cusps'' of WIMPs at the centers of galaxies---that is, most every numerical simulation produces a distribution of WIMPs in which the density at the galactic center is unbounded. Astronomical observations which attempt to map the distribution of dark matter in actual galaxies do not observe cusps but ``cores''---the density of dark matter levels off at the center. This is known as the cusp/core problem. Finally, direct detection experiments for WIMPs have been operating on Earth for decades and have never reported any unambiguously positive results.

There are many, many hypothetical particles which theorists have put forward as candidates for dark matter. For a good introduction to the possibilities we recommend \cite{krauss00,nicolson07}.

\section{Wave Dark Matter} \label{sec:wdm}
\subsection{Dark Matter from Geometry}
The theories of dark matter described in the previous section draw on particle physics for their inspiration. But since the dark matter mystery is fundamentally a gravitational mystery, and given the numerous successes of general relativity as a gravitational theory (\cref{ssec:gr_success}), perhaps we should turn to general relativity for inspiration. Moreover, if we take seriously the key insight of general relativity that gravity is geometry, then perhaps we should turn to geometry for inspiration. In 2013 Hubert Bray published a paper \cite{bray13b} in which he proposed to examine what happens if we do not make the usual assumption that the connection of the universe is the Levi-Civita connection. He found that the deviation of a spacetime connection from the Levi-Civita connection could be described using a scalar function and that this function would obey the Klein-Gordon wave equation. The full equations of the theory are the coupled Einstein-Klein-Gordon equations
\begin{subequations} \label{eq:ekg}
\begin{align}
G + \Lambda g &= 8\pi\left( \frac{df\tensor d\bar{f}+d\bar{f}\tensor df}{\Upsilon^2} - \left(\frac{\abs{df}^2}{\Upsilon^2}+\abs{f}^2\right)g\right) \label{eq:ekg1} \\
\Box f &= \Upsilon^2 f. \label{eq:ekg2}
\end{align}
\end{subequations}
The reader will note that \cref{eq:ekg1} is just the Einstein equation \eqref{eq:einstein} with a particular form for the energy-momentum tensor. \Cref{eq:ekg2} is the Klein-Gordon equation and $f$ is the scalar function which describes the deviation of the connection from the Levi-Civita connection. The parameter $\Upsilon$ is a fundamental constant of the theory. \Cref{chap:ekgps} of this dissertation disusses some results concerning the Einstein-Klein-Gordon system in spherical symmetry.

Bray proposed that the scalar function $f$ should be identified with dark matter and sought to model dark matter in spiral galaxies. He  obtained intriguing results with respect to the growth of spiral structure. In a subsequent paper \cite{bray12} he obtained results indicating that the wave dark matter is consistent with the existence of ``shells'' in elliptical galaxies. In \cite{bray13a,parrythesis}, Alan Parry examined wave dark matter in the context of dwarf spheroidal galaxies. In this dissertation we propose to investigate wave dark matter in the context of disk galaxies and its prospects for accounting for an astronomical phenomenon called the baryonic Tully-Fisher relation. We introduce the baryonic Tully-Fisher relation in \cref{sec:btfr}.

\Cref{eq:ekg1,eq:ekg2} (along with Maxwell's equations for electromagnetism) can be derived using a variational principle from the following action:
\begin{equation} \label{eq:ekgmaction}
\SF(g,f,A) = \int \bigg[R-2\Lambda -16\pi\bigg(\frac{\abs{df}^2}{\Upsilon^2}+\abs{f}^2+\frac14\abs{dA}^2\bigg)\bigg]\,dV.
\end{equation}
Here $A$ is a one-form, the electromagnetic potential. In \cref{app:action} we use a variational principle to derive the Einstein-Klein-Gordon-Maxwell equations. The action \eqref{eq:ekgmaction} is remarkably simple, and, if the theory of wave dark matter is correct, describes most of the matter and energy content of the universe for most of its history. For in its earliest times the universe was dominated by electromagnetic radiation, at later times by dark matter, and today is dominated by dark energy.

\subsection{Wave Dark Matter as a Particle Theory}
Although Bray motivated the theory of wave dark matter using a geometrical perspective from general relativity, the theory has been investigated before under other names such as scalar field dark matter \cite{guzman00,guzman04,magana12,suarez14,sahni00,khlopov85} and boson star dark matter or Bose-Einstein condensate (BEC) dark matter \cite{seidel90,sin94,ji94,lai07,matos07,sharma08,urenalopez10}. Results concerning axionic dark matter are also relevant as the axion is also a scalar field particle. Under all these other names the motivation is a more conventional one from particle physics. But the underlying equations---the Einstein-Klein-Gordon equations \eqref{eq:ekg1} and \eqref{eq:ekg2}---are the same, so which motivation we prefer is irrelevant for the theory's predictions. In a particle physics approach, the relationship between the fundamental constant $\Upsilon$ and the particle mass $m$ is
\begin{equation}
m = \frac{\hbar\Upsilon}{c} = \SI{2.09e-23}{eV}\left(\frac{\Upsilon}{\SI{1}{\lightyear^{-1}}}\right).
\end{equation}
If wave dark matter turns out to be the correct theory of dark matter, the question of which approach corresponds more closely to reality might remain ambiguous. The only way we foresee this question being resolved is the direct detection of the dark matter particle.

\section{The Baryonic Tully-Fisher Relation} \label{sec:btfr}
\Cref{chap:tullyfisher} of this dissertation discusses the relevance of wave dark matter for an empirical astrophysical relation concerning disk galaxies known as the baryonic Tully-Fisher relation. In this section we introduce this relation.

\subsection{Observational Astronomy}
We first need to introduce a few basic concepts from observational astronomy.

The \emph{luminosity} of a galaxy is a measure of the amount of energy emitted by the galaxy per unit time. It is correlated with brightness, but brightness usually refers to the galaxy's output in the visible wavelengths, whereas luminosity accounts for the galaxy's output at all wavelengths. It is a distance-independent quantity. The \emph{baryonic mass} of a galaxy is the total mass of the galaxy's baryons, i.e. the ordinary matter of the periodic table. Measuring/estimating luminosity and baryonic mass is a very complicated business which we will not attempt to describe. Note that both concepts suffer from being rather ill-defined because galaxies do not have sharp edges. Different conventions for how to define where a galaxy ends can lead to wildly different values for these parameters in the literature.

\emph{Emission lines} are sharp peaks in the galaxy's electromagnetic spectrum which correspond to various transitions between electron energy levels in atoms. Their presence or absence can reveal a lot about the galaxy's composition. For example, in neutral hydrogen atoms there is a ``hyperfine'' transition which occurs very rarely. When it occurs, a photon is emitted with a wavelength of $\mathord{\sim}\SI{21}{\centi\meter}$ at a frequency of $\mathord{\sim}\SI{1420}{\mega\hertz}$. A source whose spectrum has a $\SI{21}{\centi\meter}$ line contains large amounts of neutral hydrogen (``large amounts'' because the transition is so rare). The spectra of galaxies always contain the neutral hydrogen line because galaxies contain huge ``H~I'' regions of neutral hydrogen gas. This line is particularly useful in observational astronomy because the microwave wavelength penetrates cosmic dust whereas visible wavelengths do not.

Emission lines are subject to the Doppler effect. If a source is moving towards us its spectrum is shifted towards shorter wavelengths/higher frequencies and if it is moving away from us its spectrum is shifted towards longer wavelengths/lower frequencies. These shifts are most commonly called ``blueshift'' and ``redshift'', respectively, because in the visual spectrum the color blue lies on the low wavelength end and the color red on the high wavelength end. A disk galaxy's spectrum is the composite of the spectra of its stars, gas, and dust. If we are observing it at an angle, then because it is rotating the components will have slightly different Doppler shifts. This causes \emph{Doppler line broadening} of the emission lines. The width of the emission lines then provides some measure of the average rotational velocity of the galaxy. Galaxies which rotate faster have broader emission lines.

The \emph{rotation curve} of a disk galaxy gives the rotational velocity of the galaxy as a function of radius. We have already introduced this concept in the discussion of the evidence for dark matter in \cref{ssec:evidence}. Most disk galaxies have rotation curves which are roughly flat. Obtaining a galaxy's rotation curve involves analyzing the Doppler shifts in the galaxy's spectrum at different radii.

\subsection{The Baryonic Tully-Fisher Relation}
In a seminal paper \cite{tully77}, Richard Tully and James Fisher reported that for disk galaxies luminosity $L$ seemed to be related to the width of the \SI{21}{\centi\meter} line $w$ by a power law:
\begin{equation*}
L\varpropto w^x.
\end{equation*}
Because the width of the \SI{21}{\centi\meter} line is a proxy for a galaxy's rotational velocity, qualitatively the Tully-Fisher relation says that brighter galaxies rotate faster. This makes sense because a galaxy's brightness is correlated with its mass---more mass means more stars and more light, and also stronger gravity and faster rotation.

McGaugh proposed in \cite{mcgaugh00} an update to the Tully-Fisher relation. He noted that many smaller galaxies (low surface brightness galaxies, or LSBs) fell below the Tully-Fisher line in luminosity vs. line width log-log space because much of their mass was in the form of dim gas and dust instead of luminous stars. The Tully-Fisher relation, while linear for large galaxies, broke down for these LSBs. See \cref{fig:btfr}. The problem was that the luminosity measure was effectively counting only the stellar mass and ignoring the mass of gas and dust---and rotational velocity depends on mass irrespective of its form. McGaugh noted that when the mass of gas and dust was included and total baryonic mass was plotted vs. rotational velocity, a linear relationship was restored over the entire range of galactic masses. Again, see \cref{fig:btfr}. The baryonic Tully-Fisher relation has the form
\begin{equation} \label{eq:btfr}
M_b\varpropto v^x
\end{equation}
where $M_b$ is baryonic mass and $v$ is rotational velocity. For galaxies whose rotation curves can be resolved, $v$ can be the maximum observed velocity, or the average velocity, or the velocity at a fixed radius. Since the rotation curves are close to flat all of these choices are effectively equivalent. In studies where rotation curves are not available $v$ is obtained via measuring the neutral hydrogen line width.

\begin{figure}[p]
\centering
\includegraphics[width=\textwidth]{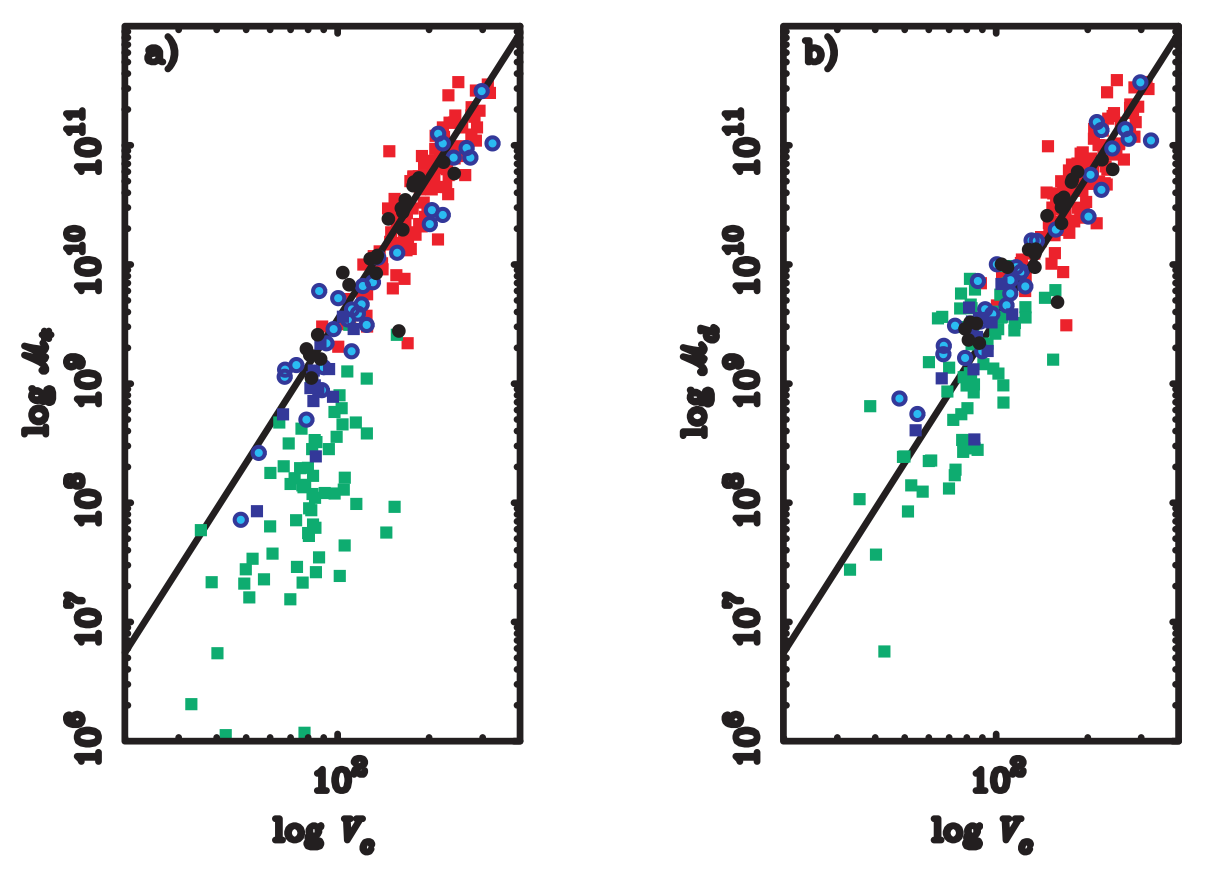}
\caption{This graphic---see \cite{mcgaugh00}---illustrates (a) the original Tully-Fisher relation and (b) the baryonic Tully-Fisher relation. Each point represents a galaxy. Squares represent galaxies whose rotational velocities have been estimated using the width of the \SI{21}{\centi\meter} line and circles represent galaxies whose rotational velocities have been estimated from resolved rotation curves. The left graph (a) plots stellar mass (which tracks luminosity) vs. rotational velocity, whereas the right graph (b) plots total baryonic mass (stellar, gas, and dust) vs. rotational velocity. Galaxies whose masses are gas-dominated are shown in green; they do not follow the original Tully-Fisher relation, but do follow the baryonic Tully-Fisher relation. The black line has slope~4. Reproduced from Figure 1 in \emph{McGaugh, S. S., Schombert, J. M., Bothun, G. D., \& de Blok, W. J. G. 2000, ApJ, 533, L99}. \textcopyright AAS. Reproduced with permission.}
\label{fig:btfr}
\end{figure}

The slope $x$ of the baryonic Tully-Fisher relation has been the subject of many investigations (see \cite{bell01,verheijen01,gurovich04,mcgaugh05,pfenniger05,geha06,begum08,stark09,trachternach09,gurovich10}). The literature seems to agree that $3\leq x\leq 4$ with numbers closer to $4$ being favored, but beyond that there is little agreement. Most of the uncertainty seems to stem from disagreement about how to estimate $M_b$. In addition, there is a danger of bias because a slope of 3 is thought to be more consonant with the standard $\Lambda$CDM cosmology whereas a slope of 4 is consonant with Modified Newtonian Dynamics, the most popular alternative to the dark matter hypothesis (see \cref{ssec:mond}).

\subsection{Modified Newtonian Dynamics (MOND)} \label{ssec:mond}
Recall the discussion in \cref{ssec:interpretation} in which we pointed out that in trying to account for the confusing observations described in \cref{ssec:evidence}, we can either seek a solution within the paradigm of general relativity and Newtonian gravity or look for a new theory of gravity. The most popular proposed solution of the latter type is known as Modified Newtonian Dynamics or MOND for short \cite{milgrom83,famaey12}. MOND, while it has other issues, can claim to explain the flat rotation curves of spiral galaxies and the baryonic Tully-Fisher relation. Indeed, it was designed for this purpose. In essence, whereas the combination of Newton's second law and law of gravity gives an acceleration due to gravity
\begin{equation*}
a = \frac{GM}{r^2},
\end{equation*}
MOND postulates
\begin{equation*}
a = \frac{\sqrt{GMa_0}}{r},
\end{equation*}
for an acceleration $a$ much less than a threshold acceleration $a_0$. The inclusion of the threshold acceleration is to leave solar system dynamics virtually unchanged. One immediately sees that for circular motion where $a=v^2/r$, we get $v=(GMa_0)^{1/4}$ (velocity independent of radius) and $M/v^4=(Ga_0)^{-1}$ (a Tully-Fisher relation), which seems promising. However, the theory has its own conflicts with data. One of the most problematic is that although MOND was created to get rid of the missing mass problem in galaxies, it has a missing mass problem at the level of clusters \cite{gerbal92,sanders99}. Even more problematic is the ``bullet cluster'' described in \cref{ssec:evidence}, whose existence seems to demonstrate that dark matter exists in large quantities and can be separated from baryonic matter. MOND remains a minority viewpoint among astrophysicists. Even so, there is still the important question of why it works so well for disk galaxies.

\subsection{Dark Matter and the Baryonic Tully-Fisher Relation} \label{sec:dmbtfr}
We have introduced the baryonic Tully-Fisher relation because it may have relevance for discriminating between theories of dark matter. Dark matter comprises the bulk of a galaxy's mass, is responsible for the flatness of the rotation curves, and is the dominant factor determining the $v$ in the relation \eqref{eq:btfr}. However the mass that appears in the relation is the baryonic mass, not the total mass or the dark mass. Thus the baryonic Tully-Fisher relation gives a quantitative link between the dark and the baryonic mass. A successful theory of dark matter needs to be consistent with the existence of the relation. We will have more to say on this subject in \cref{chap:tullyfisher} after we have described our results.

\chapter{The Einstein-Klein-Gordon and Poisson-Schr\"odinger Systems in Spherical Symmetry} \label{chap:ekgps}

In this chapter we study the Einstein-Klein-Gordon system in spherical symmetry from a mathematical perspective. In the next chapter we will consider the relevance of our results for theoretical astrophysics and the problem of dark matter.

\section{Introduction}
The Einstein-Klein-Gordon system couples the Einstein equation \eqref{eq:einstein} together with the Klein-Gordon equation for a complex-valued function. The full system is
\begin{subequations} \label{eqr:ekg}
\begin{align}
G+\Lambda g &= 8\pi\left( \frac{df\tensor d\bar{f}+d\bar{f}\tensor df}{\Upsilon^2} - \left(\frac{\abs{df}^2}{\Upsilon^2}+\abs{f}^2\right)g\right) \label{eqr:ekg1} \\
\Box f &= \Upsilon^2 f. \label{eqr:ekg2}
\end{align}
\end{subequations}
Here $\Upsilon$ is a positive constant. This system can be obtained via a variational principle from the action
\begin{equation*}
\SF(g,f) = \int R-2\Lambda-16\pi\left(\frac{\abs{df}^2}{\Upsilon^2}+\abs{f}^2\right)\,dV.
\end{equation*}
Requiring the metric $g$ and the scalar function $f$ to be critical points of the action $\SF$ leads to \cref{eqr:ekg1,eqr:ekg2}. See \cref{app:action}.

A full solution to the system \eqref{eqr:ekg} consists of a spacetime $N$ with a Lorentzian metric $g$ and a scalar function function $f:N\to\C$ such that \cref{eqr:ekg1,eqr:ekg2} are satisfied. We are interested in studying spherically symmetric static solutions to the Einstein-Klein-Gordon system where $N$ has the topology of $\R^4$. As an ansatz for the metric $g$ we take
\begin{equation} \label{eq:metric}
g = -e^{2V(r)}\,dt^2 + \left(1-\frac{2M(r)}{r}\right)^{-1}\,dr^2 + r^2\,d\theta^2 + r^2\sin^2\theta\,d\phi^2.
\end{equation}
Here $M(r)$ and $V(r)$ are functions of the coordinate $r$ only and $r^2\,d\theta^2 + r^2\sin^2\theta\,d\phi^2$ is the standard area form on the coordinate sphere of radius $r$. The reason for choosing this particular ansatz, with the functions $M(r)$ and $V(r)$, is that these two functions have natural physical interpretations---in the Newtonian limit, $M(r)$ will turn out to be the total mass contained inside the sphere of radius $r$ and $V(r)$ will turn out to be the Newtonian gravitational potential. For clarity we define
\begin{equation} \label{eq:Phi}
\Phi(r) = 1-\frac{2M(r)}{r}
\end{equation}
so that the metric \eqref{eq:metric} can be written as
\begin{equation} \label{eq:metric_Phi}
g = -e^{2V(r)}\,dt^2 + \Phi(r)^{-1}\,dr^2 + r^2\,d\theta^2 + r^2\sin^2\theta\,d\phi^2.
\end{equation}

\section{The Spherically Symmetric Static States}
The scalar function $f$ must be spherically symmetric, a function of $t$ and $r$ only. We make the additional ansatz that it can be written in the form
\begin{equation} \label{eq:fstaticstate}
f(t,r) = F(r)e^{i\omega t}
\end{equation}
for some real function $F(r)$ and real number $\omega$. A solution in which $f$ has the form \eqref{eq:fstaticstate} is called a \emph{static state}. In \cref{app:odederivation} we show that with these ans\"atze, solving the Einstein-Klein-Gordon system reduces to solving the following system of coupled ordinary differential equations:
\begin{subequations} \label{eq:ekgode}
\begin{gather}
M_r = 4\pi r^2\cdot\frac{1}{\Upsilon^2}\left[ \left(\Upsilon^2+\omega^2e^{-2V}\right)F^2 + \Phi F_r^2\right]\label{eq:ekgodeM} \\
\Phi V_r = \frac{M}{r^2} - 4\pi r\cdot\frac{1}{\Upsilon^2}\left[ \left(\Upsilon^2-\omega^2 e^{-2V}\right)F^2 - \Phi F_r^2 \right] \label{eq:ekgodeV} \\
F_{rr} + \frac{2}{r} F_r + V_rF_r + \frac12\frac{\Phi_r}{\Phi}F_r = \Phi^{-1}\left(\Upsilon^2-\omega^2e^{-2V}\right)F. \label{eq:ekgodeF}
\end{gather}
\end{subequations}
Note that the dependence on $t$ has disappeared, which is why solutions to these ODEs are called static states. Our notation for a solution will be
\begin{equation*}
(\omega;M,V,F)
\end{equation*}
because to specify a solution we need to specify the frequency $\omega$ appearing in \cref{eq:fstaticstate} and the three functions $M$, $V$, and $F$. Solutions can be found by numerical integration using a computer.

For initial conditions we take
\begin{align}
M(0)&=0 \label{eq:Minit}\\
V(0)&=V_0 \label{eq:Vinit}\\
F(0)&=F_0>0 \label{eq:Finit}\\
F_r(0)&=0. \label{eq:Frinit}
\end{align}
\Cref{eq:Minit} must hold for smoothness of the solution at $r=0$; the function $\Phi(r)$ cannot become infinite as $r\to 0$. We do not place any restriction on $V_0$ in \cref{eq:Vinit}---but see \cref{conv:Vinf} below. In \cref{eq:Finit} we require that $F_0$ be positive; this is because if $(\omega;M,V,F)$ is a solution, so is $(\omega;M,V,-F)$, and thus we might as well take $F(0)\geq 0$. We exclude $F(0)=0$ because this leads to the trivial solution. Finally, \cref{eq:Frinit} is another requirement for smoothness at $r=0$ because of the spherical symmetry.

Note that we have the freedom to add an arbitrary constant $\Vtil$ to the potential function $V(r)$. Looking at the form of the ODEs \eqref{eq:ekgode}, we see that if $(\omega;M,V,F)$ is a solution, then so is $(\omega e^{\Vtil};M,V+\Vtil,F)$. This corresponds to just a rescaling of the $t$ coordinate by a factor of $e^{\Vtil}$ in the spacetime metric \eqref{eq:metric_Phi}. Thus, adding a constant to $V(r)$ amounts to a change of coordinates which does not affect the solution.

For brevity we make the following definitions:
\begin{definition}
\begin{align}
M_\infty = \lim_{r\to\infty} M(r) \\
V_\infty = \lim_{r\to\infty} V(r).
\end{align}
\end{definition}

Numerical experimentation with values for $\omega$, $V_0$, $F_0$ reveals that almost all solutions blow up as $r\to\infty$; that is, $M_\infty=V_\infty=+\infty$ and $\lim_{r\to\infty} F(r)=\pm\infty$. We are not interested in these solutions that blow up because they cannot represent anything physical. We are interested in the rare solutions where $M_\infty$ and $V_\infty$ are finite and $\lim_{r\to\infty} F(r)=0$. In \cref{chap:tullyfisher} we will use these special solutions to model dark matter in disk galaxies. In a moment we will begin to describe the characteristics of these solutions, but first we introduce the following convention:

\begin{convention} \label{conv:Vinf}
For a solution $(\omega;M,V,F)$ with $V_\infty$ finite, we will assume that
\begin{equation} \label{eq:Vinf}
V_\infty = 0.
\end{equation}
\end{convention}
Any solution with $V_\infty$ finite can always be made to satisfy \cref{eq:Vinf} by performing the transformation mentioned above regarding the function $V(r)$; this will involve changing the value of $\omega$, but as we said above, this is really just a change of coordinates. There are at least two good reasons for using this convention: (1) The metric \eqref{eq:metric_Phi} is then asymptotic to Minkowski spacetime at infinity. (2) Since the physical interpretation of $V(r)$ is as the Newtonian potential, this convention agrees with the common convention in physics to have the potential equal zero at infinity.

With \cref{conv:Vinf} in place we now describe the special solutions just mentioned. (From now on when we will use ``static state'' to refer to one of these special solutions.) The static states come in the form of \emph{ground states} and \emph{excited states}. In \cref{fig:n0123}
\begin{figure}[p]
\centering
\hrule
$n=0$\includegraphics[width=0.9\textwidth]{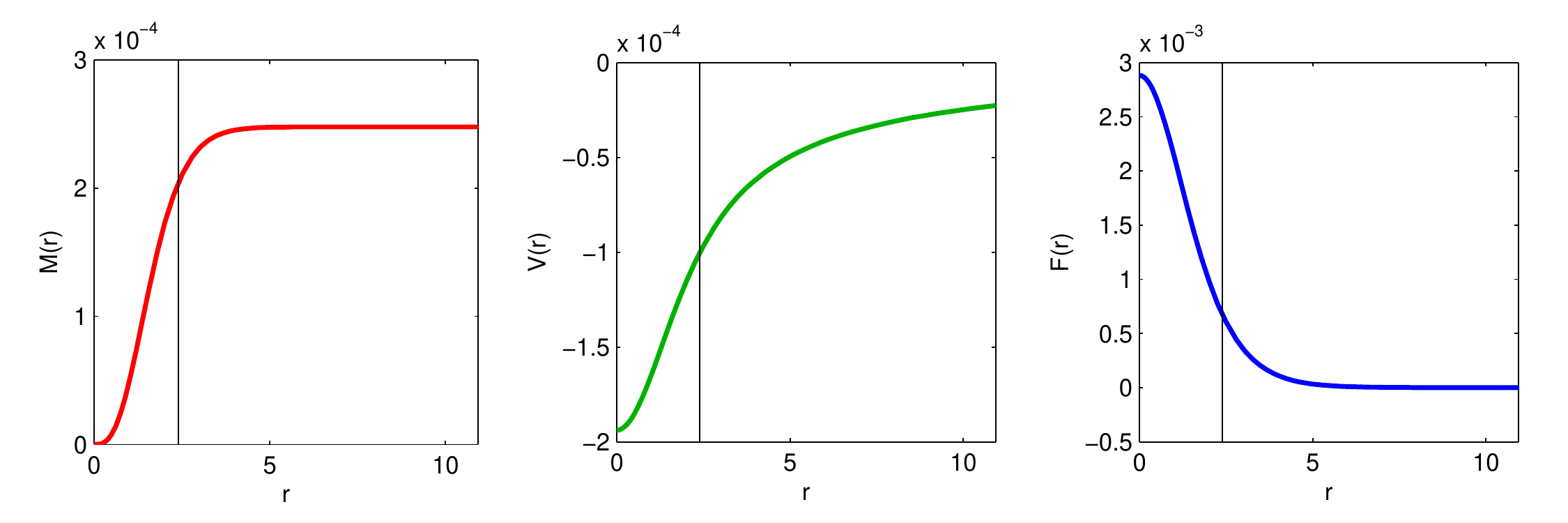}
\vspace{1pt}
\hrule
$n=1$\includegraphics[width=0.9\textwidth]{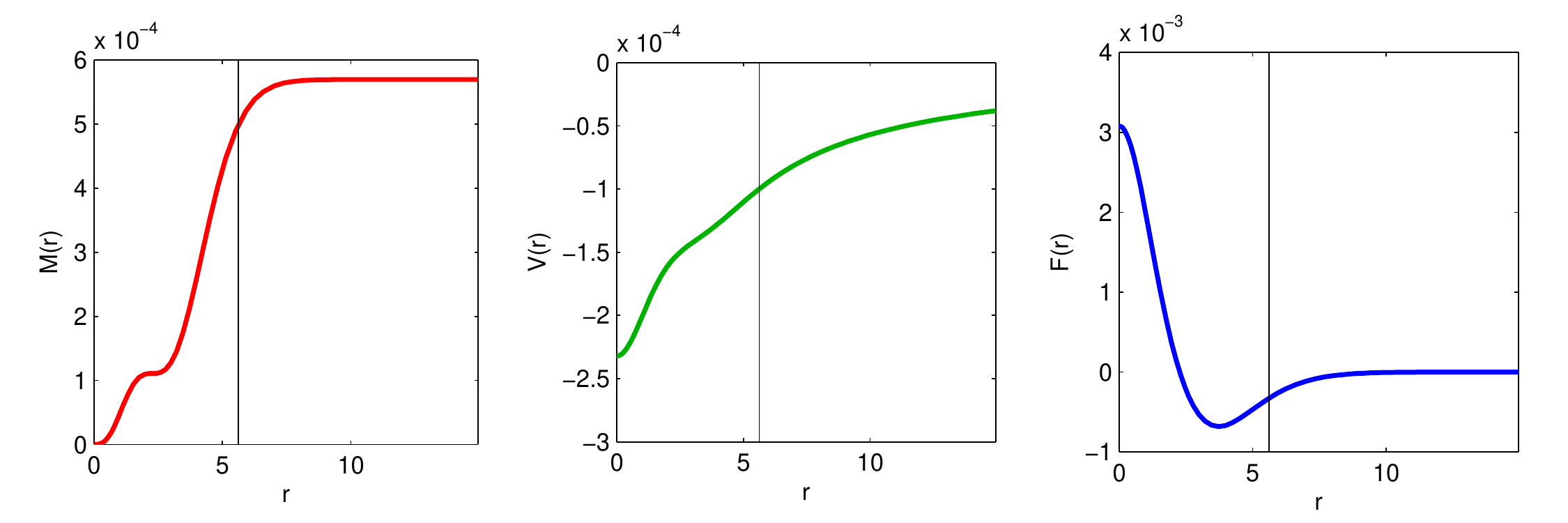}
\vspace{1pt}
\hrule
$n=2$\includegraphics[width=0.9\textwidth]{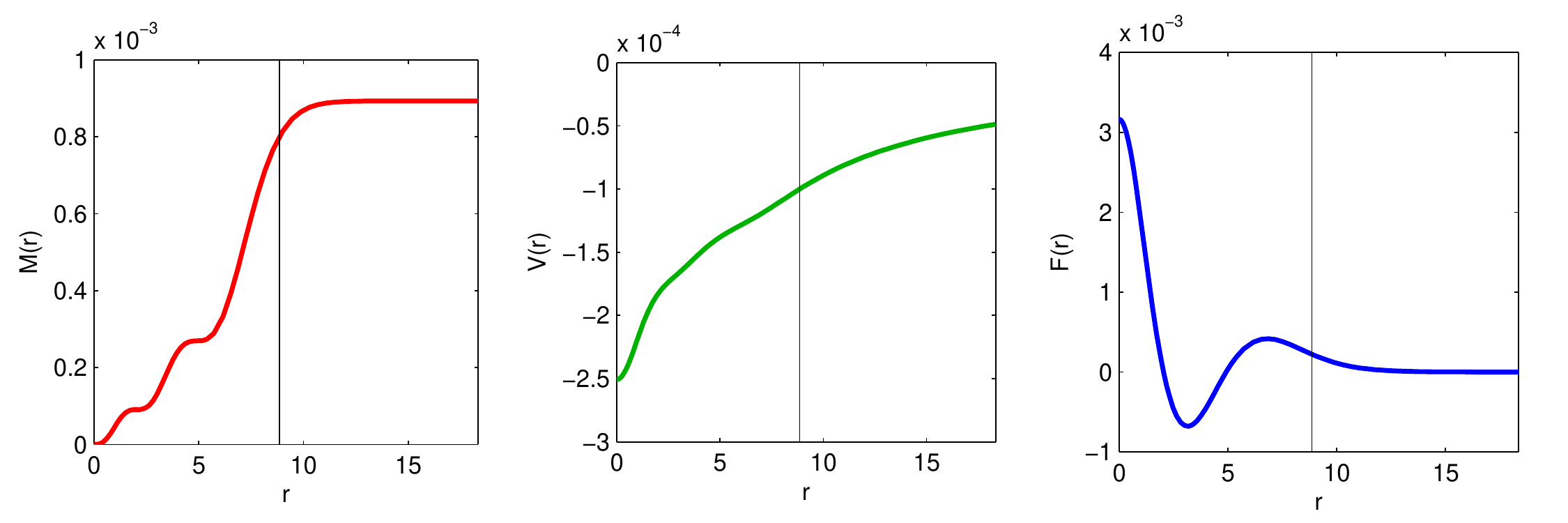}
\vspace{1pt}
\hrule
$n=3$\includegraphics[width=0.9\textwidth]{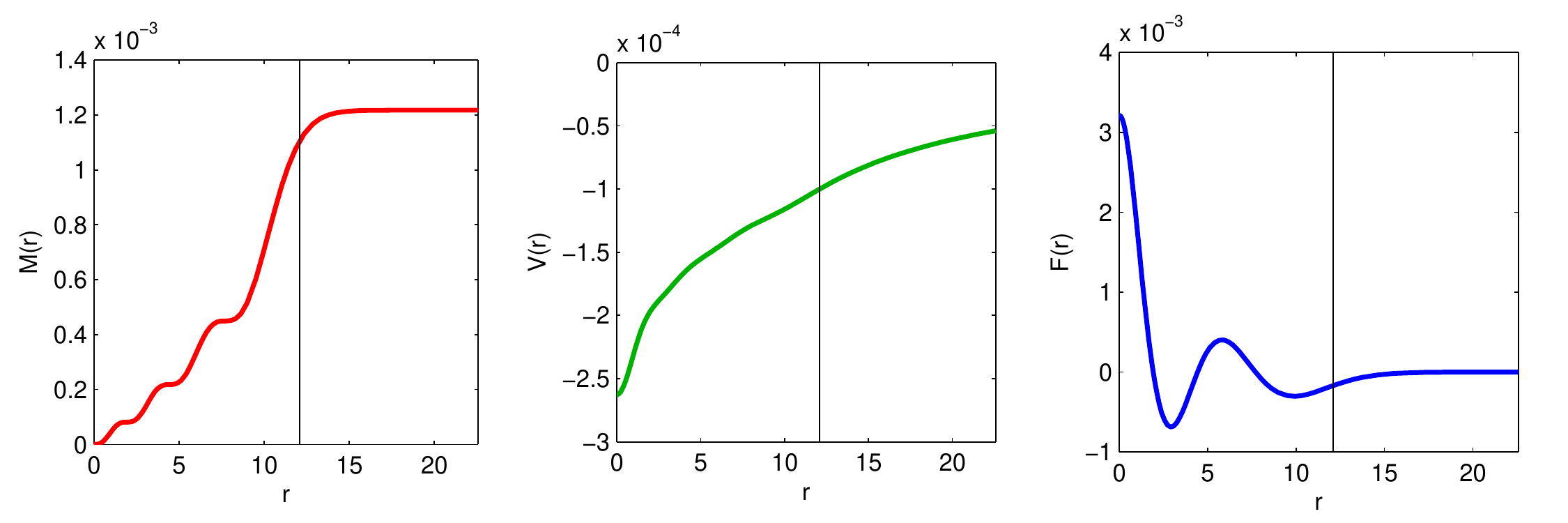}
\vspace{1pt}
\hrule
\caption{The ground state (top row) and first three excited states (second, third, and fourth rows) with $\Upsilon=100$ and $\omega=99.99$. The vertical black line in each plot marks the location of $\RDM$.}
\label{fig:n0123}
\end{figure}
we have graphed a ground state and first three excited states with $\Upsilon=100$ and $\omega=99.9$. For fixed $\Upsilon$ and $\omega$, there are countably many solutions corresponding to $n=0,1,2\ldots$ where $n$ is the number of zeros (or ``nodes'') of $F$. The static states live in a 3-parameter space, where two of the parameters are continuous and one is discrete (for example, $\Upsilon$, $\omega$, and $n$---but there are other ways of parametrizing the static states.)

\section{The Poisson-Schr\"odinger System} \label{sec:ps}
To gain more insight into the static states, we consider the system \eqref{eq:ekgode} and make the following approximations:
\begin{gather}
e^{2V}\approx 1, \qquad \Phi\approx 1, \qquad \frac{V_r}{\Upsilon\norm{V}_\infty}\approx 0, \qquad \frac{\Phi_r}{\Upsilon\norm{V}_\infty}\approx 0 \label{eq:minkowskiapprox} \\
\frac{\omega}{\Upsilon} \approx 1, \qquad \frac{F_r}{\Upsilon\norm{F}_\infty}\approx 0. \label{eq:lowspeedapprox}
\end{gather}
The approximations in \eqref{eq:minkowskiapprox} may be interpreted as saying that the metric \eqref{eq:metric_Phi} is close to the Minkowski metric (the low-field limit), and the approximations in \eqref{eq:lowspeedapprox} may be interpreted as saying that the group velocities of wave dark matter are much less than the speed of light (the non-relativistic limit). Applying all these approximations to the system \eqref{eq:ekgode} leads to the system
\begin{subequations} \label{eq:psode}
\begin{align}
M_r &= 4\pi r^2\cdot 2F^2 \label{eq:psodeM} \\
V_r &= \frac{M}{r^2} \label{eq:psodeV} \\
\frac{1}{2\Upsilon}\left(F_{rr} + \frac{2}{r} F_r\right) &= (\Upsilon-\omega+\Upsilon V)F. \label{eq:psodeF}
\end{align}
\end{subequations}
This system is the Poisson-Schr\"odinger system written in spherical symmetry. Here we have a special case of the well-known fact that the Poisson-Schr\"odinger system is the non-relativistic limit of the Einstein-Klein-Gordon system---see \cite{giulini12}. Also see \cref{fig:lowfieldlimit}. Thus, to understand the system \eqref{eq:ekgode} in the low-field, non-relativistic limit it suffices to understand the system \eqref{eq:psode}.

\Cref{eq:psodeM,eq:psodeV} immediately lead to the physical interpretations of $M(r)$ as the total mass contained inside a sphere of radius $r$ and $V(r)$ as the Newtonian gravitational potential. The quantity $2F^2$ represents the mass-energy density $\rho$ of the static state. Combining \cref{eq:psodeM,eq:psodeV} gives the Poisson equation $\laplacian V=4\pi\rho$ in spherical symmetry. For in spherical coordinates, the Laplacian of $V$ is given by $V_{rr}+(2/r)V_r$; thus
\begin{equation*}
\laplacian_r V = \left(\frac{M}{r^2}\right)_r + \frac{2}{r}\cdot\frac{M}{r^2} = \frac{M_r}{r^2} = 4\pi\cdot 2F^2 = 4\pi\rho.
\end{equation*}
The third equation \eqref{eq:psodeF} is the most interesting. It is the Schr\"odinger half of the Poisson-Schr\"odinger system. On the left side we recognize the Laplacian of $F$. Let us consider for a moment the differential equation
\begin{equation} \label{eq:ekgodeosc}
F_{rr} + \frac{2}{r}F_r = k F.
\end{equation}
If we let $h(r)=rF(r)$, then the corresponding differential equation for $h$ is $h_{rr}=kh$. Thus the general solution to \cref{eq:ekgodeosc} is
\begin{equation} \label{eq:ekgFsol}
F(r) = \begin{cases} A\frac{e^{\sqrt{k}r}}{r} + B\frac{e^{-\sqrt{k}r}}{r} & k>0 \\ A + \frac{B}{r} & k=0 \\ A\frac{\sin(\sqrt{-k}r)}{r} + B\frac{\cos(\sqrt{-k}r)}{r}& k<0 \end{cases}
\end{equation}
Thus, on any interval where the expression $k(r)=2\Upsilon(\Upsilon-\omega+\Upsilon V(r))$ is roughly constant, the solution to \cref{eq:psodeF} will look like one of the three solutions in \cref{eq:ekgFsol}. More explicitly, when $k(r)$ is negative, $F$ will exhibit oscillatory behavior, and when $k(r)$ is positive, $F$ will exhibit exponential behavior. Since $V$ is an increasing function, $k(r)$ increases with $r$. Remembering that one of our requirements on a static state solution $(\omega;M,V,F)$ is that $M_\infty<\infty$, looking at \cref{eq:psodeM} we see that $F$ must exponentially decay after a certain point. Since we are also requiring $V_\infty=0$, we must take $\omega<\Upsilon$ so that $\lim_{r\to\infty} k(r) = 2\Upsilon(\Upsilon-\omega)>0$. On the other hand, $k(r)$ cannot be positive (i.e., $F$ cannot have exponential behavior) for all $r$; looking at \cref{eq:ekgFsol}, we see that this would be in contradiction with the requirements from spherical symmetry that $M_r(0)=F_r(0)=0$. Thus, the only situation consistent with our requirements is when $0<\omega<\Upsilon$, such that $k(r)$ begins negative and limits to $2\Upsilon(\Upsilon-\omega)$. The corresponding behavior of $F$ is to start out oscillating and then to switch over to exponential behavior. The switch occurs at the point where $k(r)=0$; we label this point $\RDM$ (see \cref{fig:RDM}) in anticipation of the role it will play in \cref{chap:tullyfisher} as the effective radius of a dark matter halo.

\begin{definition} \label{def:rdm}
Given a static state $(\omega;M,V,F)$, we define $\RDM$ to be the radius at which the function $F$ switches from oscillatory to exponential behavior. See \cref{fig:RDM}.
\end{definition}

The previous discussion concerned the Poisson-Schr\"odinger system, but a similar discussion applies for the Einstein-Klein-Gordon system, in which case the expression which controls the oscillatory/exponential behavior of $F$ is $k(r)=\Upsilon^2-\omega^2e^{-2V}$. Thus for the Einstein-Klein-Gordon system, $\RDM$ satisfies $\Upsilon^2-\omega^2 e^{-2V(\RDM)} = 0$, which implies
\begin{equation} \label{eq:ekgrdmdef}
\omega e^{-V(\RDM)}=\Upsilon. \qquad\text{(EKG)}
\end{equation}
For the Poisson-Schr\"odinger system, $\RDM$ satisfies $\Upsilon-\omega+V(\RDM)=0$, which implies
\begin{equation} \label{eq:psrdmdef}
V(\RDM)=\omega-\Upsilon. \qquad\text{(PS)}
\end{equation}

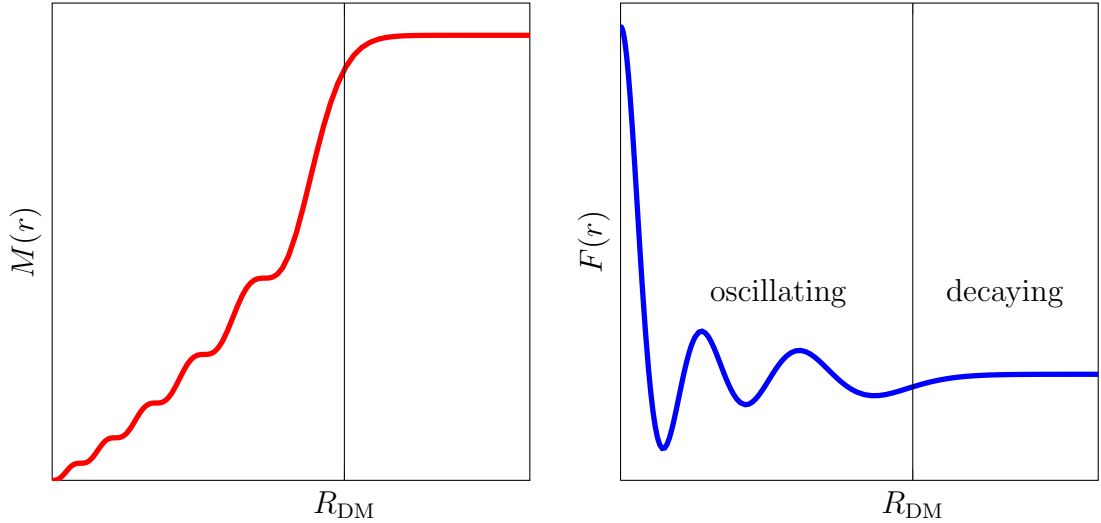
\begin{figure}[tb]
\begin{minipage}{0.5\textwidth}
\centering
%
%
\begin{tikzpicture}

\begin{axis}[%
width=2.5in,
height=2.5in,
scale only axis,
xmin=0,
xmax=30.3671705142261,
xtick={18.5777082618857},
xticklabels={{$\RDM$}},
ymin=0,
ymax=0.002,
ytick={\empty},
ylabel={$M(r)$}
]
\addplot [color=red,solid,line width=2.0pt,forget plot]
  table[row sep=crcr]{%
0	0\\
0.0332356029814985	3.30393565761331e-09\\
0.0663421225062712	2.62159629016282e-08\\
0.10895010289847	1.15495094428698e-07\\
0.162720510647507	3.80784511390699e-07\\
0.227456179238181	1.02147402656692e-06\\
0.303635119169971	2.36100618381545e-06\\
0.392444218060836	4.87929229198223e-06\\
0.496509417141059	9.25769395723457e-06\\
0.61385252606422	1.59721894745456e-05\\
0.745613582304799	2.52892154203276e-05\\
0.899943700679824	3.74080152110055e-05\\
1.06715083195142	5.01033272625052e-05\\
1.23394732163963	6.04437057683893e-05\\
1.40679660101579	6.75241658663642e-05\\
1.58779317224012	7.09746464473203e-05\\
1.77698477022245	7.16977277663693e-05\\
1.93263515320482	7.18414577944492e-05\\
2.05855625240157	7.27987194258511e-05\\
2.19705178572113	7.56201243794328e-05\\
2.34489060807404	8.13453893264509e-05\\
2.49748483291217	9.0436781456343e-05\\
2.65124315064379	0.00010250950703349\\
2.80395156120509	0.000116489068122393\\
2.95481121804899	0.000131015029432632\\
3.10475265069429	0.000144849554418745\\
3.25369441013513	0.000156848014718459\\
3.40160877229764	0.000166237351433433\\
3.54833422671333	0.00017270766451209\\
3.69337957138205	0.000176432932759527\\
3.83566103092266	0.000178017136128541\\
3.97298214009595	0.000178355422067279\\
4.10048272783418	0.000178398752716061\\
4.2151311967687	0.000178839243171304\\
4.35837208841213	0.000180707454745647\\
4.52543204319736	0.000185701517968503\\
4.72199769153315	0.000196372701529952\\
4.97929366236527	0.000218059715639265\\
5.18207053143531	0.000239564777279029\\
5.38484740050535	0.00026233804426298\\
5.57244461508394	0.000282180720327505\\
5.7525659440315	0.000298341747158436\\
5.92764994265099	0.000310285963519315\\
6.09841743904801	0.000318033087302064\\
6.26465113237884	0.00032217713130623\\
6.42517397557998	0.000323754499652774\\
6.57730262068593	0.000324013994018763\\
6.71473645990409	0.00032409928348552\\
6.84589888915889	0.000324798378288252\\
7.00342056069896	0.000327335438981547\\
7.1829085752079	0.00033353632032721\\
7.38584553116686	0.000345737080811963\\
7.61980367396072	0.000366813902711448\\
7.92153565891401	0.000402721005451255\\
8.15964669564255	0.000434250113150965\\
8.39775773237109	0.000464581372801489\\
8.61845219343352	0.000488704533935465\\
8.82978686583184	0.000506462248678552\\
9.03453542231297	0.000518079897431494\\
9.23324206291973	0.000524437495892202\\
9.4250027683414	0.000526988559089664\\
9.60710571603504	0.000527483089842252\\
9.77301490517159	0.000527562849038937\\
9.92836471369826	0.000528374790315957\\
10.1115273063333	0.000531477711167144\\
10.3207857380744	0.000539322868169832\\
10.5565194885201	0.000555061969829905\\
10.8244021987855	0.000582430906586004\\
11.1455729293125	0.00062681108839906\\
11.4671092558459	0.000678876893624165\\
11.7886455823793	0.000731527503146344\\
12.0705025206859	0.00077268170178589\\
12.338241105806	0.000804081305770549\\
12.5989469574928	0.000825969791145792\\
12.854635630129	0.000839168433823304\\
13.1052722855428	0.000845432627698092\\
13.348739191266	0.000847247874936332\\
13.5791433493141	0.000847374515086316\\
13.7950699688249	0.000848166278891185\\
14.0336093147311	0.000852051525788519\\
14.3161821124093	0.000863611630140755\\
14.643636001985	0.000889757859568889\\
15.0279392334862	0.000940553664913736\\
15.5165374339546	0.00103538475504352\\
15.9676658662819	0.00114578079081227\\
16.4187942986091	0.00126709788542637\\
16.8367280257115	0.00137941838143683\\
17.2538906475726	0.00148380072632171\\
17.6849671573728	0.00157822287086016\\
18.146893201513	0.00166127222215547\\
18.6785447941452	0.00173351486559996\\
19.1851829139022	0.00178186046122615\\
19.6918210336593	0.00181434294761119\\
20.2112583885931	0.00183558089556507\\
20.755594772506	0.0018489306179238\\
21.2158585153598	0.00185560565227938\\
21.6761222582137	0.001859620133222\\
22.109941561054	0.0018618689839798\\
22.5252820119604	0.00186315496909781\\
22.9341137685184	0.00186390880140458\\
23.342398600116	0.00186435165916902\\
23.7541779439895	0.00186460979386899\\
24.1726540956997	0.00186475810062806\\
24.6006628908652	0.00186484164785438\\
25.0409308878338	0.00186488757667135\\
25.4962468002845	0.00186491209773101\\
25.9695999236272	0.00186492474660279\\
26.4643124905998	0.00186493101358671\\
26.9841833346068	0.00186493397452749\\
27.5336562722658	0.00186493529572617\\
28.1180218348194	0.00186493584411651\\
28.7436423216585	0.00186493604951054\\
29.4181105183013	0.00186493611332514\\
30.1499232862591	0.00186493612496156\\
30.3671705142261	0.001864936125092\\
};
\addplot [color=black,solid,forget plot]
  table[row sep=crcr]{%
18.5777082618857	0\\
18.5777082618857	0.002\\
};
\end{axis}
\end{tikzpicture}%
\end{minipage}%
\begin{minipage}{0.5\textwidth}
\centering
%
%
\begin{tikzpicture}

\begin{axis}[%
width=2.5in,
height=2.5in,
scale only axis,
scaled y ticks = false,
xmin=0,
xmax=30.3671705142261,
xtick={18.5777082618857},
xticklabels={{$\RDM$}},
ymin=-0.001,
ymax=0.0035,
ytick={\empty},
ylabel={$F(r)$}
]
\node at (axis cs:5,.001) [anchor=north west] {oscillating};
\node at (axis cs:20,.001) [anchor=north west] {decaying};
\addplot [color=blue,solid,line width=2.0pt,forget plot]
  table[row sep=crcr]{%
0	0.00327856487224789\\
0.0332356029814985	0.00327640824431574\\
0.0663421225062712	0.00326997910264732\\
0.10895010289847	0.00325545319208777\\
0.162720510647507	0.00322720156146554\\
0.227456179238181	0.003178842126406\\
0.303635119169971	0.00310265385826968\\
0.392444218060836	0.00298918456916862\\
0.496509417141059	0.00282575753991805\\
0.61385252606422	0.00260779857541694\\
0.745613582304799	0.00232964863364224\\
0.899943700679824	0.00197349213154801\\
1.06715083195142	0.0015699729164322\\
1.23394732163963	0.00116859788730362\\
1.40679660101579	0.00077167684364812\\
1.58779317224012	0.000392351548955964\\
1.77698477022245	4.83517400467923e-05\\
1.93263515320482	-0.000188164775833143\\
2.05855625240157	-0.000346666068655715\\
2.19705178572113	-0.000486324935310273\\
2.34489060807404	-0.000595627695978931\\
2.49748483291217	-0.00066690365728107\\
2.65124315064379	-0.0006988490546968\\
2.80395156120509	-0.000694804001126813\\
2.95481121804899	-0.000660337441606461\\
3.10475265069429	-0.000601033917502111\\
3.25369441013513	-0.00052254953609107\\
3.40160877229764	-0.000430405690492823\\
3.54833422671333	-0.000329896622374977\\
3.69337957138205	-0.000226096991086832\\
3.83566103092266	-0.000123900739513789\\
3.97298214009595	-2.81306427406701e-05\\
4.10048272783418	5.5910416356983e-05\\
4.2151311967687	0.00012592521826009\\
4.35837208841213	0.000204298424187268\\
4.52543204319736	0.000280894837513526\\
4.72199769153315	0.000348366423204584\\
4.97929366236527	0.000397805061246115\\
5.18207053143531	0.00040628428004793\\
5.38484740050535	0.000390116652237687\\
5.57244461508394	0.000356050379776085\\
5.7525659440315	0.000309127421290063\\
5.92764994265099	0.000253228785382282\\
6.09841743904801	0.000191872827812056\\
6.26465113237884	0.000128357555193946\\
6.42517397557998	6.58275466677929e-05\\
6.57730262068593	7.38902516045775e-06\\
6.71473645990409	-4.33581500137987e-05\\
6.84589888915889	-8.89810706836057e-05\\
7.00342056069896	-0.000139027052873765\\
7.1829085752079	-0.000188398356059791\\
7.38584553116686	-0.000232915361417326\\
7.61980367396072	-0.000267979322995201\\
7.92153565891401	-0.000286745639360501\\
8.15964669564255	-0.000281370177235227\\
8.39775773237109	-0.000260078630576165\\
8.61845219343352	-0.000228322000035125\\
8.82978686583184	-0.000189324135974881\\
9.03453542231297	-0.000145687965011081\\
9.23324206291973	-9.98271436410379e-05\\
9.4250027683414	-5.40469546929078e-05\\
9.60710571603504	-1.06349855942986e-05\\
9.77301490517159	2.78046279609392e-05\\
9.92836471369826	6.20732019312365e-05\\
10.1115273063333	9.95082541712363e-05\\
10.3207857380744	0.000137355486904156\\
10.5565194885201	0.000172585777976788\\
10.8244021987855	0.000202016659915702\\
11.1455729293125	0.000221672408790372\\
11.4671092558459	0.000224565381654044\\
11.7886455823793	0.000212116157420825\\
12.0705025206859	0.000190421728865263\\
12.338241105806	0.000162332586457953\\
12.5989469574928	0.000129727187139395\\
12.854635630129	9.43524497458497e-05\\
13.1052722855428	5.78894137146767e-05\\
13.348739191266	2.20361702356143e-05\\
13.5791433493141	-1.13021412233971e-05\\
13.7950699688249	-4.13081578973614e-05\\
14.0336093147311	-7.23963825933512e-05\\
14.3161821124093	-0.000105650258721645\\
14.643636001985	-0.000138466675150176\\
15.0279392334862	-0.000168337916890567\\
15.5165374339546	-0.000192492714086825\\
15.9676658662819	-0.000201856265052708\\
16.4187942986091	-0.000200543761929997\\
16.8367280257115	-0.0001917688082581\\
17.2538906475726	-0.000177758074850968\\
17.6849671573728	-0.000159762567152237\\
18.146893201513	-0.000138560416556115\\
18.6785447941452	-0.000113968813037258\\
19.1851829139022	-9.20251408925479e-05\\
19.6918210336593	-7.25482676913559e-05\\
20.2112583885931	-5.5600600841758e-05\\
20.755594772506	-4.11684399107916e-05\\
21.2158585153598	-3.14326031793821e-05\\
21.6761222582137	-2.36849342881517e-05\\
22.109941561054	-1.79381304568496e-05\\
22.5252820119604	-1.36181737803495e-05\\
22.9341137685184	-1.02958798407498e-05\\
23.342398600116	-7.72594176034108e-06\\
23.7541779439895	-5.73995521297925e-06\\
24.1726540956997	-4.21302349880298e-06\\
24.6006628908652	-3.04850490246409e-06\\
25.0409308878338	-2.1697955094716e-06\\
25.4962468002845	-1.51536398983384e-06\\
25.9695999236272	-1.03548313496879e-06\\
26.4643124905998	-6.89925688412711e-07\\
26.9841833346068	-4.46243317496398e-07\\
27.5336562722658	-2.78411956794146e-07\\
28.1180218348194	-1.65704856690997e-07\\
28.7436423216585	-9.16808061506311e-08\\
29.4181105183013	-4.31454735408111e-08\\
30.1499232862591	-8.80357845674182e-09\\
30.3671705142261	-1.64549371683279e-10\\
};
\addplot [color=black,solid,forget plot]
  table[row sep=crcr]{%
18.5777082618857	-0.001\\
18.5777082618857	0.0035\\
};
\end{axis}
\end{tikzpicture}%
\end{minipage}
\caption{A typical fifth excited state ($n=5$) demonstrating the location of $\RDM$ (see \cref{def:rdm}). We have omitted the plot of the potential $V(r)$. To the left of $\RDM$, $F(r)$ exhibits oscillatory behavior. To the right of $\RDM$, $F(r)$ exhibits exponentially decaying behavior.}
\label{fig:RDM}
\end{figure}

Solving for one of the Einstein-Klein-Gordon spherically symmetric static states on the computer is conceptually easy but computationally intensive. If we fix $\Upsilon$ and $\omega$, then to solve for the static state of order $n$ we proceed as follows: Fix a value for $V_0$ consistent with $k(0)<0$ and then begin varying $F_0$. Increasing $F_0$ means that we get more oscillations before reaching $\RDM$ and decreasing $F_0$ means fewer oscillations. Thus to get a ground state or a particular excited state we can only consider values for $F_0$ in a particular range. In that range, there is only one value of $F_0$ such that $F(r)$ exponentially decays for all $r>\RDM$. For all other values of $F_0$, the solution includes the exponential growth term from \cref{eq:ekgFsol} which dominates as $r\to\infty$, and we have agreed to ignore these solutions. Once we have found (to within the desired accuracy) the correct value of $F_0$, we then consider the fact that we probably do not have $V_\infty=0$. We then vary $V_0$ and begin again. By varying $V_0$ in an outer loop and $F_0$ in an inner loop, we can find a solution $(\omega;M,V,F)$ of order $n$.

Next we consider the effect of changing $\omega$. Recall that physically, $\Upsilon$ is to be regarded as a fundamental constant of nature. On the other hand, $\omega$ is a parameter we can vary. Once we have chosen $\omega$, there is then a unique ground state, first excited state, second excited state, etc. We find experimentally using the exact system \eqref{eq:ekgode} that for values of $\omega$ close to $\Upsilon$ we are in the low-field, non-relativistic limit, but as $\omega$ decreases, the mass of the state increases, $\RDM$ decreases, and we leave the low-field, non-relativistic limit. This is illustrated in \cref{fig:lowfieldlimit} for the ground state. We repeat again that the static states live in a 3-parameter space, where two of the parameters (e.g. $\Upsilon,\omega$) are continuous and one (e.g. $n$) is discrete.

\begin{figure}[p]
\centering
\hrule
$\omega=99$\includegraphics[width=0.89\textwidth]{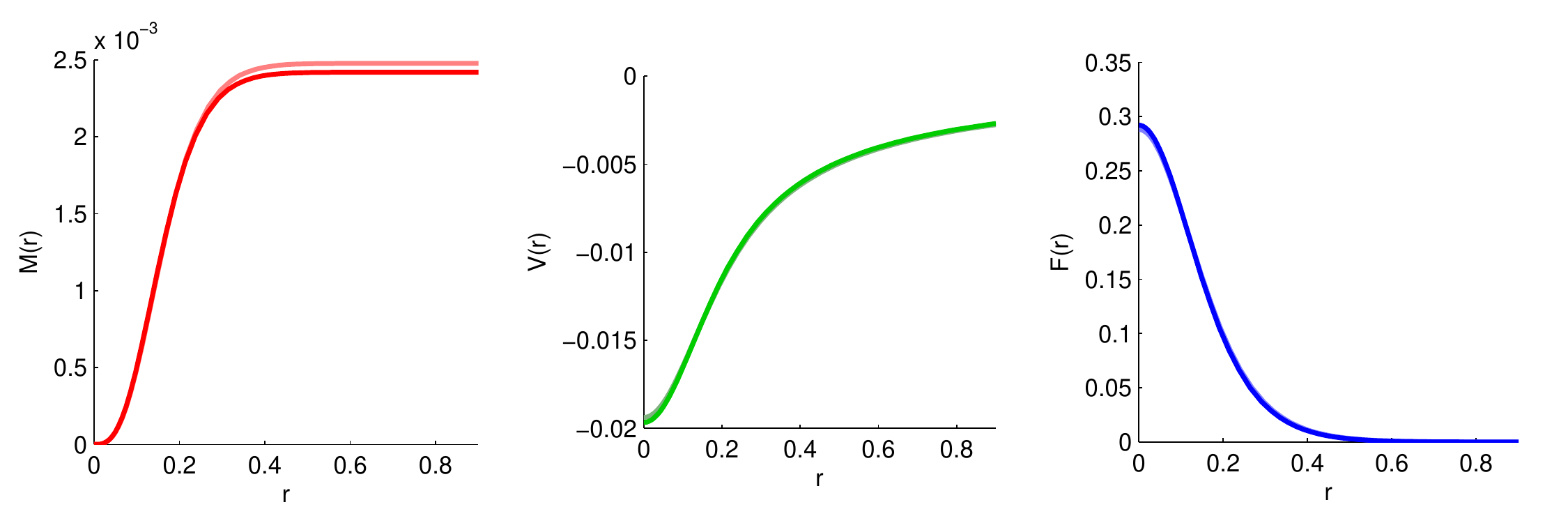}
\vspace{1pt}
\hrule
$\omega=95$\includegraphics[width=0.89\textwidth]{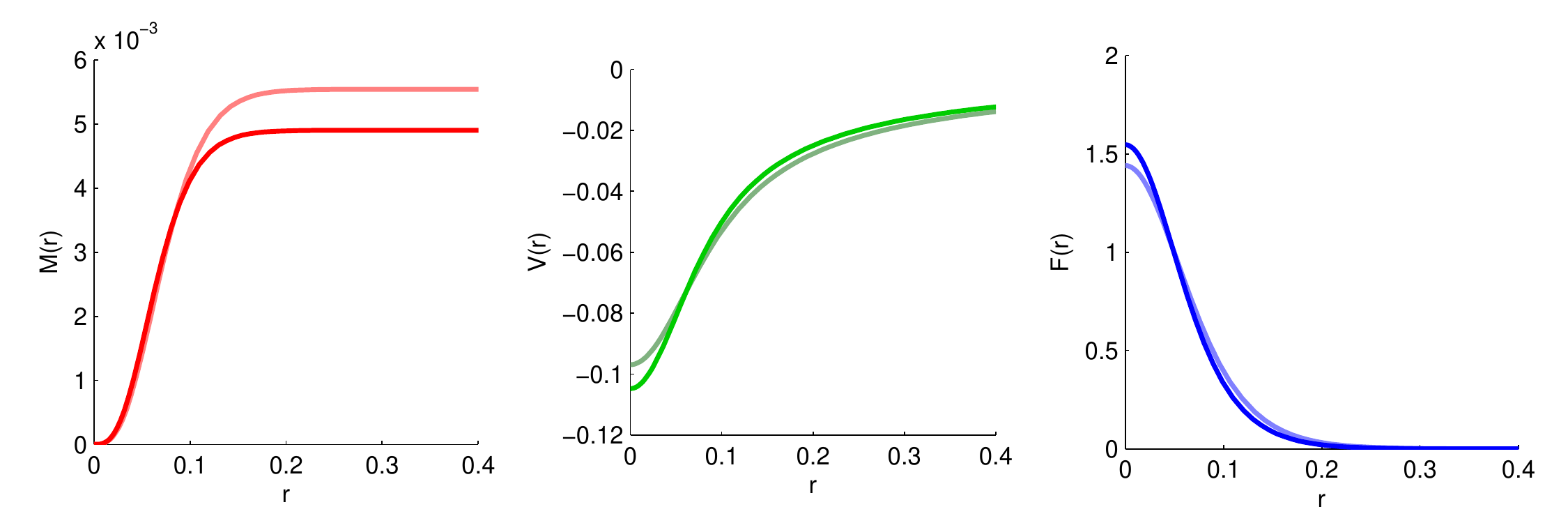}
\vspace{1pt}
\hrule
$\omega=90$\includegraphics[width=0.89\textwidth]{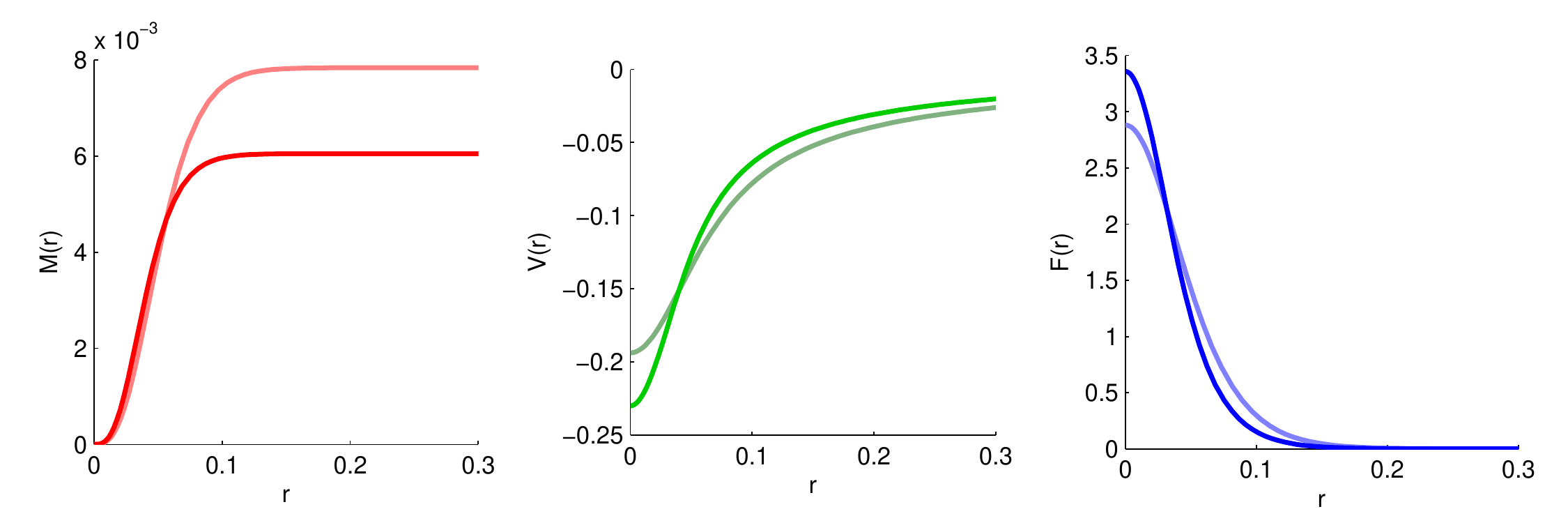}
\vspace{1pt}
\hrule
$\omega=85$\includegraphics[width=0.89\textwidth]{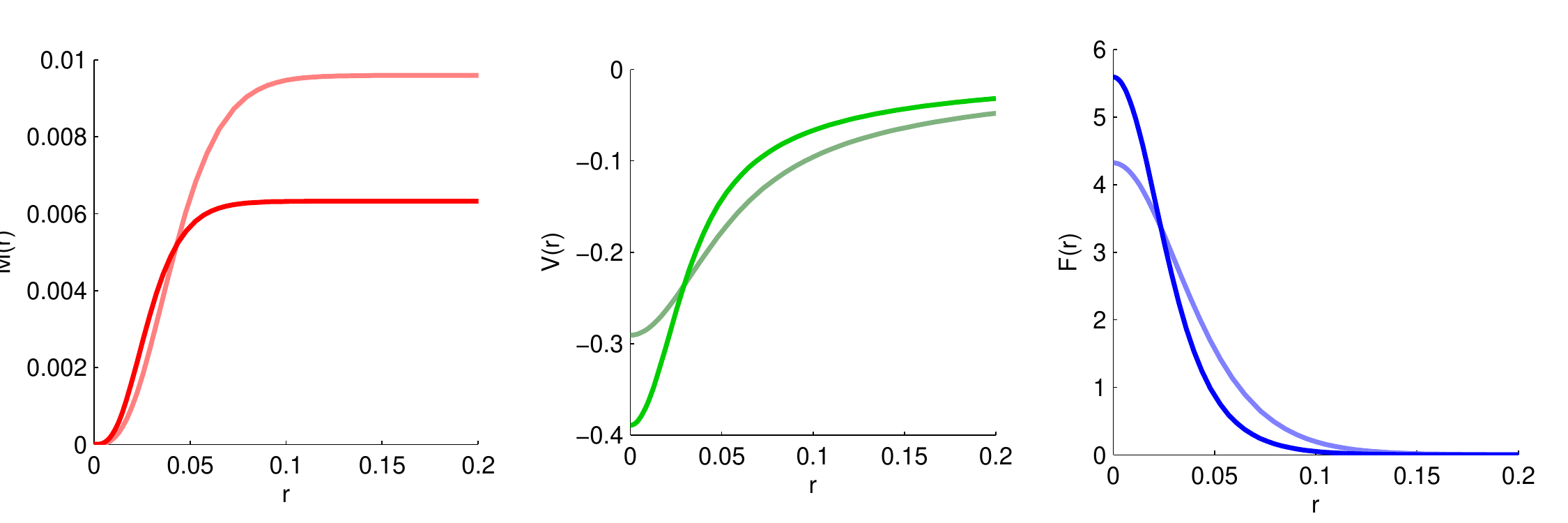}
\vspace{1pt}
\hrule
\caption{Solving for the ground state with $\Upsilon=100$ and $\omega=99,95,90,85$. As $\omega$ decreases we leave the low-field, non-relativistic limit. The more solid graphs are the solutions obtained using the exact Einstein-Klein-Gordon system \eqref{eq:ekgode}, and the fainter graphs are the solutions obtained using the low-field, non-relativistic limit Poisson-Schr\"odinger system \eqref{eq:psode}.}
\label{fig:lowfieldlimit}
\end{figure}

In the next two sections we discuss properties of the Poisson-Schr\"odinger static states. The reader should keep in mind that all these properties are valid for the Einstein-Klein-Gordon static states provided they are in the low-field, non-relativistic regime.

\section{Scalings of the Poisson-Schr\"odinger Static States} \label{sec:scalings}
For the Poisson-Schr\"odinger system, it is a remarkable fact that up to scalings, there is a unique state for each $n\geq 0$. This statement is made precise in the following theorem.
\begin{theorem} \label{thm:psscalings}
Suppose $(\omega;M,V,F)$ solves the Poisson-Schr\"odinger system \eqref{eq:psode} for a particular value of $\Upsilon$. Let $\alpha,\beta>0$. Then $(\bar{\omega};\bar{M},\bar{V},\bar{F})$ defined by
\begin{subequations} \label{eq:psscalings}
\begin{align}
\bar{r} &= \alpha^{-1}\beta^{-1} r \\
\bar{M} &= \alpha\beta^{-3} M \\
\bar{V} &= \alpha^2\beta^{-2} V \\
\bar{F} &= \alpha^2 F \\
\bar{\Upsilon} &= \beta^2\Upsilon \label{eq:mu} \\
(\bar{\Upsilon}-\bar{\omega}) &= \alpha^2(\Upsilon-\omega)
\end{align}
\end{subequations}
is also a solution of \eqref{eq:psode} with $\Upsilon$ replaced by $\bar{\Upsilon}$.
\end{theorem}
\begin{proof}
With the above definitions,
\begin{align*}
\bar{M}(\bar{r}) &= \alpha\beta^{-3} M(\alpha\beta\bar{r}) \\
\bar{V}(\bar{r}) &= \alpha^2\beta^{-2}V(\alpha\beta\bar{r}) \\
\bar{F}(\bar{r}) &= \alpha^2 F(\alpha\beta\bar{r}).
\end{align*}
Therefore
\begin{align*}
\bar{M}_{\bar{r}} &= \alpha\beta^{-3}M_r(\alpha\beta\bar{r})\cdot\alpha\beta \\
&= \alpha^2\beta^{-2}M_r(\alpha\beta\bar{r}) \\
&= \alpha^2\beta^{-2}\cdot 8\pi(\alpha\beta\bar{r})^2F(\alpha\beta\bar{r})^2 \\
&= 8\pi\bar{r}^2(\alpha^2F(\alpha\beta\bar{r})^2) \\
&= 8\pi\bar{r}^2\bar{F}^2
\end{align*}
and we see that $\bar{M}(\bar{r})$ satisfies \cref{eq:psodeM}. Similarly
\begin{align*}
\bar{V}_{\bar{r}} &= \alpha^2\beta^{-2}V_r(\alpha\beta\bar{r})\cdot \alpha\beta \\
&= \alpha^3\beta^{-1}V_r(\alpha\beta\bar{r}) \\
&= \alpha^3\beta^{-1}\cdot\frac{M(\alpha\beta\bar{r})}{(\alpha\beta\bar{r})^2} \\
&= \frac{\alpha\beta^{-3}M(\alpha\beta\bar{r})}{\bar{r}^2} \\
&= \frac{\bar{M}}{\bar{r}^2}
\end{align*}
and we see that $\bar{V}(\bar{r})$ satisfies \cref{eq:psodeV}. Finally
\begin{align*}
\frac{1}{2\bar{\Upsilon}}\left(\bar{F}_{\bar{r}\bar{r}} + \frac{2}{\bar{r}}\bar{F}_{\bar{r}}\right) &= \frac{1}{2\beta^2\Upsilon}\left(\alpha^2F_{rr}(\alpha\beta\bar{r})\cdot (\alpha\beta)^2 + \frac{2}{\bar{r}}\alpha^2F_r(\alpha\beta\bar{r})\cdot(\alpha\beta)\right) \\
&= \alpha^4\cdot\frac{1}{2\Upsilon}\left(F_{rr}(\alpha\beta\bar{r}) + \frac{2}{\alpha\beta\bar{r}}F_r(\alpha\beta\bar{r})\right) \\
&= \alpha^4(\Upsilon-\omega+\Upsilon V(\alpha\beta\bar{r}))F(\alpha\beta\bar{r}) \\
&= (\alpha^2(\Upsilon-\omega)+\beta^2\Upsilon\alpha^2\beta^{-2}V(\alpha\beta\bar{r}))\cdot\alpha^2F(\alpha\beta\bar{r}) \\
&= (\bar{\Upsilon}-\bar{\omega}+\bar{\Upsilon}\bar{V})\bar{F}
\end{align*}
and we see that $\bar{F}(\bar{r})$ satisfies \cref{eq:psodeF}.
\end{proof}

Note that $\beta$ effectively controls the fundamental constant of nature $\Upsilon$ and $\alpha$ can be used to control the ``peakedness'' of the solution. States which are more compact have higher mass; states which are more spread out have lower mass. In \cref{fig:scalings} we show the effect of scaling a ground state using values of $\alpha$ near $1$. It follows from these scalings that having fixed $\Upsilon$, there is a one-parameter family of static states of order $n$. Once we fix a further characteristic of the solution, a static state is uniquely determined for each $n$, giving us a sequence of static states. We call this ``fixing a scaling'' or ``imposing a scaling condition''.

\begin{figure}[hbt]
\centering

\begin{tikzpicture}
\begin{axis}[
	width=4.9in, height=2in,
	scale only axis,
	xlabel=$r$, ylabel=$F(r)$,
	y label style={rotate=-90},
	xtick pos=left, ytick pos=left,
	xmin=0,xmax=8000,
	ymin=-3e-8, ymax=1.2e-7,
]
\addplot[magenta] table[x expr=\thisrow{r}/1.1,y expr=(1.1)^2*\thisrow{F}]{1.data};
\addplot[magenta] table[x expr=\thisrow{r}/1.2,y expr=(1.2)^2*\thisrow{F}]{1.data};
\addplot[magenta] table[x expr=\thisrow{r}/1.3,y expr=(1.3)^2*\thisrow{F}]{1.data};
\addplot[magenta] table[x expr=\thisrow{r}/1.4,y expr=(1.4)^2*\thisrow{F}]{1.data};
\addplot[magenta] table[x expr=\thisrow{r}/1.5,y expr=(1.5)^2*\thisrow{F}]{1.data};
\addplot[blue,very thick] table[x=r,y=F]{1.data};
\end{axis}
\end{tikzpicture}

\caption{Five scalings of the $n=1$ excited state from \cref{fig:1} (shown in blue) using $\alpha=1.1,1.2,1.3,1.4,1.5$ in the scaling formulas \eqref{eq:psscalings}. Here we only show the dark matter profile function $F(r)$.}
\label{fig:scalings}

\end{figure}
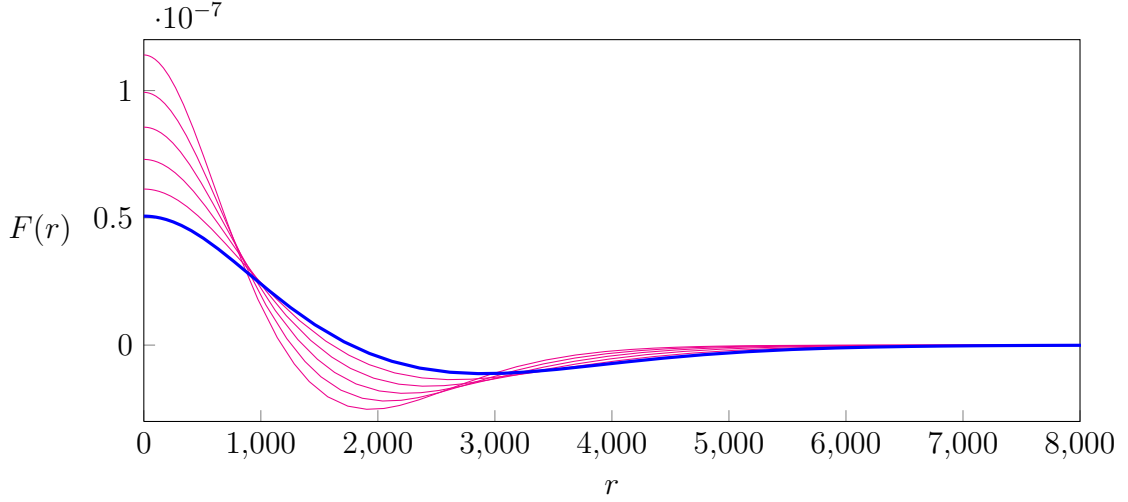

\section{Approximations of the Poisson-Schr\"odinger Static States} \label{sec:approximations}
We can approximate a Poisson-Schr\"odinger static state around $\RDM$ by using a linear approximation for $V(r)$:
\begin{equation} \label{eq:Vapprox}
\begin{split}
V(r) &\approx V(\RDM) + V'(\RDM)(r-\RDM) \\
&= V(\RDM) + \frac{M(\RDM)}{\RDM^2}(r-\RDM)
\end{split}
\end{equation}
for $r\approx\RDM$. Substituting this approximation for $V$ into \cref{eq:psodeF} and using \cref{eq:psrdmdef}, we find that for $r\approx\RDM$, $F(r)$ approximately solves the differential equation
\begin{equation*}
F_{rr} + \frac{2}{r}F_r = \frac{2\Upsilon^2M(\RDM)}{\RDM^2}(r-\RDM)F.
\end{equation*}
Multiplying through by $r$, we can write this differential equation as
\begin{equation} \label{eq:airyforrF}
(rF)_{rr} = \frac{2\Upsilon^2M(\RDM)}{\RDM^2}(r-\RDM)(rF).
\end{equation}
Compare this to the Airy differential equation $y''=xy$ whose general solution is a linear combination of the two standard functions $\Ai(x)$ and $\Bi(x)$. Both $\Ai(x)$ and $\Bi(x)$ oscillate for negative $x$ and have exponential behavior for positive $x$. The function $\Ai(x)$ is shown in \cref{fig:airy}. It exponentially decreases to zero for $x>0$, whereas the function $\Bi(x)$ (not shown) exponentially increases for $x>0$.

\begin{figure}[hbt]
\centering
\begin{tikzpicture}
\begin{axis}[
	width=5in, height=1.5in,
	scale only axis,
	xmin=-15, xmax=5,
	ymin=-0.8, ymax=0.8,
	xtick={-12,-8,-4,0,4},
	axis lines=middle,
]
\addplot[color=black,solid,very thick,samples=500] gnuplot[id=airy,domain=-15:5] {airy(x)};
\addplot[color=brown,dashed,very thick,samples=500] gnuplot[id=airyapprox,domain=-15:-.000001] {sin(.666666666*(-x)**1.5+(pi/4))/(-x)**.25/sqrt(pi)};
\end{axis}
\end{tikzpicture}
\caption{In black, the standard solution $\Ai(x)$ to the Airy differential equation $y''=xy$. In brown and dashed, the approximation \eqref{eq:airyapprox} to $\Ai(x)$ for $x<0$.}
\label{fig:airy}
\end{figure}
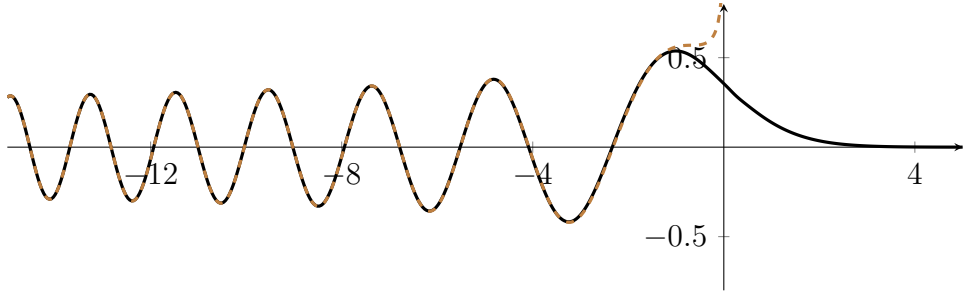

Looking at \cref{eq:airyforrF}, we see that for $r\approx\RDM$, $rF$ approximately solves an equation of the form $y''=a(x-x_0)y$. The general solution to this differential equation is $A\cdot\Ai(a^{1/3}(x-x_0))+B\cdot\Bi(a^{1/3}(x-x_0))$. Since $F$ exponentially decreases for $r>\RDM$, we must have, therefore,
\begin{equation} \label{eq:Fapproxprop}
F(r) \varpropto \frac{1}{r} \Ai\left[\left(\frac{2\Upsilon^2M(\RDM)}{\RDM^2}\right)^{1/3}(r-\RDM)\right]
\end{equation}
for $r\approx\RDM$. If we choose the proportionality constant in \cref{eq:Fapproxprop} so that the approximation has the correct value at $r=\RDM$, then we get
\begin{equation} \label{eq:Fapprox}
F(r) \approx F(\RDM)\frac{\RDM}{r} \frac{1}{\Ai(0)} \Ai\left[\left(\frac{2\Upsilon^2M(\RDM)}{\RDM^2}\right)^{1/3}(r-\RDM)\right]
\end{equation}
for $r\approx\RDM$.

This expression contains the four numbers $\Upsilon$, $\RDM$, $M(\RDM)$, and $F(\RDM)$. These cannot all be independently chosen, for the static states effectively live in a 3-parameter space. Thus, we should be able to find some relation between these four numbers. In fact, we will argue that
\begin{equation} \label{eq:Fscaling}
\abs{F(\RDM)} \approx 2^{-17/12}\Ai(0)\cdot\Upsilon^{1/6}M(\RDM)^{7/12}\RDM^{-17/12} \qquad\text{for large $n$}
\end{equation}
where
\begin{equation}
2^{-17/12}\Ai(0) \approx \num{0.133}
\end{equation}

Our goal is to derive \cref{eq:Fscaling} beginning with the approximation \eqref{eq:Fapprox}. We will utilize a well-known approximation for the Airy function $\Ai(x)$ for negative inputs:
\begin{equation} \label{eq:airyapprox}
\Ai(-x) \approx \frac{\sin\left(\frac23 x^{3/2}+\frac{\pi}{4}\right)}{\sqrt{\pi}x^{1/4}} \qquad\text{for $x>0$.}
\end{equation}
This approximation can be seen in \cref{fig:airy}. It is extremely good for $x>1$. We will also use the following definite integral, which is easily obtained:
\begin{equation} \label{eq:defint}
\int_0^L \frac{\sin^2\left(\frac23x^{3/2}+\frac{\pi}{4}\right)}{x^{1/2}}\,dx = L^{1/2} + O(1).
\end{equation}

Using \cref{eq:psodeM,eq:Fapprox}, we have
\begin{equation*}
\begin{split}
M(\RDM) &= \int_0^{\RDM} 8\pi r^2F^2\,dr \\
&\approx \int_0^{\RDM} 8\pi\frac{F(\RDM)^2\RDM^2}{\Ai(0)^2}\Ai\left[-\left(\frac{2\Upsilon^2 M(\RDM)}{\RDM^2}\right)^{1/3}(\RDM-r)\right]^2\,dr \\
&= \int_0^{\RDM} 8\pi\frac{F(\RDM)^2\RDM^2}{\Ai(0)^2}\Ai\left[-\left(\frac{2\Upsilon^2 M(\RDM)}{\RDM^2}\right)^{1/3}r\right]^2\,dr. \\
\end{split}
\end{equation*}
We will temporarily abbreviate $\RDM=R$, $M(\RDM)=M$, $F(\RDM)=F$, $\Ai(0)=A$. Making the subsitution $x=\left(\frac{2\Upsilon^2 M}{R^2}\right)^{1/3}r$ and using the approximation \eqref{eq:airyapprox}, we obtain
\begin{equation*}
\begin{split}
M &\approx 8\pi\frac{F^2R^2}{A^2}\cdot\left(\frac{2\Upsilon^2 M}{R^2}\right)^{-1/3} \int_0^{(2\Upsilon^2MR)^{1/3}} \Ai(-x)^2\,dx \\
&\approx \frac{2^{8/3}}{A^2}\Upsilon^{-2/3}F^2R^{8/3}M^{-1/3} \int_0^{[2\Upsilon^2MR]^{1/3}} \frac{\sin^2\left(\frac23 x^{3/2}+\frac{\pi}{4}\right)}{x^{1/2}}\,dx. \\
\end{split}
\end{equation*}
Using the definite integral \eqref{eq:defint}, we obtain
\begin{equation*}
\begin{split}
M &\approx \frac{2^{8/3}}{A^2}\Upsilon^{-2/3}F^2R^{8/3}M^{-1/3}\left[(2\Upsilon^2MR)^{1/6} + O(1)\right] \\
&\approx \frac{2^{17/6}}{A^2} \Upsilon^{-1/3}F^2R^{17/6}M^{-1/6}.
\end{split}
\end{equation*}
(We dropped the $O(1)$ term because the other term dominates for large $n$.) Rearranging, we get
\begin{equation*}
F^2 \approx 2^{-17/6}A^2\Upsilon^{1/3}M^{7/6}R^{-17/6}.
\end{equation*}
Taking square roots gives us \cref{eq:Fscaling}, which completes the argument.

\begin{figure}[hbt]
\centering

\begin{tikzpicture}
\begin{axis}[
	width=5in, height=1.4in,
	scale only axis,
	xmin=0, xmax=800,
	ymin=0.998, ymax=1.2,
	xtick pos=left, ytick pos=left,
	xlabel
]
\addplot[only marks,blue,mark size=1pt] table[x=n,y=ratio]{Fscaling.data};
\end{axis}
\end{tikzpicture}

\caption{The ratio given in \cref{eq:Fscalingratio} plotted as a function of $n$ ($n=0,1,\ldots,800$). This graph shows that the approximation \eqref{eq:Fscaling} is very good for large $n$ and not too bad for small $n$.}
\label{fig:Fscaling}
\end{figure}
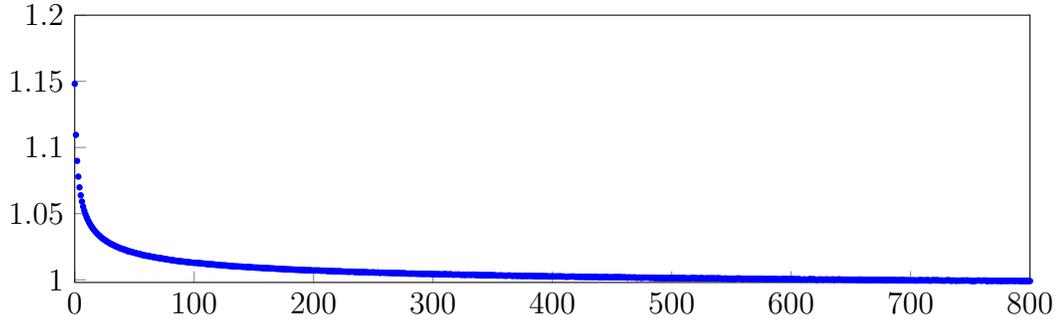

One might wonder how good the approximation \eqref{eq:Fscaling} is for small $n$. In \cref{fig:Fscaling} we have graphed the ratio
\begin{equation} \label{eq:Fscalingratio}
\frac{\abs{F(\RDM)}}{2^{-17/12}\Ai(0)\cdot\Upsilon^{1/6}M(\RDM)^{7/12}\RDM^{-17/12}}
\end{equation}
as a function of $n$. For the ground state ($n=0$), the ratio is approximately $\num{1.15}$, which is reasonably close to $1$, and for large $n$ the ratio is extremely close to $1$.

Combining \cref{eq:Fapprox,eq:Fscaling} gives the following approximation for the $F(r)$ function of a static state of order $n$. (By convention $F(0)$ is positive so the factor $(-1)^n$ is necessary.)
\begin{equation} \label{eq:Fapprox2}
F(r) \approx (-1)^n2^{-17/12}\Upsilon^{2/12}M(\RDM)^{7/12}\RDM^{-5/12}\cdot\frac{1}{r}\Ai\left[\left(\frac{2\Upsilon^2M(\RDM)}{\RDM^2}\right)^{1/3}(r-\RDM)\right]
\end{equation}
This approximation is shown in \cref{fig:0,fig:1,fig:2,fig:10,fig:100}.

\pgfplotsset{
	axisoptions/.style={width=5in,
		height=1.5in,
		scale only axis,
		y label style={rotate=-90},
		ytick pos=left,
		extra x tick labels={},
		extra x tick style={grid=major,tickwidth=0pt},
	},
	plotoptions/.style={thick},
	axisoptionsM/.style={ylabel=$M(r)$,
		xtick={\empty},
		scaled x ticks=false,
		scaled y ticks=true,
		},
	axisoptionsV/.style={ylabel=$V(r)$,
		xtick={\empty},
		scaled x ticks=false,
		scaled y ticks=true,
		},
	axisoptionsF/.style={ylabel=$F(r)$,
		xtick={\empty},
		scaled x ticks=false,
		scaled y ticks=true,
	},
	axisoptionsv/.style={ylabel=$v(r)$,
		xlabel={$r$},
		xtick pos=left,
		scaled y ticks=true,
		},
	plotoptionsM/.style={color=red},
	plotoptionsV/.style={color=black!50!green},
	plotoptionsF/.style={color=blue},
	plotoptionsv/.style={color=orange}
}
\pgfplotsset{
	axisoptions0/.style={
		xmin=0,xmax=8800,
		extra x ticks=2452,
		height=1.3in,
	},
}

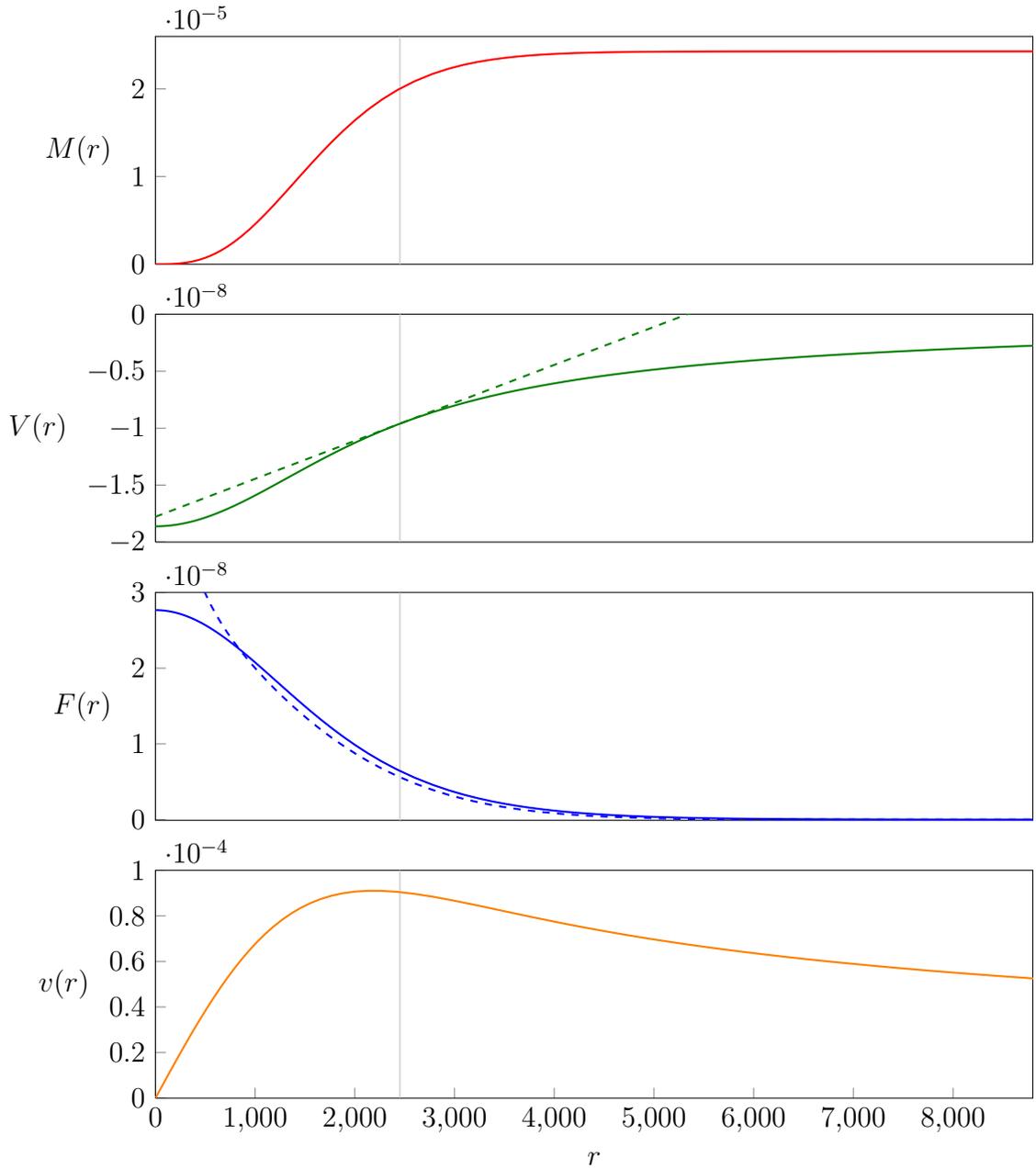
\begin{figure}[p]
\centering

\begin{tikzpicture}
\begin{axis}[
	name=M,
	axisoptions,
	axisoptionsM,
	axisoptions0,
	ymin=-1e-8, ymax=2.6e-5,
]
\addplot[plotoptions,plotoptionsM] table[x=r,y=M]{0.data};
\end{axis}

\begin{axis}[
	name=V,
	at=(M.below south west),
	anchor=above north west,
	axisoptions,
	axisoptionsV,
	axisoptions0,
	ymin=-2e-8, ymax=1e-11,
]
\addplot[plotoptions,plotoptionsV] table[x=r,y=V]{0.data};
\addplot[plotoptions,plotoptionsV,dashed] table[x=r,y=Va]{0.data};
\end{axis}

\begin{axis}[
	name=F,
	at=(V.below south west),
	anchor=above north west,
	axisoptions,
	axisoptionsF,
	axisoptions0,
	ymin=-1e-10, ymax=3e-8,
]
\addplot[plotoptions,plotoptionsF] table[x=r,y=F]{0.data};
\addplot[plotoptions,plotoptionsF,dashed] table[x=r,y=Fa]{0.data};
\end{axis}

\begin{axis}[
	name=v,
	at=(F.below south west),
	anchor=above north west,
	axisoptions,
	axisoptionsv,
	axisoptions0,
	ymin=-1e-8, ymax=1e-4,
]
\addplot[plotoptions,plotoptionsv] table[x=r,y=v]{0.data};
\end{axis}
\end{tikzpicture}

\caption{A ground state ($n=0$) with $\Upsilon=\SI{10}{\lightyear^{-1}}$. Units are geometrized and given in light-years. The first three functions graphed are the solutions $M(r)$, $V(r)$, $F(r)$ to \cref{eq:psode}. They are, respectively, the total mass profile, the Newtonian potential, and the dark matter profile. The fourth function is a rotation curve $v(r)$ calculated in accordance with \cref{eq:rotationcurvecalc}. The vertical gray line marks the location of $\RDM$. The dashed curves show the approximations described in \cref{sec:approximations} in \cref{eq:Vapprox,eq:Fapprox2}.}
\label{fig:0}

\end{figure}

\pgfplotsset{
	axisoptions0/.style={
		xmin=0,xmax=12000,
		extra x ticks=4384,
	},
}

\begin{figure}[p]
\centering

\begin{tikzpicture}
\begin{axis}[
	name=M,
	axisoptions,
	axisoptionsM,
	axisoptions0,
	ymin=-1e-9, ymax=7.7e-5,
]
\addplot[plotoptions,plotoptionsM] table[x=r,y=M]{1.data};
\end{axis}

\begin{axis}[
	name=V,
	at=(M.below south west),
	anchor=above north west,
	axisoptions,
	axisoptionsV,
	axisoptions0,
	ymin=-4.5e-8, ymax=1e-12,
]
\addplot[plotoptions,plotoptionsV] table[x=r,y=V]{1.data};
\addplot[plotoptions,plotoptionsV,dashed] table[x=r,y=Va]{1.data};
\end{axis}

\begin{axis}[
	name=F,
	at=(V.below south west),
	anchor=above north west,
	axisoptions,
	axisoptionsF,
	axisoptions0,
	ymin=-2e-8, ymax=6e-8,
]
\addplot[plotoptions,plotoptionsF] table[x=r,y=F]{1.data};
\addplot[plotoptions,plotoptionsF,dashed] table[x=r,y=Fa]{1.data};
\end{axis}

\begin{axis}[
	name=v,
	at=(F.below south west),
	anchor=above north west,
	axisoptions,
	axisoptionsv,
	axisoptions0,
	ymin=-1e-8, ymax=1.5e-4,
]
\addplot[plotoptions,plotoptionsv] table[x=r,y=v]{1.data};
\end{axis}
\end{tikzpicture}

\caption{An excited state ($n=1$). See the caption to \cref{fig:0} for a fuller description.}
\label{fig:1}

\end{figure}
\pgfplotsset{
	axisoptions0/.style={
		xmin=0,xmax=11900,
		extra x ticks=5973,
	},
}

\begin{figure}[p]
\centering

\begin{tikzpicture}
\begin{axis}[
	name=M,
	axisoptions,
	axisoptionsM,
	axisoptions0,
	ymin=-1e-9, ymax=1.4e-4,
]
\addplot[plotoptions,plotoptionsM] table[x=r,y=M]{2.data};
\end{axis}

\begin{axis}[
	name=V,
	at=(M.below south west),
	anchor=above north west,
	axisoptions,
	axisoptionsV,
	axisoptions0,
	ymin=-6e-8, ymax=1e-12,
]
\addplot[plotoptions,plotoptionsV] table[x=r,y=V]{2.data};
\addplot[plotoptions,plotoptionsV,dashed] table[x=r,y=Va]{2.data};
\end{axis}

\begin{axis}[
	name=F,
	at=(V.below south west),
	anchor=above north west,
	axisoptions,
	axisoptionsF,
	axisoptions0,
	ymin=-2e-8, ymax=8e-8,
]
\addplot[plotoptions,plotoptionsF] table[x=r,y=F]{2.data};
\addplot[plotoptions,plotoptionsF,dashed] table[x=r,y=Fa]{2.data};
\end{axis}

\begin{axis}[
	name=v,
	at=(F.below south west),
	anchor=above north west,
	axisoptions,
	axisoptionsv,
	axisoptions0,
	ymin=-1e-8, ymax=1.5e-4,
]
\addplot[plotoptions,plotoptionsv] table[x=r,y=v]{2.data};
\end{axis}
\end{tikzpicture}

\caption{An excited state ($n=2$). See the caption to \cref{fig:0} for a fuller description.}
\label{fig:2}

\end{figure}
\pgfplotsset{
	axisoptions0/.style={
		xmin=0,xmax=22000,
		extra x ticks=15061,
	},
}

\begin{figure}[p]
\centering

\begin{tikzpicture}
\begin{axis}[
	name=M,
	axisoptions,
	axisoptionsM,
	axisoptions0,
	ymin=-1e-9, ymax=9e-4,
]
\addplot[plotoptions,plotoptionsM] table[x=r,y=M]{10.data};
\end{axis}

\begin{axis}[
	name=V,
	at=(M.below south west),
	anchor=above north west,
	axisoptions,
	axisoptionsV,
	axisoptions0,
	ymin=-2e-7, ymax=1e-12,
]
\addplot[plotoptions,plotoptionsV] table[x=r,y=V]{10.data};
\addplot[plotoptions,plotoptionsV,dashed] table[x=r,y=Va]{10.data};
\end{axis}

\begin{axis}[
	name=F,
	at=(V.below south west),
	anchor=above north west,
	axisoptions,
	axisoptionsF,
	axisoptions0,
	ymin=-6e-8, ymax=2e-7,
]
\addplot[plotoptions,plotoptionsF] table[x=r,y=F]{10.data};
\addplot[plotoptions,plotoptionsF,dashed] table[x=r,y=Fa]{10.data};
\end{axis}

\begin{axis}[
	name=v,
	at=(F.below south west),
	anchor=above north west,
	axisoptions,
	axisoptionsv,
	axisoptions0,
	ymin=-1e-8, ymax=2.6e-4,
]
\addplot[plotoptions,plotoptionsv] table[x=r,y=v]{10.data};
\end{axis}
\end{tikzpicture}

\caption{An excited state ($n=10$). See the caption to \cref{fig:0} for a fuller description.}
\label{fig:10}

\end{figure}
\pgfplotsset{
	axisoptions0/.style={
		xmin=0,xmax=73000,
		extra x ticks=67742,
	},
}

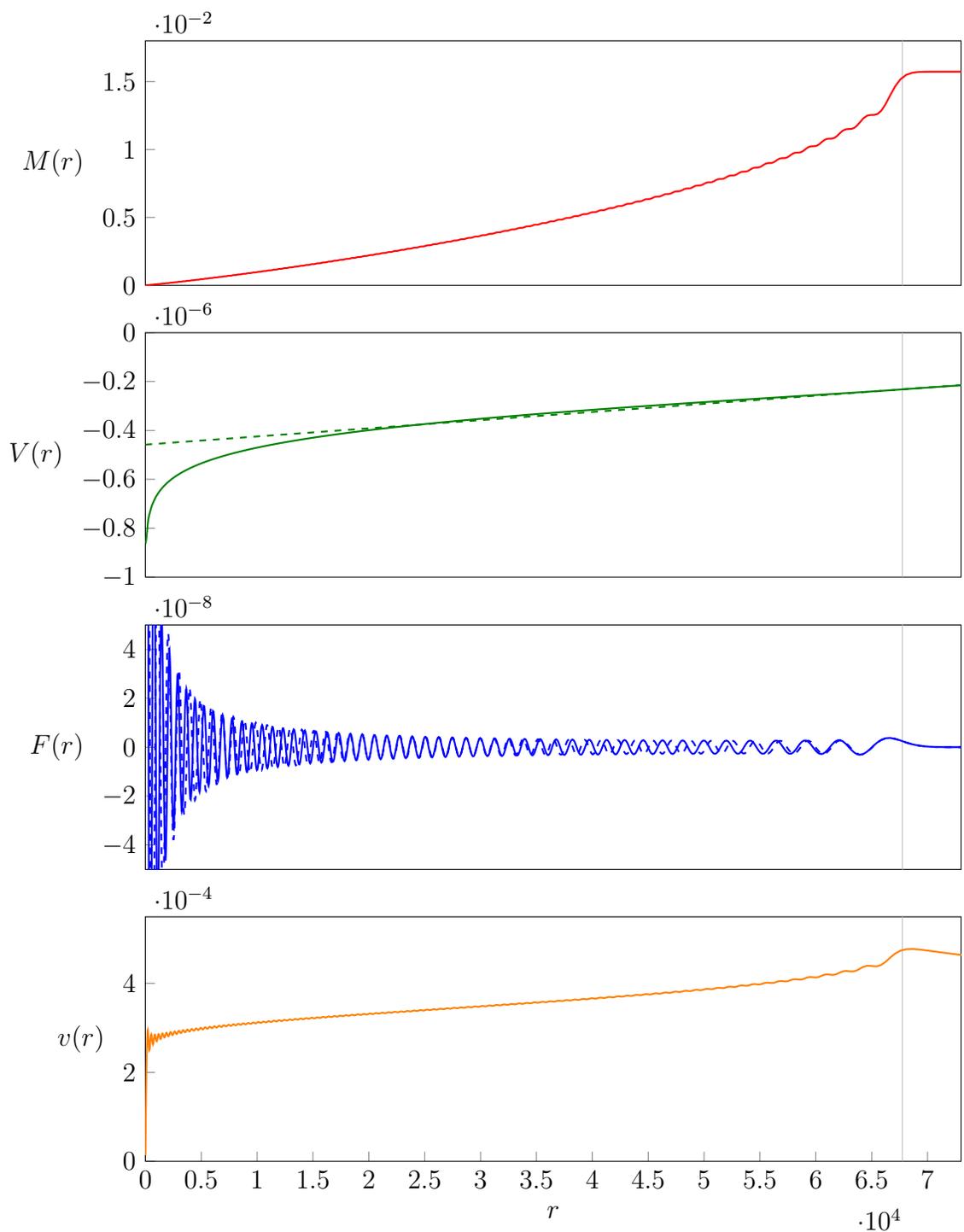
\begin{figure}[p]
\centering

\begin{tikzpicture}
\begin{axis}[
	name=M,
	axisoptions,
	axisoptionsM,
	axisoptions0,
	ymin=-1e-6, ymax=1.8e-2,
]
\addplot[plotoptions,plotoptionsM] table[x=r,y=M]{100.data};
\end{axis}

\begin{axis}[
	name=V,
	at=(M.below south west),
	anchor=above north west,
	axisoptions,
	axisoptionsV,
	axisoptions0,
	ymin=-1e-6, ymax=1e-11,
]
\addplot[plotoptions,plotoptionsV] table[x=r,y=V]{100.data};
\addplot[plotoptions,plotoptionsV,dashed] table[x=r,y=Va]{100.data};
\end{axis}

\begin{axis}[
	name=F,
	at=(V.below south west),
	anchor=above north west,
	axisoptions,
	axisoptionsF,
	axisoptions0,
	ymin=-5e-8, ymax=5e-8,
]
\addplot[plotoptions,plotoptionsF] table[x=r,y=F]{100.data};
\addplot[plotoptions,plotoptionsF,dashed] table[x=r,y=Fa]{100.data};
\end{axis}

\begin{axis}[
	name=v,
	at=(F.below south west),
	anchor=above north west,
	axisoptions,
	axisoptionsv,
	axisoptions0,
	ymin=-1e-9, ymax=5.5e-4,
]
\addplot[plotoptions,plotoptionsv] table[x=r,y=v]{100.data};
\end{axis}

\end{tikzpicture}

\caption{An excited state ($n=100$). See the caption to \cref{fig:0} for a fuller description. The top and bottom of the $F(r)$ function have been cropped out to show more detail.}
\label{fig:100}

\end{figure}

\chapter{Wave Dark Matter and The Baryonic Tully-Fisher Relation} \label{chap:tullyfisher}

In this chapter we describe some possible applications of the mathematical results from the previous chapter to theoretical astrophysics.

\section{Introduction} \label{sec:tfintro}
In \cref{sec:wdm} we introduced the theory of wave dark matter whose underlying equations were the Einstein-Klein-Gordon equations \eqref{eqr:ekg}. In this chapter for definiteness we fix a value for the fundamental constant $\Upsilon$. We use $\Upsilon=\SI{10}{\lightyear^{-1}}$, which is consistent with constraints from other studies \cite{hu00,bray13a,suarez14} of wave dark matter.

We are interested in studying dark matter in disk galaxies and, for reasons outlined in \cref{sec:dmbtfr}, possible connections to the baryonic Tully-Fisher relation. Although in disk galaxies the regular matter mostly lies in a disk, the dark matter is thought to lie in a ``halo'' which is at least approximately spherically symmetric. Sin proposed in \cite{sin94} that each galactic dark matter halo corresponds (in a first approximation) to one of the spherically symmetric static states described in \cref{chap:ekgps}. If we take this suggestion as a starting point, a natural concomitant idea is that the baryonic Tully-Fisher relation might arise because a corresponding ``Tully-Fisher-like'' relation holds for the wave dark matter halos. This Tully-Fisher-like relation would hold because of the nature of wave dark matter. The fundamental question we ask in this chapter is whether the equations of wave dark matter have the potential to give rise to such a Tully-Fisher-like relation (see \cref{q:tfscalingsimprecise}). We answer this question in the affirmative and exhibit two boundary conditions, which could be physical conditions imposed at the edge of each wave dark matter halo, that give rise to just such a Tully-Fisher-like relation. The boundary conditions are as follows:
\begin{description}[style=multiline,leftmargin=1.5cm]
\item[BC1:] Fixing a length scale at the outer edge of halos implies a Tully-Fisher-like relation with slope $4$.
\item[BC2:] Fixing the dark matter density at the outer edge of halos implies a Tully-Fisher-like relation with slope $3.4$.
\end{description}
In the rest of this chapter we work with the static states of the Poisson-Schr\"odinger system as models for dark halos. We can then freely use the scalings described in \cref{thm:psscalings} provided we are careful to stay in the low-field, non-relativistic limit. Our results will then be valid for the Einstein-Klein-Gordon system as well.

\section{Scaling Conditions}
We remind the reader of the discussion at the end of \cref{sec:scalings} regarding the scalings of the static states. With $\Upsilon$ fixed, one more condition is sufficient to ``fix a scaling'' and give a unique sequence $n=0,1,2,\ldots$ of static states. We ask:
\begin{question} \label{q:tfscalingsimprecise}
Are there are any scaling conditions for which the sequence of static states $n=0,1,2,\ldots$ obeys a Tully-Fisher-like relation
\begin{equation}
\frac{M}{v^x} = \text{constant}
\end{equation}
for some $3\leq x\leq 4$?
\end{question}
We need to make this question more precise, as it is unclear above what $M$ and $v$ are, exactly, given a particular static state. Astrophysicists attempting to calibrate the baryonic Tully-Fisher relation are faced with a similar problem. In that case, $M$ is an estimate for the total baryonic mass. For $v$, some take it to be the maximum observed velocity; some take it to be an average velocity computed in some well-defined way; some take it to be the velocity observed at a particular (arbitrarily chosen) radius. It is well-known that the rotation curves of many spiral galaxies are somewhat flat, and in this case any of these three choices will produce similar values for $v$.

We have chosen to take $M=M(\RDM)$ and $v=v(\RDM)$. This latter quantity is the circular velocity which would be observed for objects orbiting at $r=\RDM$, and is computed in the standard Newtonian way:
\begin{equation} \label{eq:rotationcurvecalc}
\text{accel.} = \frac{v(r)^2}{r} = \frac{M(r)}{r^2} \qquad\Longrightarrow\qquad v(r) = \sqrt{\frac{M(r)}{r}}.
\end{equation}
Perhaps a more natural choice for $M$ would be $M_\infty$, i.e. the total mass of the static state. But $M(\RDM)\approx M_\infty$ since $\RDM$ is by definition the ``outer edge'' of the halo (see \cref{def:rdm} and \cref{fig:RDM}), so it does not really matter which we choose. Similarly, another natural choice for $v$ would be $v_{\text{max}}$, the maximum value of the function $v(r)$. But again, by examining \cref{fig:0,fig:1,fig:2,fig:10,fig:100}, one can see that $v(\RDM)\approx v_{\text{max}}$. Thus, the following more precise version of \cref{q:tfscalingsimprecise} seems reasonable:
\begin{question} \label{q:tfscalingsprecise}
Are there are any scaling conditions for which the sequence of static states $n=0,1,2,\ldots$ obeys a Tully-Fisher-like relation
\begin{equation} \label{eq:tfscalingsprecise}
\frac{M(\RDM)}{v(\RDM)^x} = \text{constant}
\end{equation}
for some $3\leq x\leq 4$?
\end{question}

Now, in one sense the answer to \cref{q:tfscalingsprecise} is a trivial ``yes'', for we can just pick any $x$ we like and let \cref{eq:tfscalingsprecise} itself be the scaling condition. However, this is obviously cheating; we are interested in scaling conditions which might correspond to something physical. We will consider a scaling condition to be physical if it fixes some property of the dark matter profile function $F(r)$ either at a particular point (a \emph{local} condition) or in the large (a \emph{global} condition). Orthogonally to the local/global distinction, we will consider scaling conditions which fix a \emph{horizontal} property of $F$ or a \emph{vertical} property of $F$. The first type of scaling condition fixes some type of length, whereas the second type of scaling condition fixes $\abs{F}$ in some way. Looking at \cref{eq:psodeM}, we see that $2F^2$ gives mass density so that fixing a vertical property of $F$ corresponds to fixing the density of the dark matter in some way.

For local conditions, there are really only two natural places to impose a scaling: at $r=0$ or $r=\RDM$. We consider both, with horizontal and vertical scalings. In the category of global conditions, a horizontal scaling and a vertical scaling give two more scaling conditions for a total of six scaling conditions in all. We will discuss all these conditions in \cref{sec:bcs,sec:ccs,sec:gcs} after we make the definitions in \cref{sec:definitions}.

\section{Definitions} \label{sec:definitions}
\begin{definition} \label{def:halflength}
The $F(r)$ function for any static state has an exponentially decreasing ``tail'' lying to the right of $\RDM$. There is some $R>\RDM$ such that $F(R)=\frac12 F(\RDM)$; we refer to $R-\RDM$ as the \emph{halflength} of the state.
\end{definition}

\begin{definition}
For a Poisson-Schr\"odinger static state, the \emph{wavelength} of $F(r)$ at $0\leq r<\RDM$ is given by the expression
\begin{equation*}
\frac{2\pi}{2\Upsilon(\Upsilon-\omega+\Upsilon V(r))}.
\end{equation*}
This definition follows naturally from an examination of \cref{eq:psodeF}.
\end{definition}

\begin{definition} \label{def:omegatrue}
In the solution $f(t,r)=F(r)e^{i\omega t}$ (see \cref{eq:fstaticstate}), the constant $\omega$ is like a frequency but it is not a physical quantity because $t$ is just a coordinate. We define
\begin{equation} \label{eq:omegatrue}
\omegatrue(r)=\omega e^{-V(r)}.
\end{equation}
This \emph{is} a physical quantity---the true frequency of the dark matter that would be measured at a particular value of $r$. The factor $e^{-V(r)}$ comes from the metric \eqref{eq:metric_Phi}.
\end{definition}

\begin{definition}
The \emph{average wavelength} of a static state of order $n$ is defined to be
\begin{equation*}
\frac{\RDM}{n/2+1/4}.
\end{equation*}
The reader is invited to examine \cref{fig:0,fig:1,fig:2,fig:10,fig:100} to see why this is a sensible definition.
\end{definition}

\begin{definition} \label{def:averagedensity}
The \emph{average density} of a Poisson-Schr\"odinger static state is defined to be
\begin{equation}
\frac{M(\RDM)}{\frac43\pi\RDM^3} = \frac{1}{\frac43\pi\RDM^3}\int_0^{\RDM} 4\pi r^2\cdot 2F(r)^2\,dr.
\end{equation}
\end{definition}

\section{Boundary Conditions at \texorpdfstring{$r=\RDM$}{r=RDM}} \label{sec:bcs}
In this section we discuss imposing local scaling conditions at $r=\RDM$, both horizontal and vertical. For obvious reasons we refer to these two scaling conditions as boundary conditions, \textbf{BC1} and \textbf{BC2}. These were given in \cref{sec:tfintro} along with the Tully-Fisher-like relations they imply. We restate them here more precisely:
\begin{description}[style=multiline,leftmargin=1.5cm] \label{bcs}
\item[BC1:] Fixing a length scale at $\RDM$ implies a Tully-Fisher-like relation with slope~4.
\item[BC2:] Fixing $\abs{F(\RDM)}$ implies a Tully-Fisher-like relation with slope $3.4$.
\end{description}
Note that, as we stated above, fixing $\abs{F(\RDM)}$ is equivalent to fixing the dark matter density at $\RDM$.

\subsection{Boundary Condition 1: Fixing a Length Scale}
With regards to \textbf{BC1}, we should explain what ``fixing a length scale'' means. One easy way to understand this concept is in terms of the halflength of static states---see \cref{def:halflength}. Requiring the halflength of all states to be the same is one way to fix a length scale at $\RDM$. Alternatively, we can fix the true oscillation frequency of the dark matter, $\omegatrue(r)$ (see \cref{def:omegatrue}), at a fixed distance outside (or inside) $\RDM$. This condition is explored more extensively in \cite{bray14}. But most precisely, the best way to understand what fixing a length scale means is via the approximate solution \eqref{eq:Fapproxprop} from \cref{chap:ekgps}, which we repeat here. We discovered that for $r\approx\RDM$,
\begin{equation} \label{eqr:Fapproxprop}
F(r) \varpropto \frac{1}{r} \Ai\left[\left(\frac{2\Upsilon^2M(\RDM)}{\RDM^2}\right)^{1/3}(r-\RDM)\right].
\end{equation}

Up to a vertical scaling, the shape of $F(r)$ around $r=\RDM$ is given by the expression on the right side of \eqref{eqr:Fapproxprop}. Fixing a length scale around $r=\RDM$ evidently amounts to fixing the value of the constant
\begin{equation}
\frac{2\Upsilon^2M(\RDM)}{\RDM^2}.
\end{equation}
This explains why fixing a length scale is basically equivalent to fixing a value for the halflength or to any of a number of other conditions (such as the one involving $\omegatrue(r)$ stated above) which fix some length around $\RDM$. All such conditions are fixing a particular horizontal scaling of the expression in \cref{eqr:Fapproxprop}.

Now, the rotation curve calculation \eqref{eq:rotationcurvecalc} shows that
\begin{equation} \label{eq:tfconstantatRDM}
\frac{\RDM^2}{M(\RDM)} = \frac{M(\RDM)}{v(\RDM)^4}.
\end{equation}
Since $\Upsilon$ is a constant, we see that fixing a length scale is equivalent to fixing the constant in \cref{eq:tfconstantatRDM}. This explains why the scaling condition \textbf{BC1} implies a Tully-Fisher-like relation for the static states with slope $x=4$.

\subsection{Boundary Condition 2: Fixing the Density}
The truth of \textbf{BC2} follows from \cref{eq:Fscaling} from \cref{chap:ekgps}, which we repeat here:
\begin{equation} \label{eqr:Fscaling}
\abs{F(\RDM)} \approx 2^{-17/12}\Ai(0)\cdot\Upsilon^{1/6}M(\RDM)^{7/12}\RDM^{-17/12} \qquad\text{for large $n$.}
\end{equation}
From \cref{eqr:Fscaling} we see that fixing $\abs{F(\RDM)}$ amounts to fixing $M(\RDM)^7\RDM^{-17}$. By \cref{eq:rotationcurvecalc}, $\RDM = M(\RDM)v(\RDM)^{-2}$, so fixing $M(\RDM)^7\RDM^{-17}$ is equivalent to fixing $M(\RDM)^{-10}v(\RDM)^{34}=(M(\RDM)/v(\RDM)^{3.4})^{-10}$, i.e. equivalent to fixing
\begin{equation*}
\frac{M(\RDM)}{v(\RDM)^{3.4}}.
\end{equation*}
Thus \textbf{BC2} implies a Tully-Fisher-like relation with slope $x=3.4$.

\subsection{Numerical Evidence}

In addition to the theoretical arguments we have just adduced for the truth of \textbf{BC1} and \textbf{BC2}, we can also check numerically that \textbf{BC1} and \textbf{BC2} do give Tully-Fisher-like relations with the stated slopes. We used a computer to derive particular static state solutions for $n=0,1,\ldots,800$ and then scaled these solutions in accordance with \cref{eq:psscalings} so that, first, they all had the same halflength, and second, so they all had the same value for $\abs{F(\RDM)}$. (The precise values for the halflength and $\abs{F(\RDM)}$ were chosen so that the resulting static states would be of reasonable sizes to be representing galactic halos.) The results are shown in \cref{fig:rdmscalings}, in which we have plotted $M(\RDM)$ vs. $v(\RDM)$ in log-log space. For large $n$, the static states form a line with the stated slope. (For small $n$, the states deviate slightly from the correct slope because the approximation \cref{eq:Fapprox} is not as good for those $n$.)

\begin{figure}[hbt]
\centering

\begin{tikzpicture}[baseline]
\begin{loglogaxis}[
	title={\bfseries BC1:} Fixed Halflength,
	width=2in, height=3in,
	scale only axis,
	xlabel=$v(\RDM)$ (km/s), ylabel=$M(\RDM)$ ($M_\Sun$),
	xtick pos=left, ytick pos=left,
	xmin=10,xmax=1000,
	ymin=1e7, ymax=1e13,
]
\addplot[only marks,blue] table[x=v,y=M]{fix_halflength.data};
\addplot[black,dashed,domain=10:1000] {3e11*(x/3e5)^4/1.5607e-13};
\end{loglogaxis}
\end{tikzpicture}%
\hspace{0.5cm}%
\begin{tikzpicture}[baseline]
\begin{loglogaxis}[
	title={\bfseries BC2:} Fixed $\abs{F(\RDM)}$,
	width=2in, height=3in,
	scale only axis,
	xlabel=$v(\RDM)$ (km/s), ylabel=$M(\RDM)$ ($M_\Sun$),
	xtick pos=left, ytick pos=left,
	xmin=10,xmax=1000,
	ymin=1e7, ymax=1e13,
]
\addplot[only marks,blue] table[x=v,y=M]{fix_F_RDM.data};
\addplot[black,dashed,domain=10:1000] {1.9e22*(x/3e5)^3.4};
\end{loglogaxis}
\end{tikzpicture}

\caption{Numerical evidence for \textbf{BC1} and \textbf{BC2} (see \cpageref{bcs}). At left, the static states $n=0,1,\ldots,800$ with $\Upsilon=\SI{10}{\lightyear^{-1}}$ and a fixed exponential decay halflength of $\SI{884}{\lightyear}$. The dashed line has slope $4$. At right, the static states $n=0,1,\ldots,800$ with $\Upsilon=\SI{10}{\lightyear^{-1}}$ and a fixed value for $\abs{F(\RDM)}$ of $\num{2.473e-9}$. The dashed line has slope $3.4$. In these plots we have chosen not to adhere to our convention of using geometric units to facilitate comparison to \cref{fig:btfr}.}
\label{fig:rdmscalings}

\end{figure}

\section{Central Conditions at \texorpdfstring{$r=0$}{r=0}} \label{sec:ccs}

In this section we discuss imposing local scaling conditions at $r=0$, both horizontal and vertical. For obvious reasons we refer to these two scaling conditions as central conditions, \textbf{CC1} and \textbf{CC2}.
\begin{description}[style=multiline,leftmargin=1.5cm]
\item[CC1:] Fixing the wavelength of $F$ at $r=0$ (variant: fixing the location of the first node for $n\geq 1$) does not give a Tully-Fisher-like relation.
\item[CC2:] Fixing $F(0)$, i.e. fixing the density of the dark matter at the center of halos, does not give a Tully-Fisher-like relation.
\end{description}
We do not provide theoretical arguments here for the truth of \textbf{CC1} and \textbf{CC2}, only numerical evidence. See \cref{fig:r0scalings}.

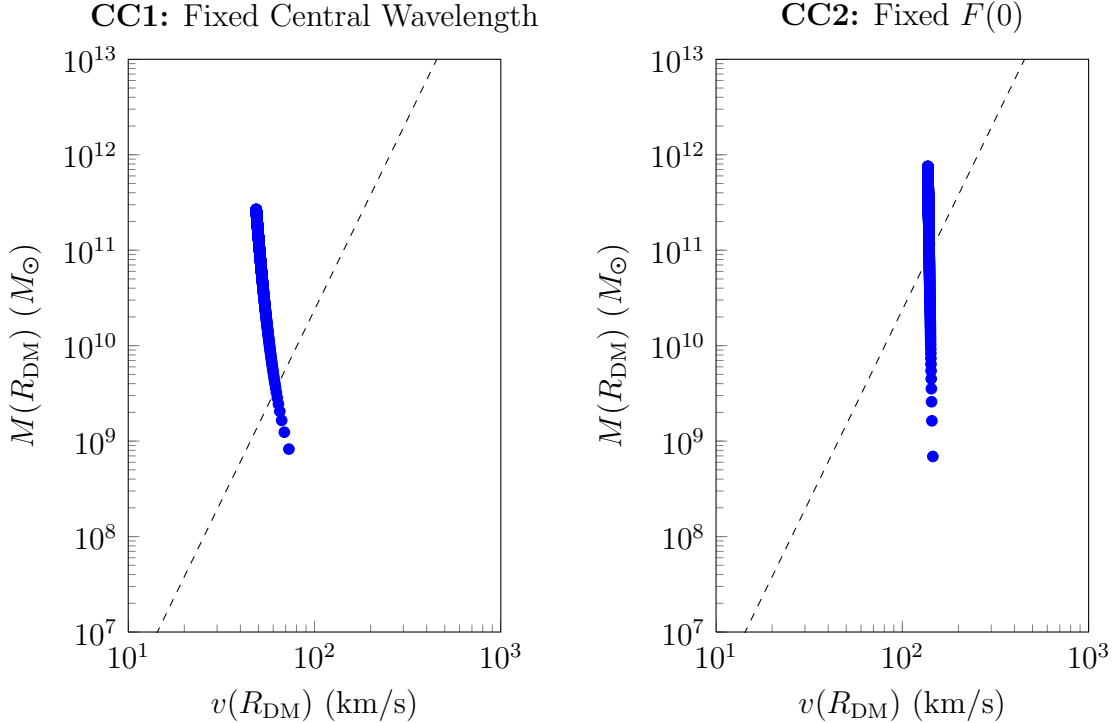
\begin{figure}[hbt]
\centering

\begin{tikzpicture}[baseline]
\begin{loglogaxis}[
	title={{\bfseries CC1:} Fixed Central Wavelength},
	width=1.95in, height=3in,
	scale only axis,
	xlabel=$v(\RDM)$ (km/s), ylabel=$M(\RDM)$ ($M_\Sun$),
	xtick pos=left, ytick pos=left,
	xmin=10,xmax=1000,
	ymin=1e7, ymax=1e13,
]
\addplot[only marks,blue] table[x=v,y=M]{fix_central_wavelength.data};
\addplot[black,dashed,domain=10:1000] {3e11*(x/3e5)^4/1.5607e-13};
\end{loglogaxis}
\end{tikzpicture}%
\hspace{0.5cm}%
\begin{tikzpicture}[baseline]
\begin{loglogaxis}[
	title={\bfseries CC2:} Fixed $F(0)$,
	width=1.95in, height=3in,
	scale only axis,
	xlabel=$v(\RDM)$ (km/s), ylabel=$M(\RDM)$ ($M_\Sun$),
	xtick pos=left, ytick pos=left,
	xmin=10,xmax=1000,
	ymin=1e7, ymax=1e13,
]
\addplot[only marks,blue] table[x=v,y=M]{fix_F0.data};
\addplot[black,dashed,domain=10:1000] {3e11*(x/3e5)^4/1.5607e-13};
\end{loglogaxis}
\end{tikzpicture}

\caption{Numerical demonstration that \textbf{CC1} and \textbf{CC2} do not give Tully-Fisher-like relations. At left, the static states $n=0,1,\ldots,800$ with $\Upsilon=\SI{10}{\lightyear^{-1}}$ and a fixed central wavelength of $\SI{1500}{\lightyear}$. (Alternatively, fixing the location of the first node for $n\geq 1$ gives a qualitatively similar plot.) At right, the static states $n=0,1,\ldots,800$ with $\Upsilon=\SI{10}{\lightyear^{-1}}$ and a fixed value for $F(0)$ of $F(0)=\num{8e-7}$. The dashed lines have slope $4$. In these plots we have chosen not to adhere to our convention of using geometric units to facilitate comparison to \cref{fig:btfr}.}
\label{fig:r0scalings}

\end{figure}

\section{Global Conditions} \label{sec:gcs}

In this section we discuss imposing the following two global scaling conditions, the first a horizontal condition, the second a vertical condition.
\begin{description}[style=multiline,leftmargin=1.5cm]
\item[GC1:] Fixing the average wavelength of $F$ does not give a Tully-Fisher-like relation.
\item[GC2:] Fixing the average density of halos gives a Tully-Fisher-like relation with slope $3$. This result is purely a Newtonian gravitational result and does not depend on the nature of wave dark matter.
\end{description}
We do not provide a theoretical argument here for the truth of \textbf{GC1}, only numerical evidence. See \cref{fig:globalscalings}. The truth of \textbf{GC2} follows from a simple Newtonian argument which does not use the specifics of the wave dark matter model or the static states. Fixing the average density (see \cref{def:averagedensity}) amounts to fixing $M(\RDM)/\RDM^3$. The standard rotation curve calculation \eqref{eq:rotationcurvecalc} gives $\RDM = M(\RDM)/v(\RDM)^2$. Substituting in, we see that fixing the average density amounts to fixing $v(\RDM)^6/M(\RDM)^2$, or equivalently, fixing
\begin{equation*}
\frac{M(\RDM)}{v(\RDM)^3}.
\end{equation*}
Thus we get a Tully-Fisher-like relation with slope $x=3$. The argument just given is basically the same as the one given to support a Tully-Fisher relation with slope 3 in the $\Lambda$CDM paradigm (see \cite{white97,gurovich10}), with $\RDM$ replaced by, for example, $r_{200}$, the radius enlosing a region with average density equal to 200 times the critical density of the universe.

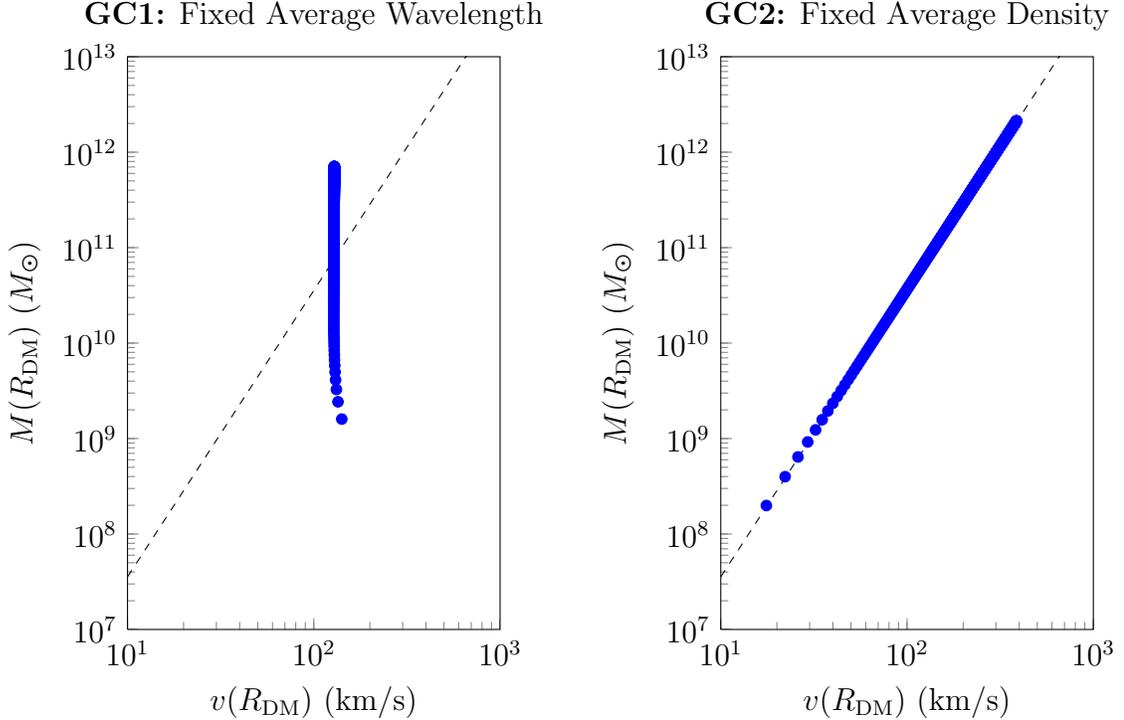
\begin{figure}[hbt]
\centering

\begin{tikzpicture}[baseline]
\begin{loglogaxis}[
	title={{\bfseries GC1:} Fixed Average Wavelength},
	width=1.95in, height=3in,
	scale only axis,
	xlabel=$v(\RDM)$ (km/s), ylabel=$M(\RDM)$ ($M_\Sun$),
	xtick pos=left, ytick pos=left,
	xmin=10,xmax=1000,
	ymin=1e7, ymax=1e13,
]
\addplot[only marks,blue] table[x=v,y=M]{fix_average_wavelength.data};
\addplot[black,dashed,domain=10:1000] {1.5e8*(x/3e5)^3/1.5607e-13};
\end{loglogaxis}
\end{tikzpicture}%
\hspace{0.5cm}%
\begin{tikzpicture}[baseline]
\begin{loglogaxis}[
	title={\bfseries GC2:} Fixed Average Density,
	width=1.95in, height=3in,
	scale only axis,
	xlabel=$v(\RDM)$ (km/s), ylabel=$M(\RDM)$ ($M_\Sun$),
	xtick pos=left, ytick pos=left,
	xmin=10,xmax=1000,
	ymin=1e7, ymax=1e13,
]
\addplot[only marks,blue] table[x=v,y=M]{fix_average_density.data};
\addplot[black,dashed,domain=10:1000] {1.5e8*(x/3e5)^3/1.5607e-13};
\end{loglogaxis}
\end{tikzpicture}

\caption{Numerical demonstration that \textbf{GC1} does not give a Tully-Fisher-like relation and \textbf{GC2} does. At left, the static states $n=0,1,\ldots,800$ with $\Upsilon=\SI{10}{\lightyear^{-1}}$ and a fixed average wavelength of $\SI{1500}{\lightyear}$. At right, the static states $n=0,1,\ldots,800$ with $\Upsilon=\SI{10}{\lightyear^{-1}}$ and a fixed average density of $\num{e-17}$. The dashed lines have slope $3$. In these plots we have chosen not to adhere to our convention of using geometric units to facilitate comparison to \cref{fig:btfr}.}
\label{fig:globalscalings}

\end{figure}

\section{Discussion and Conjectures}

In the previous section we discovered that to get a Tully-Fisher-like relation in the context of wave dark matter we need to impose a boundary condition (\textbf{BC1} or \textbf{BC2}). We could also impose the global condition \textbf{GC2}, but this does not use any of the characteristics of wave dark matter. Because \textbf{GC2} applies to every theory of dark matter, we will ignore it in the discussion which follows.

\subsection{The Slope of the Baryonic Tully-Fisher Relation}
In this section we comment briefly on the connection between a Tully-Fisher-like relation for wave dark matter halos and the observed baryonic Tully-Fisher relation. Define the \emph{baryon fraction} of a galaxy to be the fraction of the galaxy's mass which is baryonic. If the baryon fraction were the same for every galaxy, then the baryonic Tully-Fisher relation would have the same slope as a Tully-Fisher-like relation for dark halos. However, it is observed that the baryon fraction of galaxies decreases with size. Small galaxies contain mostly dark matter, and only very large galaxies contain baryons in sufficient quantities to approach the cosmic baryon fraction, which is around $1/6$. This is known as the \emph{missing baryon problem} (see \cite{mcgaugh10}).

In $M$ vs. $v$ log-log space, when switching from plotting dark mass to baryonic mass, all galaxies shift vertically downward. The missing baryon problem implies that large galaxies shift less than small galaxies. Therefore, if the baryonic Tully-Fisher relation is related to a Tully-Fisher-like relation for dark halos, the former relation should have a steeper slope than the latter.

\subsection{Conjectures}

In this section we describe a testable prediction of the wave dark matter model. Assuming a new idea called ``dark matter saturation'' described below and made precise in \cref{physconj:1}, given the distribution of the baryonic matter in a galaxy, we should be able to compute the distribution of the wave dark matter and hence the total mass distribution, as long as everything is approximately static and spherically symmetric. With the total mass distribution in hand it will be possible to compare to observations---for example, to compute rotation curves. These predictions will be possible once specific stability questions, described below, are answered.

If we assume that dark matter in many galaxies is approximately static and spherically symmetric, then it makes sense to look for static, spherically symmetric solutions to wave dark matter. Referring back to the scaling equations, we see that there is a two-parameter family of static spherically symmetric solutions. These two parameters are $(n,\alpha)$, where $n \geq 0$ is an integer referring to the excited state and $\alpha$ is the scaling factor. Recall that $n=0$ refers to the ground state, and that more generally $n$ is the number of zeros of $F(r)$.

Now suppose that we solve for wave dark matter solutions in the presence of regular matter (which, for our purposes here, we will also assume is spherically symmetric and static). While the regular matter, through gravity, will change each wave dark matter solution, we still expect to find a two parameter $(n,\alpha)$ family of solutions.

One great benefit of our discussion so far is that it removes one of the parameters, namely the continuous parameter $\alpha$. In this chapter, we showed that a boundary condition (\textbf{BC1} or \textbf{BC2}) is roughly what is needed to recover a Tully-Fisher-like relation. To be clear, we have not explained why such a boundary condition should be expected from the theory, just that something close to it seems necessary to be compatible with the observations which make up the baryonic Tully-Fisher relation. In any case, assuming one of these boundary conditions effectively determines $\alpha$, leaving only one parameter free, namely $n$.

What should the value of $n$ be? Given a regular matter distribution, for each $n$ we get a precise wave dark matter distribution which satisfies the above boundary condition. Some of these wave dark matter solutions will be stable and some will be unstable. Numerical results show that ground states of wave dark matter are stable while excited states, without any other matter around, are unstable \cite{lai07}. On the other hand, \cite{lai07} also shows that excited states may be stabilized by the presence of another matter field.  Our conjecture is that the regular, visible, baryonic matter stabilizes wave dark matter in galaxies. In fact, given a regular matter distribution, we conjecture that there exists a largest value of $n$, call it $N$, for which the corresponding wave dark matter solution is stable.  We conjecture that galaxies are described best by choosing $n = N$.

The total mass of the spherically symmetric static states described in this paper increases with $n$ and we expect the same to be true for the distributions of dark matter we are describing now. Thus, setting $n = N$ is consistent with the idea that galaxies are ``dark matter saturated'', meaning that they are holding as much dark matter as possible, subject to the boundary condition above. Since galaxies typically exist in clusters which are mostly made of dark matter, it seems likely that they are regularly bombarded by dark matter, so that it would be natural for them to reach this state of saturation $n=N$.

To make this discussion precise, we need a model for the regular matter.  In order to study stability questions, we need to know how the regular matter distribution changes as the wave dark matter distribution changes, and vice versa. 

For example, a relatively simple way to model regular matter is with another scalar field.  There are others ways to model regular matter which we do not discuss here.  We caution the reader that this second scalar field is only a practical device for approximately modeling the regular baryonic matter---namely the gas, dust, and stars in a galaxy. In no way are we suggesting a second scalar field should exist physically. Furthermore, the parameters of this second scalar field are chosen simply to fit the regular matter distribution of a galaxy as well as possible.

Let $f_1$ exactly model wave dark matter with its fundamental constant of nature $\Upsilon_1$.  Let $f_2$ be a convenient device for approximately modeling the regular baryonic matter consisting of the gas, dust, and stars of a galaxy, where $\Upsilon_2$, which is not a fundamental constant of nature, is chosen as desired to best fit the regular matter.  The action is then  
\begin{equation}
\mathcal{F}(g,f_1,f_2) = \int \left[R_g-2\Lambda - 16\pi\left(\abs{f_1}^2 + \frac{\abs{df_1}^2}{\Upsilon_1^2} + \abs{f_2}^2 + \frac{\abs{df_2}^2}{\Upsilon_2^2} \right)\right]\,dV_g,
\end{equation}
where $\Lambda$ is the cosmological constant and may as well be assumed to be zero for our discussion on the scale of galaxies.  The above action results in the following Euler-Lagrange equations:
\begin{subequations} \label{eq:EL}
\begin{align}
\begin{split}
G+\Lambda g &= 8\pi\Bigg[ \frac{df_1\tensor d\bar{f}_1+d\bar{f}_1\tensor df_1}{\Upsilon_1^2}-\left(\abs{f_1}^2+\frac{\abs{df_1}^2}{\Upsilon_1^2}\right)g \\
& \qquad\quad +\frac{df_2\tensor d\bar{f}_2+d\bar{f}_2\tensor df_2}{\Upsilon_2^2}-\left(\abs{f_2}^2+\frac{\abs{df_2}^2}{\Upsilon_2^2}\right)g \Bigg]
\end{split} \\
\Box f_1 &= \Upsilon_1^2 f_1 \\
\Box f_2 &= \Upsilon_2^2 f_2
\end{align}
\end{subequations}
We approximate the regular matter distribution with a ground state solution for $f_2$. We have two free parameters with which to approximate the given regular matter distribution, namely $\Upsilon_2$ and the ``scaling parameter'' for the ground state solution, which we could call $\alpha_2$. This should allow us to choose two physical characteristics of the regular matter. We choose to specify the total mass $M_b$ and the radius $R_b$ of the regular matter, perhaps defined as that radius within which some fixed percentage of the regular matter is contained.

As described already, we impose the boundary condition \textbf{BC1} or \textbf{BC2} for $f_1$. Then for each choice of $n \ge 0$, we get a solution to the system of equations \eqref{eq:EL} which reduces to a system of ODEs in a manner very similar as before. Some solutions will be stable and some will be unstable. Since we now have the dynamical equations \eqref{eq:EL}, these stability questions are now fairly well defined.  

Hence, for each $M_b$, $R_b$, and $n$, we get a static, spherically symmetric solution to \cref{eq:EL} satisfying the boundary condition \textbf{BC1} or \textbf{BC2}.  
\begin{mathconjecture}
In the low-field, nonrelativistic limit, for each choice of total regular mass $M_b$ and regular matter radius $R_b$, there exists an integer $N \ge 0$ such that static, spherically symmetric solutions to \cref{eq:EL} satisfying the boundary condition \textbf{BC1} or \textbf{BC2} with $n \le N$ are stable and those with $n > N$ are unstable.   
\end{mathconjecture}
If this math conjecture is true, or even if there is just a largest or most massive stable $n$, then there is a natural physics conjecture to make as well.
\begin{physicsconjecture}[``Dark Matter Saturation''] \label{physconj:1}
The dark matter and total matter distributions of most galaxies which are approximately static and spherically symmetric are approximately described by static, spherically symmetric solutions to \cref{eq:EL} satisfying the boundary condition \textbf{BC1} or \textbf{BC2} with $n = N$.
\end{physicsconjecture}
This last conjecture only leaves two parameters open, namely $\Upsilon$, the fundamental constant of nature in the wave dark matter theory, and the length scale (or halflength) from \textbf{BC1} or the boundary dark matter density from \textbf{BC2}. Hence, there are effectively only two parameters left open with which to fit the dark matter and total matter distributions of most of the galaxies in the universe. Therefore the physics conjecture stated above should be a good test of the wave dark matter theory.
\appendix
\chapter{The Einstein-Klein-Gordon-Maxwell Equations} \label{app:action}

In this appendix we derive the Einstein-Klein-Gordon-Maxwell equations from an action using a variational principle. On a spacetime with metric $g$, complex scalar field $f$, and one-form $A$, we define the action
\begin{equation} \label{eqr:ekgmaction}
\SF_U(g,f,A) = \int_U \bigg[R-2\Lambda -16\pi\bigg(\frac{\abs{df}^2}{\Upsilon^2}+\abs{f}^2+\frac14\abs{dA}^2\bigg)\bigg]\,dV.
\end{equation}
Here $U$ is any precompact open set with smooth boundary, so that the integral converges. We will require that $g$, $f$ and $A$ are critical points of the functional $\SF_U$ for all possible choices of $U$ and all variations compactly supported within $U$. In the three sections below we vary each of $g$, $f$, $A$ independently to derive the Einstein equation, the Klein-Gordon equation, and Maxwell's equations.

We will make use of the Hodge star operator $\hodge$ and the codifferential operator $\delta$ which is adjoint to $d$ with respect to the scalar product $\ip{\cdot}{\cdot}$. For brevity we will omit the domain of integration $U$. Because the variations we consider are compactly supported, all boundary terms are zero.

\section{Varying the Metric \texorpdfstring{$g$}{g}}
In the calculation below we will use the following well-known results concerning a variation $\dot{g}=h$ of the metric:
\begin{align*}
\dot{R} &= -\ip{h}{\Ric} + \divergence(\divergence h) - \Box(\trace h) \\
\dot{dV} &= \frac12 (\trace h)\,dV.
\end{align*}
We will also need to know the variations of $\abs{df}^2$ and $\abs{dA}^2$. First, since $g_{\lambda\mu}g^{\mu\nu}=\delta_{\lambda}^\nu$, differentiation gives $g_{\lambda\mu}\dot{g^{\mu\nu}}=-h_{\lambda\mu}g^{\mu\nu}$. Multiplying by $g^{\lambda\rho}$ gives $\dot{g^{\nu\rho}}=-h^{\nu\rho}$. To calculate the variation of $\abs{df}^2$, begin with
\begin{equation*}
\abs{df}^2 = g^{\lambda\mu}f_\lambda\cc{f}_\mu.
\end{equation*}
Differentiating, we get
\begin{equation*}
\dot{\abs{df}^2} = -h^{\lambda\mu}f_\lambda\cc{f}_\mu = -\ip{h}{df\tensor d\cc{f}} = -\ip{h}{d\cc{f}\tensor df}.
\end{equation*}
Therefore
\begin{equation*}
\dot{\abs{df}^2} = -\ip{h}{\frac12(df\tensor d\cc{f} + d\cc{f}\tensor df)}.
\end{equation*}
To calculate the variation of $\abs{dA}^2$, write $F=dA$ so that in coordinates,
\begin{equation*}
\abs{dA}^2 = g^{\lambda\mu}g^{\nu\rho}F_{\lambda\nu}F_{\mu\rho}.
\end{equation*}
Differentiating, we get
\begin{equation*}
\dot{\abs{dA}^2} = -h^{\lambda\mu}g^{\nu\rho}F_{\lambda\nu}F_{\mu\rho} - g^{\lambda\mu}h^{\nu\rho}F_{\lambda\nu}F_{\mu\rho}.
\end{equation*}
The two terms above are equal by the anti-symmetry of $F$. Thus
\begin{equation*}
\dot{\abs{dA}^2} = -2\ip{h}{\mathbf{C}_{13}(dA\tensor dA)}
\end{equation*}
where $\mathbf{C}_{13}(dA\tensor dA)$ denotes the tensor obtained by contracting on the first and third slots of $dA\tensor dA$.

With these results in hand, if we perform the variation in \cref{eqr:ekgmaction}, we get
\begin{equation*}
\begin{split}
0 &= \int \Bigg\{ -\ip{h}{\Ric} + 8\pi\ip*{h}{\frac{df\tensor d\cc{f} + d\cc{f}\tensor df}{\Upsilon^2} + \mathbf{C}_{13}(dA\tensor dA)} \\
&\qquad\qquad + (\trace h)\bigg[\frac{R}{2}-\Lambda -8\pi\bigg(\frac{\abs{df}^2}{\Upsilon^2}+\abs{f}^2+\frac14\abs{dA}^2\bigg)\bigg]\Bigg\}\,dV.
\end{split}
\end{equation*}
We immediately dropped the terms $\divergence(\divergence h)$ and $\Box(\trace h)$ because using the divergence theorem, they give boundary terms which are zero. If we note that $\trace h = \ip{h}{g}$, the equation above takes the form
\begin{equation*}
0 = \int -\ip{h}{Z}\,dV
\end{equation*}
where $Z$ is the $(0,2)$-tensor
\begin{equation*}
\begin{split}
Z &= \Ric-\frac12 Rg + \Lambda g - 8\pi\bigg[\frac{df\tensor d\cc{f} + d\cc{f}\tensor df}{\Upsilon^2} -\left(\frac{\abs{df}^2}{\Upsilon^2}+\abs{f}^2\right)g \\
&\qquad\qquad +\mathbf{C}_{13}(dA\tensor dA) - \frac14\abs{dA}^2g\bigg].
\end{split}
\end{equation*}
This equation must hold for all variations $h$, which implies that $Z=0$. Thus the variational principle has led us to the Einstein equation
\begin{equation} \label{eqr:einstein}
G+\Lambda g = 8\pi T
\end{equation}
(recall the definition $G=\Ric-(1/2)Rg$) where the energy-momentum tensor $T$ is given by
\begin{equation} \label{eq:kgmT}
T = \frac{df\tensor d\cc{f} + d\cc{f}\tensor df}{\Upsilon^2} -\left(\frac{\abs{df}^2}{\Upsilon^2}+\abs{f}^2\right)g +\mathbf{C}_{13}(dA\tensor dA) - \frac14\abs{dA}^2g.
\end{equation}

\section{Varying the Scalar Function \texorpdfstring{$f$}{f}}
It suffices to consider the variation of the integral
\begin{equation*}
\int \big[\abs{df}^2+\Upsilon^2\abs{f}^2\big]\,dV = \int \big[\ip{df}{d\cc{f}} + \Upsilon^2 f\cc{f}\big]\,dV.
\end{equation*}
Performing the variation, we get
\begin{equation*}
\begin{split}
0 &= \int \big[\ip{d\dot{f}}{d\cc{f}} + \ip{df}{d\cc{\dot{f}}} + \Upsilon^2\dot{f}\cc{f} + \Upsilon^2 f\cc{\dot{f}}\big]\,dV \\
&= \int \big[ \ip{\dot{f}}{\delta d\cc{f}} + \ip{\delta df}{\cc{\dot{f}}} + \Upsilon^2\dot{f}\cc{f} + \Upsilon^2 f\cc{\dot{f}}\big]\,dV \\
&= \int \big[ \dot{f}(\delta d+\Upsilon^2)\cc{f} + \cc{\dot{f}}(\delta d + \Upsilon^2)f\big]\,dV.
\end{split}
\end{equation*}
This equation must hold for all variations $\dot{f}$. Since we can vary the real and imaginary parts of $f$ independently, this implies that $(\delta d+\Upsilon^2)f=0$ and $(\delta d + \Upsilon^2)\cc{f}=0$. The second equation is the complex conjugate of the first and is superfluous. The operator $\delta d$ acting on functions is equal to $\delta d+d\delta$, which is the Laplace-de Rham operator, the negative of the Laplace-Beltrami operator (in this case the d'Alembertian) $\Box$. Therefore the variational principle has led us to the Klein-Gordon equation
\begin{equation} \label{eq:kg}
\Box f = \Upsilon^2 f.
\end{equation}

\section{Varying the One-Form \texorpdfstring{$A$}{A}}
It suffices to consider the variation of the integral
\begin{equation*}
\int\abs{dA}^2\,dV = \int dA\wedge\hodge dA.
\end{equation*}
Performing the variation, we get
\begin{equation} \label{eq:Avar1}
0 = \int d\dot{A}\wedge\hodge dA + dA\wedge\hodge d\dot{A}.
\end{equation}
The two summands above are equal because, by the definition of the Hodge star,
\begin{equation*}
d\dot{A}\wedge\hodge dA = dA\wedge\hodge d\dot{A} = \ip{dA}{d\dot{A}}\,dV.
\end{equation*}
Since
\begin{equation*}
d(\dot{A}\wedge\hodge dA) = d\dot{A}\wedge\hodge dA - \dot{A}\wedge d\hodge dA
\end{equation*}
\cref{eq:Avar1} becomes
\begin{equation*}
0 = \int \dot{A}\wedge d\hodge dA.
\end{equation*}
This equation must hold for all variations $\dot{A}$, which implies  that $d\hodge dA=0$, and this equation is equivalent to $\hodge d\hodge dA=0$. Recalling that, up to a sign, $\hodge d\hodge$ is equal to the codifferential operator $\delta$, the equation we derive from varying $A$ is
\begin{equation} \label{eq:maxwell}
\delta dA=0.
\end{equation}
This is half of what are usually referred to as Maxwell's equations. A more conventional way of writing them is to define the Faraday tensor $F=dA$. We automatically have $dF=0$; this equation paired with \cref{eq:maxwell} becomes
\begin{align*}
dF &= 0 \\
\delta F &= 0.
\end{align*}

\section{The Einstein-Klein-Gordon-Maxwell Equations}
Requiring the metric $g$, the complex scalar function $f$, and the one-form $A$ to be critical points of the functional \eqref{eqr:ekgmaction} has led us to the Einstein-Klein-Gordon-Maxwell equations collected in \cref{eqr:einstein,eq:kgmT,eq:kg,eq:maxwell}, which we repeat here:
\begin{subequations}
\begin{align}
G + \Lambda g &= 8\pi T \\
\Box f &= \Upsilon^2 f \\
\delta dA &= 0.
\end{align}
\end{subequations}
where
\begin{equation*}
T = \frac{df\tensor d\cc{f} + d\cc{f}\tensor df}{\Upsilon^2} -\left(\frac{\abs{df}^2}{\Upsilon^2}+\abs{f}^2\right)g +\mathbf{C}_{13}(dA\tensor dA) - \frac14\abs{dA}^2g.
\end{equation*}

\chapter{The Einstein-Klein-Gordon Equations in Spherical Symmetry} \label{app:odederivation}

In this appendix we derive the system of ordinary differential equations for the spherically symmetric static states of wave dark matter. The calculations we carry out here are a special case of calculations carried out in \cite{parrythesis}, but are much simpler because we assume from the beginning that the metric \eqref{eqr:metric} is static.

The spherically symmetric static states are solutions to the Einstein-Klein-Gordon equations
\begin{subequations} \label{eqrr:ekg}
\begin{align}
G &= 8\pi\left( \frac{df\tensor d\bar{f}+d\bar{f}\tensor df}{\Upsilon^2} - \left(\frac{\abs{df}^2}{\Upsilon^2}+\abs{f}^2\right)g\right) \label{eqrr:ekg1} \\
\Box f &= \Upsilon^2 f. \label{eqrr:ekg2}
\end{align}
\end{subequations}
We begin with the spherically symmetric static metric
\begin{equation} \label{eqr:metric}
g = -e^{2V(r)}\,dt^2 + \left(1-\frac{2M(r)}{r}\right)^{-1}\,dr^2 + r^2\,d\theta^2 + r^2\sin^2\theta\,d\phi^2.
\end{equation}
We define
\begin{equation} \label{eqr:Phi}
\Phi(r) = 1-\frac{2M(r)}{r}
\end{equation}
so that the metric \eqref{eqr:metric} can be written as
\begin{equation} \label{eqr:metric_Phi}
g = -e^{2V(r)}\,dt^2 + \Phi(r)^{-1}\,dr^2 + r^2\,d\theta^2 + r^2\sin^2\phi\,d\phi^2.
\end{equation}
\begin{lemma} \label{lem:christoffel}
The nonzero Christoffel symbols of the metric \eqref{eqr:metric_Phi} in the $(t,r,\theta,\phi)$ coordinates are as follows:
\begin{align*}
\Gamma_{tr}^t = \Gamma_{rt}^t &= V_r & \Gamma_{r\theta}^\theta = \Gamma_{\theta r}^\theta &= r^{-1} \\
\Gamma_{tt}^r &= V_re^{2V}\Phi & \Gamma_{\phi\phi}^\theta &= -\sin\theta\cos\theta \\
\Gamma_{rr}^r &= -(1/2)\Phi^{-1}\Phi_r & \Gamma_{r\phi}^\phi = \Gamma_{\phi r}^\phi &= r^{-1} \\
\Gamma_{\theta\theta}^r &= -r\Phi & \Gamma_{\theta\phi}^\phi = \Gamma_{\phi\theta}^\phi &= \cot\theta. \\
\Gamma_{\phi\phi}^r &= -r\sin^2\theta\,\Phi
\end{align*}
\end{lemma}
\begin{proof}
These results follow immediately from short computations using the usual formula for the Christoffel symbols:
\begin{equation*}
\Gamma_{\lambda\mu}^\nu = \frac12 g^{\nu\rho}(g_{\lambda\rho,\mu}+g_{\mu\rho,\lambda}-g_{\lambda\mu,\rho}).\qedhere
\end{equation*}
\end{proof}

\begin{lemma} \label{lem:riccicurvature}
The nonzero components of the Ricci curvature tensor associated with the metric \eqref{eqr:metric_Phi} in the $(t,r,\theta,\phi)$ coordinates are as follows:
\begin{align*}
\Ric_{tt} &= (V_{rr}+V_r^2+2r^{-1}V_r)e^{2V}\Phi + (1/2)V_re^{2V}\Phi_r \\ 
\Ric_{rr} &= -(V_{rr}+V_r^2) - (1/2)V_r\Phi^{-1}\Phi_r - r^{-1}\Phi^{-1}\Phi_r \\
\Ric_{\theta\theta} &= 1-\Phi - (1/2)r\Phi_r - rV_r\Phi \\
\Ric_{\phi\phi} &= \sin^2\theta\,(1-\Phi - (1/2)r\Phi_r - rV_r\Phi).
\end{align*}
\end{lemma}
\begin{proof}
The formula for the components of the Ricci curvature tensor in terms of the metric and the Christoffel symbols is
\begin{equation} \label{eq:riccitensor}
\Ric_{\lambda\mu} = \Gamma_{\lambda\mu,\nu}^\nu - \Gamma_{\lambda\nu,\mu}^\nu + \Gamma_{\lambda\mu}^\nu\Gamma_{\nu\rho}^\rho - \Gamma_{\lambda\nu}^\rho\Gamma_{\mu\rho}^\nu.
\end{equation}
In the calculations below we enclose each of the four terms in \cref{eq:riccitensor} in square braces, omitting all the Christoffel symbols which are zero. Then we make whatever cancellations are possible, substitute in for the remaining symbols, and simplify. We have
\begin{align*}
\Ric_{tt} &= [\Gamma_{tt,r}^r] - [] + [\Gamma_{tt}^r(\Gamma_{rt}^t+\Gamma_{rr}^r+\Gamma_{r\theta}^\theta+\Gamma_{r\phi}^\phi)] - [\Gamma_{tt}^r\Gamma_{tr}^t + \Gamma_{tr}^t\Gamma_{tt}^r] \\
&= \Gamma_{tt,r}^r + \Gamma_{tt}^r\Gamma_{rr}^r + \Gamma_{tt}^r\Gamma_{r\theta}^\theta + \Gamma_{tt}^r\Gamma_{r\phi}^\phi - \Gamma_{tr}^t\Gamma_{tt}^r \\
&= (V_re^{2V}\Phi)_r - (1/2)V_re^{2V}\Phi_r + 2r^{-1}V_re^{2V}\Phi - V_r^2e^{2V}\Phi \\
&= V_{rr}e^{2V}\Phi + 2V_r^2e^{2V}\Phi + V_re^{2V}\Phi_r - (1/2)V_re^{2V}\Phi_r + 2r^{-1}V_re^{2V}\Phi - V_r^2e^{2V}\Phi \\
&= (V_{rr}+V_r^2+2r^{-1}V_r)e^{2V}\Phi + (1/2)V_re^{2V}\Phi_r
\intertext{and}
\Ric_{tr} &= [\Gamma_{tr,t}^t] - [] + [] - [] = 0 \\
\intertext{and}
\Ric_{t\theta} &= [] - [] + [] - [] = 0
\intertext{and}
\Ric_{t\phi} &= [] - [] + [] - [] = 0
\intertext{and}
\Ric_{rr} &= [\Gamma_{rr,r}^r] - [\Gamma_{rt,r}^t + \Gamma_{rr,r}^r + \Gamma_{r\theta,r}^\theta + \Gamma_{r\phi,r}^\phi] + [\Gamma_{rr}^r(\Gamma_{rt}^t+\Gamma_{rr}^r+\Gamma_{r\theta}^\theta+\Gamma_{r\phi}^\phi)] \\
&\qquad - [\Gamma_{rt}^t\Gamma_{rt}^t+\Gamma_{rr}^r\Gamma_{rr}^r+\Gamma_{r\theta}^\theta\Gamma_{r\theta}^\theta+\Gamma_{r\phi}^\phi\Gamma_{r\phi}^\phi] \\
&= -\Gamma_{rt,r}^t - \Gamma_{r\theta,r}^\theta - \Gamma_{r\phi,r}^\phi + \Gamma_{rr}^r\Gamma_{rt}^t + \Gamma_{rr}^r\Gamma_{r\theta}^\theta + \Gamma_{rr}^r\Gamma_{r\phi}^\phi - \Gamma_{rt}^t\Gamma_{rt}^t - \Gamma_{r\theta}^\theta\Gamma_{r\theta}^\theta - \Gamma_{r\phi}^\phi\Gamma_{r\phi}^\phi \\
&= -V_{rr} + 2r^{-2} + (-(1/2)\Phi^{-1}\Phi_r)V_r + (-(1/2)\Phi^{-1}\Phi_r)(2r^{-1}) - V_r^2 - 2r^{-2} \\
&= -(V_{rr}+V_r^2) - (1/2)V_r\Phi^{-1}\Phi_r - r^{-1}\Phi^{-1}\Phi_r
\intertext{and}
\Ric_{r\theta} &= [\Gamma_{r\theta,\theta}^\theta] - [\Gamma_{rt,\theta}^t + \Gamma_{rr,\theta}^r + \Gamma_{r\theta,\theta}^\theta + \Gamma_{r\phi,\theta}^\phi ] + [\Gamma_{r\theta}^\theta\Gamma_{\theta\phi}^\phi] - [\Gamma_{r\phi}^\phi\Gamma_{\theta\phi}^\phi] \\
&= \Gamma_{\theta\phi}^\phi(\Gamma_{r\theta}^\theta-\Gamma_{r\phi}^\phi) \\
&= \cot\theta(r^{-1}-r^{-1}) \\
&= 0
\intertext{and}
\Ric_{r\phi} &= [\Gamma_{r\phi,\phi}^\phi] - [\Gamma_{rt,\phi}^t + \Gamma_{rr,\phi}^r + \Gamma_{r\theta,\phi}^\theta + \Gamma_{r\phi,\phi}^\phi] + [] - [] = 0
\intertext{and}
\Ric_{\theta\theta} &= [\Gamma_{\theta\theta,r}^r] - [\Gamma_{\theta\phi,\theta}^\phi] + [\Gamma_{\theta\theta}^r(\Gamma_{rt}^t+\Gamma_{rr}^r+\Gamma_{r\theta}^\theta+\Gamma_{r\phi}^\phi)] - [\Gamma_{\theta r}^\theta\Gamma_{\theta\theta}^r + \Gamma_{\theta\theta}^r\Gamma_{\theta r}^\theta + \Gamma_{\theta\phi}^\phi\Gamma_{\theta\phi}^\phi] \\
&= \Gamma_{\theta\theta,r}^r - \Gamma_{\theta\phi,\theta}^\phi + \Gamma_{\theta\theta}^r\Gamma_{rt}^t + \Gamma_{\theta\theta}^r\Gamma_{rr}^r + \Gamma_{\theta\theta}^r\Gamma_{r\phi}^\phi - \Gamma_{\theta\theta}^r\Gamma_{\theta r}^\theta - \Gamma_{\theta\phi}^\phi\Gamma_{\theta\phi}^\phi \\
&= (-r\Phi)_r - (\cot\theta)_\theta + (-r\Phi)V_r + (-r\Phi)(-(1/2)\Phi^{-1}\Phi_r) \\
&\qquad + (-r\Phi)r^{-1} - (-r\Phi)r^{-1} - \cot^2\theta \\
&= -\Phi - r\Phi_r + \csc^2\theta - rV_r\Phi + (1/2)r\Phi_r - \cot^2\theta \\
&= 1-\Phi - (1/2)r\Phi_r - rV_r\Phi
\intertext{and}
\Ric_{\theta\phi} &= [\Gamma_{\theta\phi,\phi}^\phi] - [\Gamma_{\theta\phi,\phi}^\phi] + [] - [] = 0
\intertext{and}
\Ric_{\phi\phi} &= [\Gamma_{\phi\phi,r}^r + \Gamma_{\phi\phi,\theta}^\theta] - [] + [\Gamma_{\phi\phi}^r(\Gamma_{rt}^t+\Gamma_{rr}^r+\Gamma_{r\theta}^\theta+\Gamma_{r\phi}^\phi) + \Gamma_{\phi\phi}^\theta\Gamma_{\theta\phi}^\phi] \\
&\qquad - [\Gamma_{\phi r}^\phi\Gamma_{\phi\phi}^r + \Gamma_{\phi\theta}^\phi\Gamma_{\phi\phi}^\theta + \Gamma_{\phi\phi}^r\Gamma_{\phi r}^\phi + \Gamma_{\phi\phi}^\theta\Gamma_{\phi\theta}^\phi ] \\
&= \Gamma_{\phi\phi,r}^r + \Gamma_{\phi\phi,\theta}^\theta + \Gamma_{\phi\phi}^r\Gamma_{rt}^t + \Gamma_{\phi\phi}^r\Gamma_{rr}^r + \Gamma_{\phi\phi}^r\Gamma_{r\theta}^\theta - \Gamma_{\phi\phi}^r\Gamma_{\phi r}^\phi - \Gamma_{\phi\phi}^\theta\Gamma_{\phi\theta}^\phi \\
&= (-r\sin^2\theta\,\Phi)_r + (-\sin\theta\cos\theta)_\theta + (-r\sin^2\theta\,\Phi)V_r \\
&\qquad + (-r\sin^2\theta\,\Phi)(-(1/2)\Phi^{-1}\Phi_r) + (-r\sin^2\theta\,\Phi)r^{-1} - (-r\sin^2\theta\,\Phi)r^{-1} \\
&\qquad - (-\sin\theta\cos\theta)\cot\theta \\
&= -\sin^2\theta\,\Phi - r\sin^2\theta\,\Phi_r - \cos^2\theta + \sin^2\theta - rV_r\sin^2\theta\,\Phi \\
&\qquad + (1/2)r\sin^2\theta\,\Phi_r + \cos^2\theta \\
&= \sin^2\theta\,(1-\Phi - (1/2)r\Phi_r - rV_r\Phi).
\end{align*}
These calculations suffice because the Ricci curvature tensor is symmetric.
\end{proof}

\begin{lemma} \label{lem:scalarcurvature}
The scalar curvature associated with the metric \eqref{eqr:metric_Phi} in the $(t,r,\theta,\phi)$ coordinates is
\begin{align*}
R &= -2(V_{rr}+V_r^2+2r^{-1}V_r)\Phi - V_r\Phi_r + 2r^{-2}(1-\Phi-r\Phi_r).
\end{align*}
\end{lemma}
\begin{proof}
The scalar curvature is the trace of the Ricci curvature tensor. Using \cref{lem:riccicurvature} and the fact that the metric \eqref{eqr:metric_Phi} is diagonal,
\begin{align*}
R &= g^{tt}\Ric_{tt} + g^{rr}\Ric_{rr} + g^{\theta\theta}\Ric_{\theta\theta} + g^{\phi\phi}\Ric_{\phi\phi} \\
&= (-e^{-2V})[(V_{rr}+V_r^2+2r^{-1}V_r)e^{2V}\Phi + (1/2)V_re^{2V}\Phi_r] \\
&\qquad + \Phi[-(V_{rr}+V_r^2) - (1/2)V_r\Phi^{-1}\Phi_r - r^{-1}\Phi^{-1}\Phi_r] \\
&\qquad + r^{-2}[1-\Phi - (1/2)r\Phi_r - rV_r\Phi] \\
&\qquad + r^{-2}\sin^{-2}\theta\,[\sin^2\theta\,(1-\Phi - (1/2)r\Phi_r - rV_r\Phi)] \\
&= -(V_{rr}+V_r^2+2r^{-1}V_r)\Phi - (1/2)V_r\Phi_r - (V_{rr}+V_r^2)\Phi - (1/2)V_r\Phi_r - r^{-1}\Phi_r \\
&\qquad + 2r^{-2} - 2r^{-2}\Phi - r^{-1}\Phi_r - 2r^{-1}V_r\Phi \\
&= -2(V_{rr}+V_r^2+2r^{-1}V_r)\Phi - V_r\Phi_r + 2r^{-2}(1-\Phi-r\Phi_r).\qedhere
\end{align*}
\end{proof}

\begin{lemma} \label{lem:einsteincurvature}
The nonzero components of the Einstein curvature tensor $G$ associated with the metric \eqref{eqr:metric_Phi} in the $(t,r,\theta,\phi)$ coordinates are as follows:
\begin{align*}
G_{tt} &= r^{-2}e^{2V}(1-\Phi-r\Phi_r) \\
G_{rr} &= -r^{-2}\Phi^{-1}(1-\Phi-2rV_r\Phi) \\
G_{\theta\theta} &= r^2[(V_{rr}+V_r^2+r^{-1}V_r)\Phi + (1/2)r^{-1}\Phi_r +  (1/2)V_r\Phi_r] \\
G_{\phi\phi} &= r^2\sin^2\theta\,[(V_{rr}+V_r^2+r^{-1}V_r)\Phi + (1/2)r^{-1}\Phi_r +  (1/2)V_r\Phi_r].
\end{align*}
\end{lemma}
\begin{proof}
The Einstein curvature tensor is defined by \cref{eq:einsteincurvature}. (Remember that in this dissertation, $\Lambda=0$.) Using \cref{lem:riccicurvature,lem:scalarcurvature}, we have
\begin{align*}
G_{tt} &= \Ric_{tt} - (1/2)Rg_{tt} \\
&= [(V_{rr}+V_r^2+2r^{-1}V_r)e^{2V}\Phi + (1/2)V_re^{2V}\Phi_r] \\
&\qquad - (1/2)[-2(V_{rr}+V_r^2+2r^{-1}V_r)\Phi - V_r\Phi_r + 2r^{-2}(1-\Phi-r\Phi_r)](-e^{2V}) \\
&= r^{-2}e^{2V}(1-\Phi-r\Phi_r)
\intertext{and}
G_{rr} &= \Ric_{rr} - (1/2)Rg_{rr} \\
&= [-(V_{rr}+V_r^2) - (1/2)V_r\Phi^{-1}\Phi_r - r^{-1}\Phi^{-1}\Phi_r] \\
&\qquad - (1/2)[-2(V_{rr}+V_r^2+2r^{-1}V_r)\Phi - V_r\Phi_r + 2r^{-2}(1-\Phi-r\Phi_r)]\Phi^{-1} \\
&= -r^{-2}\Phi^{-1}(1-\Phi-2rV_r\Phi)
\intertext{and}
G_{\theta\theta} &= \Ric_{\theta\theta} - (1/2)Rg_{\theta\theta} \\
&= [1-\Phi - (1/2)r\Phi_r - rV_r\Phi] \\
&\qquad - (1/2)[-2(V_{rr}+V_r^2+2r^{-1}V_r)\Phi - V_r\Phi_r + 2r^{-2}(1-\Phi-r\Phi_r)]r^2 \\
&= r^2[(V_{rr}+V_r^2+r^{-1}V_r)\Phi + (1/2)r^{-1}\Phi_r +  (1/2)V_r\Phi_r]
\intertext{and}
G_{\phi\phi} &= \Ric_{\phi\phi} - (1/2)Rg_{\phi\phi} \\
&= [\sin^2\theta\,(1-\Phi - (1/2)r\Phi_r - rV_r\Phi)] \\
&\qquad - (1/2)[-2(V_{rr}+V_r^2+2r^{-1}V_r)\Phi - V_r\Phi_r + 2r^{-2}(1-\Phi-r\Phi_r)]r^2\sin^2\theta \\
&= r^2\sin^2\theta\,[(V_{rr}+V_r^2+r^{-1}V_r)\Phi + (1/2)r^{-1}\Phi_r +  (1/2)V_r\Phi_r].
\end{align*}
All other components are zero by \cref{lem:riccicurvature} and the the fact that the metric \cref{eqr:metric_Phi} is diagonal.
\end{proof}

Einstein's equation is $G=8\pi T$. For wave dark matter the energy-momentum tensor $T$ is (refer to \cref{eqrr:ekg})
\begin{equation} \label{eq:energymomentumtensor}
T = \frac{df\tensor d\bar{f}+d\bar{f}\tensor df}{\Upsilon^2} - \left(\frac{\abs{df}^2}{\Upsilon^2}+\abs{f}^2\right)g.
\end{equation}
To solve the Einstein equation we need to equate the components of the Einstein curvature tensor with the components of the energy-momentum tensor, e.g. $G_{tt}=8\pi T_{tt}$, $G_{rr}=8\pi T_{rr}$, etc. To do this we need to write down the components of $T$. We are solving the Einstein-Klein-Gordon equations in spherical symmetry, so $f$ is a function of $t$ and $r$ only.
\begin{lemma} \label{lem:energymomentumtensor}
With $f=f(t,r)$, the nonzero components of the energy-momentum tensor \eqref{eq:energymomentumtensor} in the $(t,r,\theta,\phi)$ coordinates are as follows:
\begin{align*}
T_{tt} &= e^{2V}\abs{f}^2 + \Upsilon^{-2}\abs{f_t}^2 + \Upsilon^{-2}e^{2V}\Phi\abs{f_r}^2 \\
T_{tr} = T_{rt} &= \Upsilon^{-2}(f_t\bar{f}_r + \bar{f}_tf_r) \\
T_{rr} &= -\Phi^{-1}\abs{f}^2 + \Upsilon^{-2}e^{-2V}\Phi^{-1}\abs{f_t}^2 + \Upsilon^{-2}\abs{f_r}^2 \\
T_{\theta\theta} &= \Upsilon^{-2}r^2(-\Upsilon^2\abs{f}^2 + e^{-2V}\abs{f_t}^2 - \Phi\abs{f_r}^2) \\
T_{\phi\phi} &= \Upsilon^{-2}r^2\sin^2\theta\,(-\Upsilon^2\abs{f}^2 + e^{-2V}\abs{f_t}^2 - \Phi\abs{f_r}^2).
\end{align*}
\end{lemma}
\begin{proof}
First, we have
\begin{equation*}
\abs{df}^2 = g^{\lambda\mu}f_\lambda\bar{f}_{\mu} = -e^{-2V}\abs{f_t}^2 + \Phi\abs{f_r}^2.
\end{equation*}
Thus
\begin{align*}
\Upsilon^2T_{tt} &= f_t\bar{f}_t + \bar{f}_tf_t - (-e^{-2V}\abs{f_t}^2 + \Phi\abs{f_r}^2 + \Upsilon^2\abs{f}^2)(-e^{2V}) \\
&= \Upsilon^2e^{2V}\abs{f}^2 + \abs{f_t}^2 + e^{2V}\Phi\abs{f_r}^2
\intertext{and}
\Upsilon^2 T_{tr} = \Upsilon^2T_{rt} &= f_t\bar{f}_r + \bar{f}_tf_r
\intertext{and}
\Upsilon^2 T_{rr} &= f_r\bar{f}_r + \bar{f}_rf_r - (-e^{-2V}\abs{f_t}^2 + \Phi\abs{f_r}^2 + \Upsilon^2\abs{f}^2)\Phi^{-1} \\
&= -\Upsilon^2\Phi^{-1}\abs{f}^2 + e^{-2V}\Phi^{-1}\abs{f_t}^2 + \abs{f_r}^2
\intertext{and}
\Upsilon^2 T_{\theta\theta} &= -(-e^{-2V}\abs{f_t}^2 + \Phi\abs{f_r}^2 + \Upsilon^2\abs{f}^2)r^2 \\
&= r^2(-\Upsilon^2\abs{f}^2 + e^{-2V}\abs{f_t}^2 - \Phi\abs{f_r}^2)
\intertext{and}
\Upsilon^2 T_{\phi\phi} &= -(-e^{-2V}\abs{f_t}^2 + \Phi\abs{f_r}^2 + \Upsilon^2\abs{f}^2)r^2\sin^2\theta \\
&= r^2\sin^2\theta\,(-\Upsilon^2\abs{f}^2 + e^{-2V}\abs{f_t}^2 - \Phi\abs{f_r}^2).
\end{align*}
The other components of $T$ are zero because $f$ depends only on $t$ and $r$ and the metric \eqref{eqr:metric_Phi} is diagonal.
\end{proof}

\begin{theorem} \label{lem:einsteinequation}
In the spacetime with the static metric \eqref{eqr:metric_Phi} and energy-momentum tensor \eqref{eq:energymomentumtensor}, where $f=f(t,r)$ is complex-valued, Einstein's equation \eqref{eqrr:ekg1} reduces to the following four PDEs:
\begin{align}
1-\Phi-r\Phi_r &= 8\pi r^2[\abs{f}^2 + \Upsilon^{-2}e^{-2V}\abs{f_t}^2 + \Upsilon^{-2}\Phi\abs{f_r}^2] \label{eq:einsteinpde_tt}\\
0 &= f_t\bar{f}_r + \bar{f}_tf_r \label{eq:einsteinpde_tr}\\
1-\Phi-2rV_r\Phi &= 8\pi r^2[\abs{f}^2 - \Upsilon^{-2}e^{-2V}\abs{f_t}^2 - \Upsilon^{-2}\Phi\abs{f_r}^2] \label{eq:einsteinpde_rr}\\
\begin{split}
\Upsilon^2r^2[(V_{rr}+V_r^2+r^{-1}V_r)\Phi \\
+ (1/2)r^{-1}\Phi_r +  (1/2)V_r\Phi_r] &= 8\pi r^2(-\Upsilon^2\abs{f}^2 + e^{-2V}\abs{f_t}^2 - \Phi\abs{f_r}^2).
\end{split} \label{eq:einsteinpde_thetatheta}
\end{align}
\end{theorem}
\begin{proof}
Using the results of \cref{lem:einsteincurvature,lem:energymomentumtensor}, set $G=8\pi T$ and write down the resulting PDEs. Note that equations coming from equating the $\theta\theta$ and $\phi\phi$ components are identical.
\end{proof}

In the Einstein-Klein-Gordon system, the Einstein equation \eqref{eqrr:ekg1} is coupled to the Klein-Gordon equation \eqref{eqrr:ekg2}. The following lemma is the partner to \cref{lem:einsteinequation}.
\begin{theorem} \label{lem:kleingordonequation}
In the spacetime with the static metric \eqref{eqr:metric_Phi}, where $f=f(t,r)$ is complex-valued, the Klein-Gordon equation \eqref{eqrr:ekg2} reduces to the following PDE:
\begin{equation} \label{eq:kleingordonpde}
-e^{-2V}f_{tt} + V_r\Phi f_r + (1/2)\Phi_r f_r + 2r^{-1}\Phi f_r + \Phi f_{rr} = \Upsilon^2f.
\end{equation}
\end{theorem}
\begin{proof}
We use the well-known coordinate expression for the d'Alembertian
\begin{equation*}
\Box f = \abs{g}^{-1/2}\d_\lambda(\abs{g}^{1/2}g^{\lambda\mu}\d_\nu f)
\end{equation*}
where $\abs{g}$ denotes the absolute value of the determinant of the matrix representing $g$. For the metric \eqref{eqr:metric_Phi}, $\abs{g}=e^{2V}\Phi^{-1}r^4\sin^2\theta$. Thus
\begin{align*}
\Box f &= (e^{-V}\Phi^{1/2}r^{-2}\sin^{-1}\theta)[(e^V\Phi^{-1/2}r^2\sin\theta\cdot -e^{-2V}f_t)_t + (e^V\Phi^{-1/2}r^2\sin\theta\cdot \Phi f_r)_r] \\
&= (e^{-V}\Phi^{1/2}r^{-2}\sin^{-1}\theta)[-e^{-V}\Phi^{-1/2}r^2\sin\theta\,f_{tt} + V_re^V\Phi^{1/2}r^2\sin\theta\,f_r \\
&\qquad + (1/2)e^V\Phi^{-1/2}\Phi_r r^2\sin\theta\,f_r + 2e^V\Phi^{1/2}r\sin\theta\,f_r + e^V\Phi^{1/2}r^2\sin\theta\,f_{rr}] \\
&= -e^{-2V}f_{tt} + V_r\Phi f_r + (1/2)\Phi_r f_r + 2r^{-1}\Phi f_r + \Phi f_{rr}.
\end{align*}
The result follows.
\end{proof}

We have shown that to solve the Einstein-Klein-Gordon system in the metric \eqref{eqr:metric_Phi} with $f=f(t,r)$, it suffices to solve \cref{eq:einsteinpde_tt,eq:einsteinpde_tr,eq:einsteinpde_rr,eq:einsteinpde_thetatheta,eq:kleingordonpde}. This system of equations is overdetermined. We have the following theorem:
\begin{theorem} \label{thm:ekgpdes}
In the spacetime with the static metric \eqref{eqr:metric_Phi}, where $f=f(t,r)$ is complex-valued, to solve the Einstein-Klein-Gordon system \eqref{eqrr:ekg} it suffices to solve \cref{eq:einsteinpde_tt,eq:einsteinpde_rr,eq:kleingordonpde}, which we repeat here:
\begin{gather}
1-\Phi-r\Phi_r = 8\pi r^2[\abs{f}^2 + \Upsilon^{-2}e^{-2V}\abs{f_t}^2 + \Upsilon^{-2}\Phi\abs{f_r}^2] \tag{\ref{eq:einsteinpde_tt}}\\
1-\Phi-2rV_r\Phi = 8\pi r^2[\abs{f}^2 - \Upsilon^{-2}e^{-2V}\abs{f_t}^2 - \Upsilon^{-2}\Phi\abs{f_r}^2] \tag{\ref{eq:einsteinpde_rr}}\\
-e^{-2V}f_{tt} + V_r\Phi f_r + (1/2)\Phi_r f_r + 2r^{-1}\Phi f_r + \Phi f_{rr} = \Upsilon^2f. \tag{\ref{eq:kleingordonpde}}
\end{gather}
\end{theorem}
\begin{proof}
Assuming \cref{eq:einsteinpde_tt,eq:einsteinpde_rr,eq:kleingordonpde} hold, we show \cref{eq:einsteinpde_tr} holds. By adding and subtracting \cref{eq:einsteinpde_tt,eq:einsteinpde_rr}, we get
\begin{align}
1-\Phi - (1/2)r\Phi_r - rV_r\Phi &= 8\pi r^2\abs{f}^2 \label{eq:add} \\
-(1/2)r\Phi_r + rV_r\Phi &= 8\pi r^2\Upsilon^{-2}[e^{-2V}\abs{f_t}^2+\Phi\abs{f_r}^2]. \label{eq:subtract}
\end{align}
Differentiating \cref{eq:add,eq:subtract} with respect to $t$, we get
\begin{align}
0 &= (\abs{f}^2)_t = f_t\cc{f} + \cc{f_t}f \label{eq:tderiv1}\\
0 &= [e^{-2V}\abs{f_t}^2+\Phi\abs{f_r}^2]_t = e^{-2V}(f_{tt}\cc{f_t}+\cc{f_{tt}}f_t) + \Phi(f_{tr}\cc{f_r}+\cc{f_{tr}}f_r). \label{eq:tderiv2}
\end{align}
Using \cref{eq:kleingordonpde,eq:tderiv1,eq:tderiv2}, we have
\begin{align*}
\Phi(f_t\cc{f_r}+\cc{f_t}f_r)_r &= \Phi(f_{tr}\cc{f_r}+\cc{f_{tr}}f_r) + \Phi(f_t\cc{f_{rr}}+\cc{f_t}f_{rr}) \\
&= -e^{-2V}(f_{tt}\cc{f_t}+\cc{f_{tt}}f_t) + \Phi(f_t\cc{f_{rr}}+\cc{f_t}f_{rr}) \\
&= (-e^{-2V}f_{tt}+\Phi f_{rr})\cc{f_t} + (\cc{-e^{-2V}f_{tt}+\Phi f_{rr}})f_t \\
&= \Upsilon^2(f\cc{f_t}+\cc{f}f_t) - (V_r\Phi+(1/2)\Phi_r+2r^{-1}\Phi)(f_r\cc{f_t}+\cc{f_r}f_t) \\
&= -(V_r\Phi+(1/2)\Phi_r+2r^{-1}\Phi)(f_r\cc{f_t}+\cc{f_r}f_t).
\end{align*}
Thus
\begin{equation*}
(f_t\cc{f_r}+\cc{f_t}f_r)_r = -(V_r+2r^{-1}+(1/2)\Phi^{-1}\Phi_r)(f_r\cc{f_t}+\cc{f_r}f_t).
\end{equation*}
Using the equation just obtained, we then have
\begin{align*}
[r^2e^V\Phi^{1/2}(f_t\cc{f_r}+\cc{f_t}f_r)]_r &= 2re^V\Phi^{1/2}(f_t\cc{f_r}+\cc{f_t}f_r) + r^2V_re^V\Phi^{1/2}(f_t\cc{f_r}+\cc{f_t}f_r) \\
&\qquad + (1/2)r^2e^V\Phi^{-1/2}\Phi_r(f_t\cc{f_r}+\cc{f_t}f_r) \\
&\qquad - r^2e^V\Phi^{1/2}(V_r+2r^{-1}+(1/2)\Phi^{-1}\Phi_r)(f_r\cc{f_t}+\cc{f_r}f_t) \\
&= 0.
\end{align*}
Therefore for each fixed $t$ the function $r^2e^V\Phi^{1/2}(f_t\cc{f_r}+\cc{f_t}f_r)$ is a constant. Taking the limit as $r\to 0$ shows this constant must be zero. Since $r^2e^V\Phi^{1/2}$ is nonzero for all $r>0$, we conclude that $f_t\cc{f_r}+\cc{f_t}f_r=0$, which is \cref{eq:einsteinpde_tr}.

Now assuming \cref{eq:einsteinpde_tt,eq:einsteinpde_rr,eq:kleingordonpde,eq:einsteinpde_tr} hold, we show \cref{eq:einsteinpde_thetatheta} holds. Differentiating \cref{eq:einsteinpde_rr} with respect to $r$, we obtain
\begin{multline*}
-\Phi_r - 2V_r\Phi - 2rV_{rr}\Phi - 2rV_r\Phi_r = \\
16\pi r[\abs{f}^2-\Upsilon^{-2}e^{-2V}\abs{f_t}^2-\Upsilon^{-2}\Phi\abs{f_r}^2] + 8\pi r^2[f_r\cc{f}+\cc{f_r}f+2\Upsilon^{-2}V_re^{-2V}\abs{f_t}^2 \\
\qquad -\Upsilon^{-2}e^{-2V}(f_{tr}\cc{f_t}+\cc{f_{tr}}f_t)-\Upsilon^{-2}\Phi_r\abs{f_r}^2-\Upsilon^{-2}\Phi(f_{rr}\cc{f_r}+\cc{f_{rr}}f_r)].
\end{multline*}
Multiplying through by $-(1/2)\Upsilon^2r$ gives
\begin{multline} \label{eq:temp1}
\Upsilon^2r^2[V_{rr}\Phi + r^{-1}V_r\Phi + V_r\Phi_r + (1/2)r^{-1}\Phi_r] = \\
8\pi r^2[-\Upsilon^2\abs{f}^2+e^{-2V}\abs{f_t}^2+\Phi\abs{f_r}^2] - 4\pi r^3[\Upsilon^2(f_r\cc{f}+\cc{f_r}f)+2V_re^{-2V}\abs{f_t}^2 \\
\qquad -e^{-2V}(f_{tr}\cc{f_t}+\cc{f_{tr}}f_t)-\Phi_r\abs{f_r}^2-\Phi(f_{rr}\cc{f_r}+\cc{f_{rr}}f_r)].
\end{multline}
Multiplying \cref{eq:subtract} by $\Upsilon^2rV_r$ gives
\begin{equation} \label{eq:temp2}
\Upsilon^2r^2(V_r^2\Phi-(1/2)V_r\Phi_r) = 4\pi r^3[2V_re^{-2V}\abs{f_t}^2 + 2V_r\Phi\abs{f_r}^2]
\end{equation}
Adding \cref{eq:temp1,eq:temp2}, we get
\begin{multline*}
\Upsilon^2r^2[(V_{rr}+V_r^2+r^{-1}V_r)\Phi + (1/2)r^{-1}\Phi_r + (1/2)V_r\Phi_r] = \\
8\pi r^2[-\Upsilon^2\abs{f}^2+e^{-2V}\abs{f_t}^2-\Phi\abs{f_r}^2] - 4\pi r^3[\Upsilon^2(f_r\cc{f}+\cc{f_r}f)-e^{-2V}(f_{tr}\cc{f_t}+\cc{f_{tr}}f_t) \\
- \Phi_r\abs{f_r}^2 - 2V_r\Phi\abs{f_r}^2 - \Phi(f_{rr}\cc{f_r}+\cc{f_{rr}}f_r) - 4r^{-1}\Phi\abs{f_r}^2].
\end{multline*}
Comparing with \cref{eq:einsteinpde_thetatheta}, we see that it now suffices to show that the last expression in square brackets above is zero. We can rewrite this expression as
\begin{multline}
(\Upsilon^2 f - (1/2)\Phi_rf_r - V_r\Phi f_r - \Phi f_{rr} - 2r^{-1}\Phi f_r)\cc{f_r} \\
+ (\cc{\Upsilon^2 f - (1/2)\Phi_rf_r - V_r\Phi f_r - \Phi f_{rr} - 2r^{-1}\Phi f_r})f_r \\
- e^{-2V}(f_{tr}\cc{f_t}+\cc{f_{tr}}f_t).
\end{multline}
Using \cref{eq:kleingordonpde}, this expression becomes
\begin{equation*}
-e^{-2V}(f_{tt}\cc{f_r}+\cc{f_{tt}}f_r) - e^{-2V}(f_{tr}\cc{f_t}+\cc{f_{tr}}f_t) = -e^{-2V}(f_t\cc{f_r}+\cc{f_t}f_r)_t.
\end{equation*}
This last quantity is zero as desired by \cref{eq:einsteinpde_tr}. This completes the proof.
\end{proof}

The goal of this appendix is the following theorem:
\begin{theorem} \label{thm:ekgodes}
In the spacetime with the static metric \eqref{eqr:metric}, where $f(t,r)=F(r)e^{i\omega t}$ with $F$ real-valued and $\omega$ real, to solve the Einstein-Klein-Gordon system \eqref{eqrr:ekg} it suffices to solve the following three coupled ODEs:
\begin{gather}
M_r = 4\pi r^2\cdot\frac{1}{\Upsilon^2}\left[ \left(\Upsilon^2+\omega^2e^{-2V}\right)F^2 + \Phi F_r^2\right] \\
\Phi V_r = \frac{M}{r^2} - 4\pi r\cdot\frac{1}{\Upsilon^2}\left[ \left(\Upsilon^2-\omega^2 e^{-2V}\right)F^2 - \Phi F_r^2 \right] \\
F_{rr} + \frac{2}{r} F_r + V_rF_r + \frac12\frac{\Phi_r}{\Phi}F_r = \Phi^{-1}\left(\Upsilon^2-\omega^2e^{-2V}\right)F.
\end{gather}
\end{theorem}
\begin{proof}
These three equations are rearranged versions of \cref{eq:einsteinpde_tt,eq:einsteinpde_rr,eq:kleingordonpde} from \cref{thm:ekgpdes}, with the particular form for $f$ substituted in and with \cref{eqr:Phi} used in places so that the function $M(r)$ appears.
\end{proof}

\cleardoublepage
\normalbaselines 
\addcontentsline{toc}{chapter}{Bibliography} 
\printbibliography

\biography

Andrew Stewart Goetz was born in South Bend, Indiana on September 3, 1987. From 2005--2009 he majored in mathematics at Princeton University and graduated with an A.B. degree. From 2009--2015 he studied mathematics at Duke University and graduated with a Ph.D.

\end{document}